\documentclass[a4paper,onecolumn,11pt,unpublished]{quantumarticle}
\pdfoutput=1

\usepackage[T1]{fontenc}
\usepackage[utf8]{inputenc}
\usepackage[numbers,sort&compress]{natbib}
\usepackage{hyperref}

\usepackage[table]{xcolor}
\usepackage{colortbl}

\usepackage{booktabs}
\usepackage{enumerate}
\usepackage{enumitem}
\usepackage{caption}
\usepackage{adjustbox}
\usepackage{longtable}
\usepackage{tikz}
\usepackage{braket}
\usetikzlibrary{arrows.meta,positioning,calc,fit,backgrounds}
\usepackage{pgfplots}
\pgfplotsset{compat=1.17}
\usepackage[numbers,sort&compress]{natbib}
\usepackage{hyperref}

\usepackage{pgfplots}
\pgfplotsset{compat=1.17}
\usepgfplotslibrary{groupplots}
\usepackage[table]{xcolor}
\usepackage{listings}
\usepackage{siunitx}
\usepackage{array}
\usepackage{multicol}
\usepackage{pifont}  
\usepackage{newunicodechar}
\newunicodechar{✖}{\ding{55}}
\usepackage[margin=1in]{geometry}   
\usepackage{setspace}              
\usepackage{bm}
\usepackage{amsmath,amssymb,amsthm}             
\usepackage{tikz}            
\usepackage{quantikz} 
\usepackage{enumitem}
\usepackage{graphicx} 
\usepackage{algpseudocode}
\usepackage{algorithm}
\usepackage{algorithmicx}
\usepackage[most]{tcolorbox}
\usepackage{pgfplots}
\pgfplotsset{compat=1.17}
\newcommand{\poly}[1]{\mathrm{poly}\bigl(#1\bigr)}

\usepackage{booktabs}
\usepackage[nameinlink,noabbrev]{cleveref}
\usetikzlibrary{
    arrows.meta,
    positioning,
    calc,
    fit
}
\usepackage{xcolor}                
\usepackage{lscape}                
\usepackage{caption}
\usepackage{subcaption}
\AtBeginEnvironment{abstract}{\raggedright}

\newtheorem{theorem}{Theorem}
\newtheorem{lemma}[theorem]{Lemma}
\newtheorem{proposition}[theorem]{Proposition}
\newtheorem{corollary}[theorem]{Corollary}
\newtheorem{definition}[theorem]{Definition}

\theoremstyle{remark}

\newcommand{\OH}{\mathcal{H}_{\mathrm{OH}}}
\DeclareUnicodeCharacter{0301}{}
\newcommand{\spec}{\mathrm{spec}}

\usepackage[normalem]{ulem} 

\makeatletter
\g@addto@macro\normalsize{%
    \abovedisplayskip 3pt plus 1pt minus 1pt%
    \abovedisplayshortskip 3pt plus 1pt minus 1pt%
    \belowdisplayskip 3pt plus 1pt minus 1pt%
    \belowdisplayshortskip 3pt plus 1pt minus 1pt%
}
\makeatother

\def\BibTeX{{\rm B\kern-.05em{\sc i\kern-.025em b}\kern-.08em
    T\kern-.1667em\lower.7ex\hbox{E}\kern-.125emX}}

\author{Chinonso Onah}
\affiliation{Volkswagen AG, Berliner Ring 2, 38440 Wolfsburg, Germany}
\affiliation{Department of Physics, RWTH Aachen University, 52056 Aachen, Germany}

\author{Stuart Hadfield}
\affiliation{USRA Research Institute for Advanced Computer Science (RIACS), CA, USA}

\author{Kristel Michielsen}
\affiliation{Forschungszentrum Jülich, Germany}
\affiliation{Universit\"at zu K\"oln, 50923 K\"oln, Germany}

\title{Separating Geometry From Interference in Constrained Quantum Optimization}

\begin{document}


\maketitle
\begin{abstract}

We study the separation of geometric effects from quantum interference in quantum optimization algorithms. Constrained optimization problems such as routing, assignment, matching, and scheduling are often encoded as product spaces of local variables, together with global feasibility penalties. The central algorithmic question we address is how a constraint-preserving mixing operator transports quantum amplitude across an exponential search space in the presence of local and global constraints. To that end, we develop a framework that separates three effects that are usually intermixed: transport geometry on the search space, coherent interference among transported amplitudes, and problem-dependent classical postprocessing.

We show that the mixing operator alone does not have a target-seeking ability. Concretely, the normalized distribution induced by its amplitude transport moves toward the typical bulk of the search space, corresponding to the distance profile of a uniformly random configuration, rather than concentrating toward a target configuration. Thus, quantum sampling advantage may only arise when the phases of the many computational paths reaching a target configuration are sufficiently aligned for their amplitudes to reinforce. We show that, when the cost phases are engineered so that these paths add coherently, a number of circuit alternations growing only logarithmically with problem size suffices to convert the sum of their absolute contributions into a lower bound on the target amplitude, yielding a certified success probability independent of the ambient Hilbert-space dimension, the search-space size, or the feasible-set cardinality.

We develop applications to problem-specific transpilation diagnostics, scalable hardware probes, constraint-induced classical maps of quantum-generated samples, and the attribution of solution quality between the quantum distribution and classical post-processing in hybrid quantum-classical workflows. Although motivated by constrained quantum optimization, the mechanism applies whenever the one-step transfer kernel factorizes across coordinates or blocks. The resulting formalism connects constrained quantum optimization directly to lumped Markov chains, Krawtchouk polynomials, and the association-scheme viewpoint for distance-partitioned product spaces from classical coding theory.
\end{abstract}

\section{Introduction}
\label{subsec:motivation-radialization}

Constrained  optimization is characterized by problems whose difficulty comes from both objective optimization and combinatorial constraint structure.
Routing, assignment, matching, scheduling, packing, and covering are classical examples. In their natural forms, these problems are defined over discrete configuration spaces restricted by combinatorial constraints~~\cite{KorteVygen,Schrijver2003}. In quantum
optimization~\cite{abbas2024challenges,Lucas2014,Farhi2014QAOA,Lucas2014,Hadfield2019AOA,onahce}, the same structure appears through an encoded search space. One first identifies a logical configuration space of admissible local choices, and then represents that space in qubits or qudits through an encoding. The central question is then how the quantum dynamics transports amplitude across this structured search space \cite{Farhi2014QAOA,Hadfield2019AOA,Wang2020XYMixers,Hadfield_2022,Fuchs2022ConstrainedMixers,sawaya2023encoding,onahfund,onahce}.

A simple example is a product space of \(m\) local registers,
\begin{equation}
\label{eq:product-space}
\mathcal X=[n]^m .
\end{equation}
Here \(n\) is the number of local choices available in each register, and \(m\) is the number of registers. In a qudit implementation, one may represent each local choice directly by a qudit basis state. In a qubit implementation, the same logical alphabet can be represented, for instance, by a one-hot block or more generally by a fixed-Hamming-weight \(k\)-hot block~\cite{Hadfield2019AOA}. In the one-hot case, the encoded Hilbert space is \[ \OH=(\mathcal H_1)^{\otimes m}, \qquad \mathcal H_1 = \operatorname{span}\{\ket{e_1},\dots,\ket{e_n}\}. \] We identify the computational basis states of \(\OH\) with words \(x=(x_1,\dots,x_m)\in[n]^m\), where \[ \ket{x} = \ket{e_{x_1}}\otimes\cdots\otimes\ket{e_{x_m}}. \] Thus \(\mathcal \mathcal X\) denotes the classical product search space, while \(\OH\) denotes its Hilbert-space realization. This product form appears naturally in routing, scheduling, assignment, coloring, matching, and allocation problems \cite{Hadfield2019AOA,onahce}. 

In routing, the blocks may represent tour positions and the symbols may represent customers or cities. In scheduling, the blocks may represent jobs, machines, or time slots. In assignment and allocation, the blocks may represent decision locations and the symbols may represent resources or selected items. Global constraints then appear as penalties, feasibility projectors, or decoding rules on the product space. Classical formulations of these problems often expose the same structure through assignment tables, fixed-margin constraints, contingency-table fibers, or transportation representations
\cite{MillerTuckerZemlin1960,DiaconisSturmfels1998,DeLoeraOnn2004,
DeLoeraOnn2006,SlavkovicZhuPetrovic2015}.

Accordingly, we consider constrained optimization problems in which fixed-cardinality constraints are encoded directly into local blocks, while the remaining constraints impose global relations among those blocks. Quantum algorithms exploiting this problem structure can be realized on gate-based universal quantum processors through the 
quantum alternating operator ansatz (QAOA) framework~\cite{Hadfield2019AOA}, and in quantum annealing through the cardinality-constrained driver Hamiltonians introduced in
Refs.~\cite{HenSpedalieri2016,HenSarandy2016}. In the gate-based setting, the mixer unitary redistributes amplitude across the encoded search space; in quantum annealing, the driver Hamiltonian generates the corresponding constraint-preserving evolution. In both settings, the mixer or driver must preserve the encoded subspace and provide transport across regions of the search space relevant to feasible or high-quality configurations.

The shell formalism developed below provides a compact description of this transport on
the resulting encoded product space.  The basic object used throughout our analysis is the generalized Hamming distance which we define presently.  Fix a target configuration \(y\in \mathcal X\).  The Hamming distance from
\(y\) is
\begin{equation}
\label{eq:hamming-distance}
d(x,y)
=
\bigl|\{b\in\{1,\dots,m\}:x_b\neq y_b\}\bigr|.
\end{equation}
The resulting Hamming shells are
\[
S_r(y)=\{x\in \mathcal X:\ d(x,y)=r\},
\qquad r=0,\dots,m.
\]
Thus \(S_0(y)=\{y\}\), while \(S_r(y)\) contains the configurations that
differ from \(y\) in exactly \(r\) local registers. The corresponding unitary dynamics acts on the Hilbert space \(\mathcal H_{\mathrm{OH}}\) while the product set \(\mathcal X=[n]^m\) labels the computational basis of this space.  After this basis identification, the mixer matrix elements
\[
\langle z|U_M(\beta)|x\rangle,
\qquad x,z\in \mathcal X,
\]
define a basis-indexed amplitude-transfer kernel on \( \mathcal X\).  The shell reduction we propose pushes this transfer kernel through the distance map
\[
R_y:\mathcal X\to\{0,1,\dots,m\},
\qquad
R_y(x)=d(x,y).
\]
As a result, the coherent Hilbert-space transport is represented by a shell-resolved
transfer law on the quotient variable \(R_y\). This converts a high-dimensional transport problem into a shell-resolved one. Instead of tracking transfer on the Hilbert space \(\mathcal H_{\mathrm{OH}}\), which is combinatorially expensive, one tracks how mass moves between distance shells
around the target configuration \(y\). This gives a
reduced language for studying constrained mixers through radial drift, shell-hitting behavior, transfer concentration, dependence on depth and mixer angle, and compatibility with later phase-coherent layers. In end-to-end
hybrid quantum--classical workflows, the same shell description tracks how much amplitude is transported across neighborhoods of a target configuration.

A key observation is that the shell reduction obtained by pushing the mixer
transfer kernel through the distance map \(R_y\) isolates the phase-blind
geometric transport problem from the subsequent phase-coherent amplitude
analysis. This reflects the layered structure of the quantum approximate
optimization algorithm, the quantum alternating operator ansatz, and related
variants~\cite{Farhi2014QAOA,Hadfield2019AOA}. It first computes the mixer-induced transfer structure on the encoded manifold without requiring a phase-coherent analysis of the full path sum. Phase control then enters as a separate question of whether the path phases generated by the diagonal phase operators and mixer matrix elements align sufficiently for the geometric shell-transfer mass to survive in the target amplitude. Under the lattice-normalized phase-control regime developed in Section~\ref{subsec:shell-transfer-to-amplitude}, this separation converts radial transfer mass into a certified amplitude lower bound, and hence into a success-probability guarantee whose scale is independent of the ambient Hilbert-space dimension, the encoded-manifold size, or the feasible-set cardinality.

We analyze this mechanism in the setting of the recently introduced Constraint-Enhanced Quantum Approximate Optimization Algorithm (CE--QAOA)
\cite{onahce}. CE--QAOA provides the concrete constrained-optimization setting because it is built on a kernel with a product \(k\)-hot encoding, a compatible constraint-preserving mixer, and a diagonal cost structure. However, the resulting formalism applies more broadly whenever the configuration space has a persistent product structure and the mixer transfer factorizes across blocks. As a consequence, the resulting framework connects constrained quantum optimization to several applications where Hamming shells, product alphabets, distance-partitioned state spaces, lumped Markov chains, and the coding-theoretic language of association schemes \cite{MacWilliamsSloane1977,DelsarteLevenshtein1998} are relevant tools.

The viewpoint is consistent with earlier distance-resolved analyses of quantum walks on highly symmetric graphs, where the full dynamics collapses to a smaller invariant structure \cite{MooreRussell2002,Childs2004SpatialSearch,Childs2008WalkCorrespondence,KroviBrun2007Quotient}. Here the analogous reduction is developed for constrained quantum optimization on product spaces. The complete symmetric block-\(XY\) 
mixer case gives the clean exact model,
while other product-space mixer geometries suggest refined distance statistics, such as path distance, cyclic distance, edge classes, or association-scheme data.

\begin{figure}[t]
\centering
\resizebox{\textwidth}{!}{%
\begin{tikzpicture}[
    font=\small,
    main/.style={
        draw,
        rounded corners=2pt,
        minimum height=1.7cm,
        text width=2.95cm,
        align=center,
        inner sep=6pt,
        fill=gray!8
    },
    sub/.style={
        draw,
        rounded corners=2pt,
        minimum height=1.25cm,
        text width=3.15cm,
        align=center,
        inner sep=5pt,
        fill=gray!3,
        font=\footnotesize
    },
    flow/.style={
        -{Latex[length=3mm,width=2mm]},
        thick
    },
    link/.style={
        -{Latex[length=2.5mm,width=1.8mm]},
        dashed,
        thick
    }
]

\node[main] (constraints) {
    \textbf{Problem structure and constraints}\\[2pt]
    Local cardinality constraints\\
    Global relations among blocks
};

\node[main, right=0.5cm of constraints] (encoding) {
    \textbf{Encoded product space}\\[2pt]
    One-hot, \(k\)-hot, or related blocks\\
    \(\mathcal X=\prod_{b=1}^{m}X_b\)
};

\node[main, right=0.5cm of encoding] (transport) {
    \textbf{Mixer-induced transport}\\[2pt]
    Shell-transfer coefficients\\
    Absolute path mass\\
    Radial dynamics
};

\node[main, right=0.5cm of transport] (phase) {
    \textbf{Phase-sensitive interference}\\[2pt]
    Exact path expansion\\
    Phase alignment\\
    Target amplitude
};

\node[main, right=0.5cm of phase] (samples) {
    \textbf{Measured samples}\\[2pt]
    Born distribution on \(\mathcal \mathcal X\)\\
    Feasible and infeasible outputs
};

\node[main, right=0.5cm of samples] (classical) {
    \textbf{Problem-specific classical maps}\\[2pt]
    Hardware probes\\
    Feasibility repair\\
    Application-level interpretation
};

\draw[flow] (constraints) -- (encoding);
\draw[flow] (encoding) -- (transport);
\draw[flow] (transport) -- 
node[above, font=\scriptsize, align=center]
{}
(phase);
\draw[flow] (phase) -- (samples);
\draw[flow] (samples) -- (classical);

\node[sub, below=1.55cm of encoding, xshift=-1.9cm] (coding) {
    \textbf{Coding-theoretic tools}\\[2pt]
    Distance distributions\\
    Krawtchouk structure\\
    Weight enumerators\\
    Orthogonal polynomials and\\
    association schemes
};

\node[sub, below=1.25cm of transport] (mixers) {
    \textbf{Logical mixer diagnostics}\\[2pt]
    Complete, path, ring, star, and sparse
    \(XY\)-mixer geometries\\
    Transport range, bottlenecks, and reduced statistics
};

\node[sub, below=1.25cm of phase, xshift=0.75cm] (trotter) {
    \textbf{Transport focused mixer diagnostics}\\[2pt]
    Shell-kernel error\\
    Lumpability defect\\
    Transport-preserving logical depth\\
    Deviation from shell-transport guarantees
};

\node[sub, below=1.25cm of samples, xshift=0.9cm] (stats) {
    \textbf{Constraint statistics and fibres}\\[2pt]
    Pushforward statistics\\
    Penalty fibres\\
    Structured classical maps\\
    Contingency-tables
};

\node[sub, below=1.25cm of classical, xshift=0.79cm] (repair) {
    \textbf{Hybrid workflows}\\[2pt]
    Violation profiles and repairability\\
    Raw versus repaired success\\
    Attribution of classical post-processing gains
};

\draw[link] (encoding.south west) to[out=-120,in=20] (coding.north);
\draw[link] (transport.south) to[out=-115,in=0] (coding.north);

\draw[link] (transport.south) -- (mixers.north);
\draw[link] (phase.south) to[out=-90,in=90] (trotter.north);

\draw[link] (classical.south) to[out=-115,in=20] (stats.north);

\draw[link] (classical.south) -- (repair.north);

\node[
    draw,
    rounded corners=3pt,
    dashed,
    inner sep=8pt,
    fit=(constraints)(encoding)(transport),
    label={[font=\footnotesize]above:\textbf{Geometric transport layer}}
] {};

\node[
    draw,
    rounded corners=3pt,
    dashed,
    inner sep=8pt,
    fit=(phase),
    label={[font=\footnotesize]above:\textbf{Coherent layer}}
] {};

\node[
    draw,
    rounded corners=3pt,
    dashed,
    inner sep=8pt,
    fit=(samples)(classical),
    label={[font=\footnotesize]above:\textbf{Post-measurement layer}}
] {};

\end{tikzpicture}%
}
\caption{
Layered structure of the framework. Problem constraints determine the encoded
product space and the transport requirements imposed on the mixer. The mixer
generates shell-resolved absolute path mass, while the diagonal phase layers
determine whether this mass contributes constructively to the target
amplitude. Measurement produces a classical distribution on the product
space, after which problem-specific classical maps expose application
structure, provide hardware probes, and guide feasibility-restoring repair.
The lower boxes summarize the principal analytical and application
consequences developed in this work.
}
\label{fig:framework-synthesis}
\end{figure}
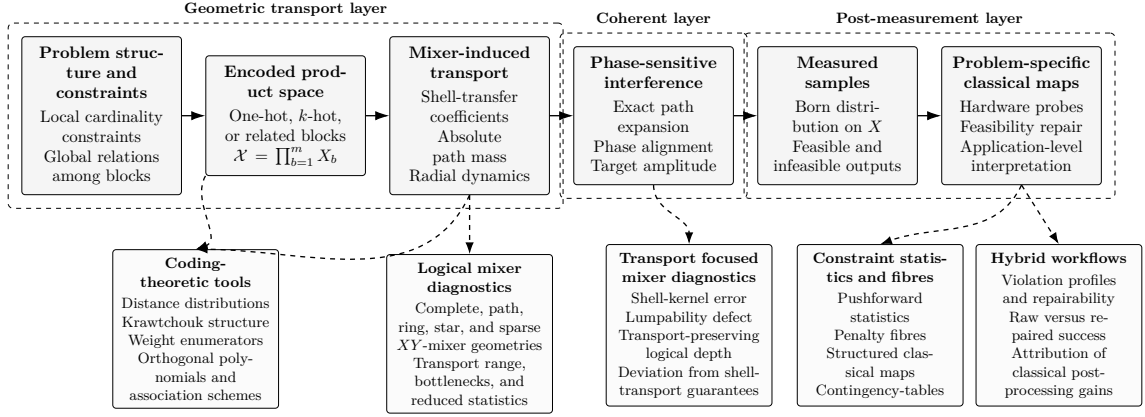

\subsection{Related Work}

Several analytical approaches have been developed to understand QAOA and related variational quantum algorithms beyond numerical optimization. One line of work studies the locality structure of shallow QAOA circuits and shows that low-depth QAOA only probes bounded neighborhoods, leading to graph-visibility limitations and worst-case indistinguishability results~\cite{FarhiGamarnikGutmannTypical,FarhiGamarnikGutmannWorstCase}. Related analyses evaluate QAOA performance on locally tree-like graph families through recursive formulas and local weak-limit methods, including large-girth regular MaxCut instances, the Sherrington--Kirkpatrick model, and sparse random hypergraph models~\cite{BassoFarhiMarwahaVillalongaZhou2022}. Other work proves concentration and limitation results for QAOA output states using polynomial approximation methods, including concentration of Hamming-weight-type observables and limitations under overlap-gap assumptions~\cite{AnshuMetger2023}. Complementary analytical and data-driven approaches study parameter concentration, fixed-angle performance, and variational landscape structure through gradient statistics and information content~\cite{AkshayRabinovichCamposBiamonte2021,WurtzLove2021,PerezSalinasWangBonetMonroig2024}.

Within the quantum alternating-operator ansatz paradigm, a closely related line of work derives exact expressions for expectation values of layered quantum optimization circuits in terms of algorithm parameters, cost-gradient operators, and cost-difference functions~\cite{Hadfield_2022}. Closest to the present setting are prior analytical CE--QAOA studies of
feasibility limitations and finite-shot success guarantees in globally constrained problems~\cite{onahfund,onahfinite}. The present work identifies a complementary mechanism at the level of the sampled distribution generated by constrained product-space dynamics and exhibits an explicit mechanism by which product-space transport and coherent phase control can raise the target-sampling probability from the \(n^{-m}\) scale of uniform sampling to a finite scale or inverse polynomial scale.

\subsection{Guide to main results}
On the product space \(\mathcal X=[n]^m\) defined in Eq.~\eqref{eq:product-space}, we use the Hamming distance from the target configuration \(y\) defined in Eq.~\eqref{eq:hamming-distance} to resolve the
mixer transport into distance shells. For the complete-graph block-\(XY\)
mixer, the one-block transition amplitude depends only on whether two local
symbols agree or differ as established in Lemma \ref{lem:single-block-xy-amplitudes}. Consequently, the total amplitude transfer
between two shells depends only on their shell indices. This gives exact
shell-transfer coefficients and a finite radial recursion, as established in
Proposition~\ref{prop:shell-transfer-generating} and
Theorem~\ref{thm:exact-shell-transfer-recursion}.

After normalization, the shell-transfer coefficients define a finite Markov process on the distance shells and provide a probabilistic description of the mixer dynamics. The formal results are collected in Appendix~\ref{app:measure-shell-reduction}, ranging from the finite Fubini principle in Lemma~\ref{lem:fubini-radial-shell} to the time-inhomogeneous Markov recursion in Theorem~\ref{thm:depth-p-radial-markov}. The resulting probabilistic interpretation shows that the mixer dynamics alone has no target-seeking bias, since the normalized shell process relaxes toward the ordinary product-space bulk rather than toward any target configuration. In this setting, quantum sampling advantage therefore requires a phase structure that coherently converts the available path mass into Born probability on useful configurations. 

To show how such a phase structure can be engineered, we introduce an ideal lattice normalization of the diagonal cost spectrum and impose the common-arc phase-alignment condition of Theorem~\ref{thm:phase-aligned-path-sum-lower-bound}. Under these conditions, the radial transfer mass yields a lower bound on the amplitude of the target configuration \(y\). Together with the constructive mixer-angle and depth conditions of Proposition~\ref{prop:constructive-radial-mass-lower-bound}, this gives a dimension free lower bound on the probability of sampling \(y\) despite the exponentially large product space \(|\mathcal X|=n^m\). The analysis therefore exhibits an explicit mechanism by which product-space transport and coherent phase control can raise the target-sampling probability from the \(n^{-m}\) scale of uniform sampling to a finite scale.

In the spirit of problem-algorithm codesign ~\cite{Li2021CoDesign}, this separation provides a concrete analytical design principle where the mixer should be chosen to generate sufficient transport mass across the relevant regions of the product space, while the phase operator should be designed to preserve and coherently harvest that mass on useful configurations. Separating these roles makes the source of a quantum sampling advantage explicit and allows each component of the algorithm to be analyzed and engineered independently for a broad class of constrained optimization problems. Figure~\ref{fig:framework-synthesis} summarizes the layered structure of the framework and shows how geometric transport, coherent interference, probabilistic interpretation, transpilation diagnostics, and problem-specific classical maps fit together within a single architecture.

\paragraph{Guide to the paper.} The rest of the paper is organized as follows. Section~\ref{subsec:subspace-constrained-dynamics} introduces \(XY\)-mixer geometries and encoded quantum product spaces. Section~\ref{sec:results} develops the shell-transfer recursion. Section~\ref{sec:main2} then separates the phase-sensitive part of the analysis from the shell algebra by showing how phase alignment converts radial transfer mass into amplitude lower bounds, with the probabilistic shell-kernel formalism collected in Appendix~\ref{app:measure-shell-reduction}. Section~\ref{sec:applications} develops the main consequences of the framework, including the normalized radial process, finite-size mixer diagnostics, shell-transport preservation under transpilation, and
problem-dependent classical maps for hardware probing and feasibility repair. Section~\ref{sec:discussion} discusses the broader product-space interpretation and summarizes the main conclusions and extensions.

\section{Constrained product-spaces in quantum optimization}
\label{subsec:subspace-constrained-dynamics}

\subsection{Block encodings and product configuration spaces}
\label{subsec:block-encodings-product spaces}

Many constrained optimization problems are naturally described by discrete decision variables associated with positions, vertices, jobs, vehicles, time slots, or resources. Each variable represents one local choice from a finite set of admissible values. A standard way to encode such variables on qubits is to use fixed-Hamming-weight registers, especially one-hot and \(k\)-hot encodings \cite{Hadfield2019AOA,Fuchs2022ConstrainedMixers, sawaya2023encoding,tsvelikhovskiy2024symmetries,onahce}. The
simplest case is a one-hot block of local size \(n\),
where only one of \(n\) alternatives is selected. For \(m\) decision variables, each taking values in \([n]\), the classical configuration space is the product alphabet space (\cref{eq:product-space}) 
\[
\mathcal X=[n]^m .
\]
Here \(n\) is the local alphabet size of each block, and \(m\) is the number of
product blocks. A configuration \(x=(x_1,\dots,x_m)\in \mathcal X\) records the value \(x_b\in[n]\) of each classical decision variable. In a one-hot qubit realization, the value \(x_b=r\) is encoded by the basis vector \[ \ket{e_r} = \ket{0\cdots 010\cdots 0}, \] whose single excitation occupies the \(r\)-th position of an \(n\)-qubit block. We use \emph{variable} for the classical coordinate \(x_b\), \emph{local register} for its logical quantum representation, and \emph{one-hot block} for this \(n\)-qubit realization. With \(m\) one-hot blocks, the corresponding encoded Hilbert space is \[ 
\mathcal H_{\mathrm{OH}} = \mathcal H_1^{\otimes m}, \qquad \mathcal H_1 = \operatorname{span}\{\ket{e_1},\dots,\ket{e_n}\}, 
\] 

and the classical word \(x\in \mathcal X\) labels the computational basis state 
\[ 
\ket{x} = \ket{e_{x_1}}\otimes\cdots\otimes\ket{e_{x_m}}. 
\] 

Thus \(\mathcal X=[n]^m\) is the classical product set, while \(\mathcal H_{\mathrm{OH}}\) is its Hilbert-space realization. \footnote{Here and throughout, the subscript \(k\) in
\(\mathcal H_k\) denotes the fixed Hamming weight of the local block, rather
than a register index.} 

This product form appears throughout constrained quantum-optimization
encodings in various forms. For example, in a position-based encoding of the  well-known 
Traveling Salesman 
Problem (TSP)~\cite{MillerTuckerZemlin1960}, \(m\) can represent tour positions and \(n\) the set of cities \cite{Lucas2014}; in graph coloring, \(m\) can represent vertices and \(n\) the available colors
\cite{Wang2020XYMixers,Lucas2014}; in vehicle-routing or scheduling variants, blocks can represent customer positions, time slots, vehicles, or assignment locations \cite{onah2026cvrp,TothVigo2014VRP}. The problem constraints then appear as additional diagonal penalty terms, feasibility projectors, or decoding rules on the product space \cite{Hadfield2019AOA,onahce}.

More generally, a local register may be realized by an \(n\)-qubit \(k\)-hot block. Its logical state space is the fixed-Hamming-weight subspace
\begin{equation}
\label{eq:multi-ecit}
\mathcal H_k
=
\mathrm{span}\{\ket{x}\in\{0,1\}^n:\|x\|_1=k\},   
\end{equation}
and the local basis can be identified with the family of
\(k\)-subsets of \([n]\). Thus one block has effective alphabet
\(\binom{[n]}{k}\), with local dimension \(\binom nk\). With \(m\) such blocks,
the classical product configuration space is

\begin{equation}
\label{eq:multi-ecit-cardinality}
\mathcal X_k
=
\binom{[n]}{k}^{\,m}. 
\end{equation}
The corresponding encoded Hilbert space is 
$\mathcal H_k^{\otimes m}$.

This form is useful for cardinality-constrained problems, multiple-choice
allocation problems, matching variants, portfolio or resource-allocation
models, and related settings in which each local variable selects \(k\) items rather than one \cite{Hadfield2019AOA,Cook2019MaxKVertexCover,B_rtschi_2019}.

The role of a constrained mixer is to preserve the encoded Hilbert space while transporting amplitudes between admissible basis configurations \cite{Hadfield2019AOA}. Thus, in an alternating quantum optimization circuit, mixer layers redistribute amplitude across the encoded space, while diagonal phase operators attach problem-dependent phases to the basis states. Measurement in the computational basis then produces samples from the resulting distribution, which may be evaluated directly or passed through a problem-specific decoding map. The analytical task is therefore to understand how much mixer-induced transfer reaches useful regions of the product space and how much of that transfer survives the subsequent phase interference. For the complete-graph block-\(XY\) mixer considered below, Lemma~\ref{lem:single-block-xy-amplitudes} reduces the one-block transition kernel to two amplitudes. This leads to the exact shell-transfer generating function in
Proposition~\ref{prop:shell-transfer-generating}, the radial recursion in
Theorem~\ref{thm:exact-shell-transfer-recursion}, and the phase-sensitive
amplitude bound in Theorem~\ref{thm:phase-aligned-path-sum-lower-bound}.


\subsection{XY mixers as constrained product-space transport geometries}
\label{subsec:xy-mixers-product spaces}

Number-conserving \(XY\)-type mixers provide a large family of constrained product space dynamics in quantum optimization~\cite{Hadfield2019AOA,Wang2020XYMixers}. The construction is based on the fact that the two-qubit exchange generator
\[
X_iX_j+Y_iY_j
=
2(\sigma_i^+\sigma_j^-+\sigma_i^-\sigma_j^+)
\]
commutes with the total excitation-number operator and therefore preserves Hamming weight. Hence, given a \emph{mixer interaction graph} \(G_{\mathrm{mix}}=(V,E_{\mathrm{mix}})\), distinct in general from any graph appearing in the optimization instance, we define \[ H_{XY}(G_{\mathrm{mix}}) = \sum_{(i,j)\in E_{\mathrm{mix}}} \bigl(X_iX_j+Y_iY_j\bigr). \] Because every exchange term commutes with the excitation-number operator on its connected component, \(H_{XY}(G_{\mathrm{mix}})\) preserves the Hamming weight of the qubits in each connected component of \(G_{\mathrm{mix}}\). The choice of \(G_{\mathrm{mix}}\) therefore specifies which encoded basis states are directly connected by the mixer dynamics \cite{Hadfield2019AOA,Wang2020XYMixers,KordonowyLeipold2026XYLie,leipold2026imposing}. The initial state must lie in the invariant sector preserved by the chosen mixer~\cite{Hadfield2019AOA,tsvelikhovskiy2024symmetries}. For a one-hot block, common choices include a one-hot computational basis state or the uniform one-excitation state
\[ \ket{W_n} = \frac{1}{\sqrt n}\sum_{r=1}^{n}\ket{e_r}. \] For \(m\) one-hot registers, the corresponding product initialization is 

\[ \ket{s_0} = \ket{W_n}^{\otimes m}.
\] 

More generally, a \(k\)-hot register may be initialized in a fixed-weight basis state or a Dicke state. The mixer then transports amplitude entirely within the fixed-Hamming-weight sector selected by the initialization \cite{Hadfield2019AOA,B_rtschi_2019,onahce}.

Different variants of \(XY\) mixers arise from the need to match the mixer dynamics to the constraint structure of the optimization problem. Local cardinality constraints determine the admissible state space of each decision variable, such as a fixed-cardinality \(k\)-hot sector, and therefore determine the class of block-\(XY\) dynamics that preserves that sector~\cite{Hadfield2019AOA}. The remaining constraints couple the blocks and determine the geometry of feasible configurations within the full product space. Routing, assignment, scheduling, and capacity constraints impose different global relations among otherwise local choices, thereby creating problem-dependent transport bottlenecks. The mixer topology must therefore be chosen to provide transport compatible with this local--global constraint geometry.

For example, in a position-based TSP encoding, the one-hot blocks represent local city choices, while the permutation constraints couple all tour positions globally. Path- or ring-\(XY\) mixers generate only short-range transport among the local symbols and cannot efficiently overcome the probability bottlenecks induced by these global permutation constraints. A complete-graph or other problem-aligned mixer instead supplies the nonlocal
transport required to connect the relevant regions of the product space. Partitioned or Trotterized realizations can reduce the simultaneous interaction requirements of the selected logical mixer, while multi-angle and extended mixers introduce additional variational control or excitation-preserving generators. These constructions preserve related constraint sectors while
producing different transition kernels, spectra, reachability properties, and implementation costs
\cite{Hadfield2019AOA,Wang2020XYMixers,He2023A,
KordonowyLeipold2026XYLie,
awasthi2026constraintpreservingxymixerstrotterized,Ruan2025XYCD}. A broader catalogue of these constructions and their product-space consequences is provided in Appendix~\ref{app:xy-mixer-catalogue}, Table~\ref{tab:xy-mixer-variants-product spaces}. The present analysis focuses on the complete-graph block-\(XY\) mixer because its permutation symmetry produces an exact Hamming-shell reduction. Other mixer geometries generally require refined quotient variables, including path distance, cyclic distance, graph distance, edge classes, or orbit data. Their finite-size transport properties are compared in Section~\ref{subsec:numerical-transport-diagnostics}.

\subsection{The OFM kernel and CE--QAOA dynamics}
\label{subsec:ofm-kernel-ceqaoa-dynamics}

The quantum dynamics studied here is motivated by the
\emph{quantum alternating-operator ansatz} paradigm and symmetry considerations, which acts invariantly on constraint projectors and thereby confines evolution to structured subspaces
\cite{Hadfield2019AOA,Hadfield_2022,Fuchs2022ConstrainedMixers,sawaya2023encoding,tsvelikhovskiy2024symmetries}.
The \emph{Constraint-Enhanced QAOA} (CE--QAOA) introduced in
Ref.~\cite{onahce} follows this tradition but makes problem--algorithm alignment explicit. It fixes a product k-hot encoding, a compatible mixer family, and a matching initial state, and then exploits additional symmetries coming from the hard
constraints of the optimization problem as a codesign choice \cite{Li2021CoDesign}. We recall the following definition.

\begin{definition}[The Onah--Firt--Michielsen (OFM) kernel\cite{onahce}]
\label{def:kernel-requirement}
An optimization instance belongs to the \emph{Onah--Firt--Michielsen kernel} if there
exist integers \(n,m\in\mathbb N\) and the one-hot encoder
\(\mathsf E_{\mathrm{1hot}}\) such that the dynamics is initialized in the
fixed-Hamming-weight space
\[
\OH=(\mathcal H_1)^{\otimes m},
\qquad
\mathcal H_1=\mathrm{span}\{\ket{e_1},\dots,\ket{e_n}\},
\]
with one excitation per block. We identify the computational basis states of
\(\OH\) with words \(x=(x_1,\dots,x_m)\in[n]^m\), where
\(
\ket{x}
=
\ket{e_{x_1}}\otimes\cdots\otimes\ket{e_{x_m}}.
\)
In this encoded basis, the problem Hamiltonian decomposes as
\[
H_C
=
H_{\mathrm{pen}}
+
H_{\mathrm{obj}}
+
H^{'}_{\mathrm{pen}},
\]
where \(H_{\mathrm{obj}}\) is diagonal in the computational basis and encodes
the objective, while \(H_{\mathrm{pen}}\) is diagonal and satisfies:

\begin{enumerate}[label=\textup{(\alph*)}, leftmargin=2.2em]

\item \emph{Penalty structure.}
\(H_{\mathrm{pen}}\) is a sum of polynomially many squared affine
one-hot/degree/capacity penalties, optionally together with linear forbids,
with integer coefficients bounded by \(\mathrm{poly}(n)\). We adopt the
penalty convention that these terms are nonnegative on the encoded
configuration space and vanish when their associated constraints are
satisfied. Then
\[
\spec(H_{\mathrm{pen}})
\subseteq
\{0,1,\dots,t_{\max}\},
\qquad
t_{\max}=\mathrm{poly}(n,m).
\]

\item \emph{Pattern symmetry, level sets, and feasibility.}
\(H_{\mathrm{pen}}\) is invariant under
(i) block permutations \(S_m\) and
(ii) global symbol relabelings \(S_n\).
Hence the configuration space decomposes into level sets
\[
L_t
:=
\left\{
x\in\mathcal X:
H_{\mathrm{pen}}(x)
=
\langle x|H_{\mathrm{pen}}|x\rangle
=
t
\right\},
\]
which are preserved setwise. The zero set of the symmetric penalty is
\[
\mathcal F_0
:=
L_0
=
\left\{
x\in\mathcal X:
H_{\mathrm{pen}}(x)=0
\right\}.
\]

The diagonal term \(H^{'}_{\mathrm{pen}}\) captures additional penalty terms that break the above symmetry. The feasible set of the full problem is
\[
\mathcal F
=
\mathcal F_0
\cap
\left\{
x\in\mathcal X:
H^{'}_{\mathrm{pen}}(x)=0
\right\}.
\]

\item \emph{Mixer and initial state.}
The block-local normalized \(XY\) mixer is
\begin{equation}
\label{eq:mixer}
\widetilde H_{XY}^{(b)}
=
\frac{1}{n-1}
\sum_{1\le u<v\le n}
\bigl(
X_u^{(b)}X_v^{(b)}
+
Y_u^{(b)}Y_v^{(b)}
\bigr),
\end{equation}
with
\(
\|\widetilde H_{XY}^{(b)}\|=O(1)
\)
on each block. The initial state is the uniform one-hot product
\begin{equation}
\label{eq:initial-state}
\ket{s_0}
=
\ket{s_{\mathrm{blk}}}^{\otimes m},
\qquad
\ket{s_{\mathrm{blk}}}
=
\frac{1}{\sqrt n}
\sum_{r=1}^{n}
\ket{e_r}.
\end{equation}

\end{enumerate}
\end{definition}

We use the name Onah--Firt--Michielsen kernel to distinguish the encoded structural class from the CE--QAOA ansatz acting on it. The OFM kernel refers to the product encoded constraint structure introduced in Ref.~\cite{onahce}, whereas CE--QAOA denotes the variational circuit built from a compatible mixer, diagonal phase operators, and an initial state on that encoded space. The definition above is stated for one-hot blocks. Many constrained
optimization problems naturally use \(k\)-hot blocks, where each local register has fixed Hamming weight \(k\). In that case one replaces \(\mathcal H_1\) by \(\mathcal H_k\) as in Eq. \ref{eq:multi-ecit} and chooses a number-conserving block-local mixer such as a suitable \(XY\)-type Hamiltonian, normalized to constant operator norm on \(\mathcal H_k\). The corresponding initial state is a tensor product of Dicke
states~\cite{B_rtschi_2019,yu2026efficient},
\[
\ket{s_{\mathrm{blk}}^{(k)}}
=
\ket{D^{(n)}_{k}}
=
\binom{n}{k}^{-1/2}
\sum_{\substack{x\in\{0,1\}^n\\ \|x\|_1=k}}
\ket{x},
\qquad
\ket{s_0}
=
\bigl(\ket{D^{(n)}_{k}}\bigr)^{\otimes m}.
\]
This extension enlarges the OFM-type setting to include constrained
combinatorial optimization problems such as TSP \cite{Lucas2014,GareyJohnson1979}, QAP
\cite{Loiola2007QAPSurvey}, CVRP \cite{TothVigo2014VRP}, shared
transportation problems \cite{onah2025waas}, generalized assignment and multiple-knapsack problems \cite{SahniGonzalez1976GAP}, and \(k\)-dimensional matching for \(k\ge3\) \cite{Karp1972}. See Sec. \ref{subsec:multidimensional-jump-statistics} for more discussion on various application classes.

\section{Main Results I: Amplitude Transport Analyses}
\label{sec:results}

A depth-\(p\) CE--QAOA circuit is given as
\begin{equation}
\ket{\psi_p(\vec\gamma,\vec\beta)}
=
\Bigl(\prod_{\ell=1}^{p}U_M(\beta_\ell)e^{-i\gamma_\ell H_C}\Bigr)\ket{s_0},
\qquad
\vec\gamma=(\gamma_1,\dots,\gamma_p),\;
\vec\beta=(\beta_1,\dots,\beta_p).
\end{equation}

Because \(U_M\) preserves the one-hot sector and \(H_C\) is diagonal, every
layer maps the encoded manifold \(\OH=(\mathcal H_1)^{\otimes m}\) to itself. This persistent  product-space structure is what makes the following shell reduction possible. From this point onward, the shell-transfer analysis is formulated on the abstract product space \(\mathcal X=[n]^m\), with CE--QAOA serving as a primary motivating instance.

\subsection{Shell transfer on the product one-hot manifold}
\label{subsec:shell-transfer-product-space}

Fix integers \(n\ge2\) and \(m\ge1\), and let
\[
\mathcal X=[n]^m
\]
be the product alphabet space. For \(x=(x_1,\dots,x_m)\in \mathcal X\), we write
\[
d(x,y):=\bigl|\{b\in\{1,\dots,m\}:x_b\neq y_b\}\bigr|
\]
for the Hamming distance on \(\mathcal \mathcal X\). For a fixed reference configuration
\(y\in \mathcal X\), define the shells
\begin{equation}
\label{eq:shells}
S_r(y):=\{x\in \mathcal X:\ d(x,y)=r\},
\qquad r=0,\dots,m.
\end{equation}

Throughout the shell-transfer analysis we use \(\beta\) as the effective
one-hot mixing angle. With the normalization in
\cref{eq:mixer}, the block mixer is therefore
\begin{equation}
\label{eq:mixer-unitary}
U_M^{(b)}(\beta)
:=
\exp\!\left(
-\frac{i\beta}{2}\widetilde H_{XY}^{(b)}
\right).
\end{equation}
Equivalently, the factor of \(2\) coming from
\(X_iX_j+Y_iY_j=2(\sigma_i^+\sigma_j^-+\sigma_i^-\sigma_j^+)\) is absorbed into the definition of the mixer angle. Observe, then, that on a single one-hot block of local size \(n\), the normalized block-XY mixer unitary has two distinct matrix element values
\begin{align}
\label{eq:block}
a_0(\beta)
=
\frac{1}{n}e^{-i\beta}
+
\left(1-\frac{1}{n}\right)e^{i\beta/(n-1)},
\qquad
a_1(\beta)
=
\frac{1}{n}\left(e^{-i\beta}-e^{i\beta/(n-1)}\right).
\end{align}

\begin{lemma}[Single-block transition amplitudes for the complete-graph XY mixer]
\label{lem:single-block-xy-amplitudes}
Let
\begin{equation}
\label{eq:one-hot-space}
\mathcal H_1=\mathrm{span}\{\ket{e_1},\dots,\ket{e_n}\}
\end{equation}
be the one-hot subspace of a block of size \(n\), and write
\[
\ket r:=\ket{e_r},\qquad r\in[n].
\]
Consider the normalized complete-graph block-XY Hamiltonian
\(
\widetilde H_{XY}^{(b)}\) as in \cref{eq:mixer} and the associated mixer unitary \(
U_M^{(b)}(\beta)\) in \cref{eq:mixer-unitary}. Then, on \(\mathcal H_1\), the matrix elements of \(U_M^{(b)}(\beta)\)
take  two values:
\[
\bra r U_M^{(b)}(\beta)\ket q
=
\begin{cases}
a_0(\beta), & r=q,\\[2mm]
a_1(\beta), & r\neq q,
\end{cases}
\]
where
\begin{align}
\label{eq:a0-single-block}
a_0(\beta)
&=
\frac{1}{n}e^{-i\beta}
+
\left(1-\frac{1}{n}\right)e^{i\beta/(n-1)},\\[1mm]
\label{eq:a1-single-block}
a_1(\beta)
&=
\frac{1}{n}
\left(
e^{-i\beta}
-
e^{i\beta/(n-1)}
\right).
\end{align}
\end{lemma}

\noindent
This result is derived from the spectral decomposition of \(U_M^{(b)}(\beta)\). Detailed proof appears in App. \ref{app:proofs}. The two one-block amplitudes also determine the total absolute transition
weight from any one-hot basis state. Since there is one unchanged symbol and
\(n-1\) distinct changed symbols, define
\begin{equation}
\label{eq:one-layer-weight}
q_n(\beta)
:=
|a_0(\beta)|+(n-1)|a_1(\beta)|.
\end{equation}
Thus \(q_n(\beta)\) is the one-block absolute row sum of the mixer kernel.
For the \(m\)-block product mixer, the corresponding total absolute
transition weight factorizes as \(q_n(\beta)^m\).

By Lemma~\ref{lem:single-block-xy-amplitudes}, for a product configuration
\[
x=(x_1,\dots,x_m),\qquad z=(z_1,\dots,z_m)\in \mathcal X=[n]^m,
\]
and the product mixer
\[
U_M(\beta)=\bigotimes_{b=1}^m U_M^{(b)}(\beta),
\]
the transition amplitude factorizes as
\[
\bra z U_M(\beta)\ket x
=
\prod_{b=1}^m
\bra {z_b}U_M^{(b)}(\beta)\ket{x_b}.
\]

Using the convention \(\binom ab=0\) whenever \(b<0\) or \(b>a\), each coordinate contributes \(a_0(\beta)\) if \(z_b=x_b\), and \(a_1(\beta)\) if \(z_b\neq x_b\). Hence,
\begin{equation}
\label{eq:product-kernel-from-single-block}
\bra z U_M(\beta)\ket x
=
a_0(\beta)^{m-d(x,z)}
a_1(\beta)^{d(x,z)}.
\end{equation}

Next, we use Lemma~\ref{lem:single-block-xy-amplitudes} together with
\cref{eq:one-layer-weight,eq:product-kernel-from-single-block} to show that
the complete-graph block-\(XY\) mixer gives a well-defined shell-transfer
mechanism and ultimately compresses the multilayer path expansion into a
transfer generating function. To that end, for \(x\in S_t(y)\), define the one-layer shell-transfer quantity
\begin{equation}
\label{eq:Trt-def-general}
T_{r,t}^{(y)}(x;\beta)
:=
\sum_{z\in S_r(y)}
\bigl|\langle z\mid U_M(\beta)\mid x\rangle\bigr|.
\end{equation}



\begin{lemma}[Well-defined shell transfer]
\label{lem:shell-transfer-well-defined}
For fixed \(r,t\in\{0,\dots,m\}\), the quantity
\(T_{r,t}^{(y)}(x;\beta)\) depends only on \(r,t,\beta\), and not on the
particular choice of \(y\in \mathcal X\) or \(x\in S_t(y)\). We therefore write
\[
T_{r,t}(\beta):=T_{r,t}^{(y)}(x;\beta)
\qquad
\text{for any }y\in \mathcal X,\ x\in S_t(y).
\]
\end{lemma}

\begin{proof}
Fix \(y\in \mathcal X\) and \(x\in S_t(y)\). Let
\[
A:=\{b\in\{1,\dots,m\}:x_b\neq y_b\},
\qquad |A|=t.
\]

By Eq.~\ref{eq:product-kernel-from-single-block}, the   kernel factorizes
coordinatewise and depends only on whether two local symbols are equal or
different. For each coordinate \(b\), the local contribution to the shell
transfer function is determined only by whether \(x_b=y_b\). If
\(x_b=y_b\), the local possibilities are \(z_b=y_b\) or \(z_b\neq y_b\). If
\(x_b\neq y_b\), the local possibilities are \(z_b=y_b\), \(z_b=x_b\), or
\(z_b\notin\{x_b,y_b\}\). Thus the full product contribution depends only on
the number \(t=d(x,y)\), and not on the positions or labels of the mismatches.
Hence \(T_{r,t}^{(y)}(x;\beta)\) depends only on \(r,t,\beta\).
\end{proof}

\begin{proposition}[Shell-transfer generating function]
\label{prop:shell-transfer-generating}
For every \(t\in\{0,\dots,m\}\), the shell-transfer coefficients
\(T_{r,t}(\beta)\) are the coefficients of the polynomial
\begin{equation}
\label{eq:shell-transfer-gf-general}
\sum_{r=0}^{m} T_{r,t}(\beta)\,\zeta^r
=
\Bigl(|a_1(\beta)| + \bigl(|a_0(\beta)|+(n-2)|a_1(\beta)|\bigr)\zeta\Bigr)^t
\Bigl(|a_0(\beta)| + (n-1)|a_1(\beta)|\zeta\Bigr)^{m-t}.
\end{equation}
Equivalently,
\begin{align}
\label{eq:shell-transfer-explicit-general}
T_{r,t}(\beta)
&=
\sum_{j=0}^{\min\{t,r\}}
\binom{t}{j}
\binom{m-t}{r-j}
|a_1(\beta)|^{\,t-j}
\bigl(|a_0(\beta)|+(n-2)|a_1(\beta)|\bigr)^j\\ \nonumber
&\times
|a_0(\beta)|^{\,m-t-r+j}
\bigl((n-1)|a_1(\beta)|\bigr)^{r-j}.
\end{align}
\end{proposition}

\begin{proof}
Fix \(y\in \mathcal X\) and \(x\in S_t(y)\), and let
\[
A:=\{b:\ x_b\neq y_b\},
\qquad |A|=t.
\]
Lemma \ref{lem:shell-transfer-well-defined} shows that $T_{r,t}(\beta)$ is well-defined and depends only on \(r,t,\beta\), and not on the
particular choice of \(y\in \mathcal X\) or \(x\in S_t(y)\). Thus, we classify the contribution of each coordinate to the shell index
\(r=d(z,y)\). By Lemma~\ref{lem:single-block-xy-amplitudes} and \cref{eq:product-kernel-from-single-block}, each coordinate contributes
\(a_0(\beta)\) if \(z_b=x_b\), and \(a_1(\beta)\) if \(z_b\neq x_b\).

For a coordinate \(b\in A\), there are three possibilities for \(z_b\):
\begin{itemize}
\item \(z_b=y_b\): shell increment \(0\),   kernel factor
      \(|a_1(\beta)|\);
\item \(z_b=x_b\): shell increment \(1\),   kernel factor
      \(|a_0(\beta)|\);
\item \(z_b\notin\{x_b,y_b\}\): shell increment \(1\), with \(n-2\) choices,
      each contributing \(|a_1(\beta)|\).
\end{itemize}
Hence each coordinate in \(A\) contributes the generating factor
\[
|a_1(\beta)|
+
\bigl(|a_0(\beta)|+(n-2)|a_1(\beta)|\bigr)\zeta.
\]

For a coordinate \(b\notin A\), we have \(x_b=y_b\), so there are two
possibilities:
\begin{itemize}
\item \(z_b=y_b=x_b\): shell increment \(0\),   kernel factor
      \(|a_0(\beta)|\);
\item \(z_b\neq y_b\): shell increment \(1\), with \(n-1\) choices, each
      contributing \(|a_1(\beta)|\).
\end{itemize}
Hence each coordinate outside \(A\) contributes the generating factor
\[
|a_0(\beta)|+(n-1)|a_1(\beta)|\zeta.
\]

Multiplying over the \(t\) coordinates in \(A\) and the \(m-t\) coordinates
outside \(A\) gives \eqref{eq:shell-transfer-gf-general}. Extracting the
coefficient of \(\zeta^r\) yields
\eqref{eq:shell-transfer-explicit-general}.
\end{proof}

A useful consequence of the generating function is obtained by setting
\(\zeta=1\). One immediately observes that the total absolute transfer mass satisfies
\[
\sum_{r=0}^{m}T_{r,t}(\beta)
=
\bigl(|a_0(\beta)|+(n-1)|a_1(\beta)|\bigr)^m,
\]
independently of the incoming shell \(t\). Thus the shell-transfer
coefficients admit a shell-independent normalization, and the resulting
normalized generating function factorizes into Bernoulli
probability-generating functions for coordinates initially matched or
mismatched with the reference configuration. This yields the
Poisson-binomial interpretation of the one-step radial law developed in
Appendix~\ref{app:measure-shell-reduction}. We now turn to the corresponding
multilayer shell recursion. Define the initial shell-count vector by
\begin{equation}
\label{eq:v0-def-general}
v_t^{(0)}
:=
|S_t(y)|
=
\binom{m}{t}(n-1)^t,
\qquad t=0,\dots,m,
\end{equation}
and for \(\ell\ge 1\) define recursively
\begin{equation}
\label{eq:vp-recursion-general}
v_r^{(\ell)}
:=
\sum_{t=0}^{m}T_{r,t}(\beta_\ell)\,v_t^{(\ell-1)},
\qquad r=0,\dots,m.
\end{equation}
\noindent
This generating function is the basepoint-radial quotient of the coordinatewise product transfer operator on the \(n\)-ary Hamming scheme \(H(m,n)\). The coefficients \(T_{r,t}(\beta)\) are the shell-to-shell orbit sums obtained by projecting the product transfer onto the Hamming distance partition around the reference point \(y\). The shell recursion is therefore given by the following
finite-dimensional transfer rule.

\begin{theorem}[Exact shell-transfer recursion]
\label{thm:exact-shell-transfer-recursion}
Let \(v^{(0)}=(v^{(0)}_0,\dots,v^{(0)}_m)\) be the initial shell-weight vector.
For every \(\ell\ge 1\) and every \(r\in\{0,\dots,m\}\),
\begin{equation}
\label{eq:exact-shell-transfer-recursion}
v_r^{(\ell)}
=
\sum_{t=0}^{m} T_{r,t}(\beta_\ell)\,v_t^{(\ell-1)} .
\end{equation}
Equivalently,
\begin{equation}
\label{eq:vr-path-meaning}
v_r^{(\ell)}
=
\sum_{\substack{x_0,\dots,x_\ell\in \mathcal X\\ x_\ell\in S_r(y)}}
\prod_{j=1}^{\ell}
\bigl|\langle x_j\mid U_M(\beta_j)\mid x_{j-1}\rangle\bigr|,
\end{equation}
with the convention that the initial weight depends only on the shell of
\(x_0\). Thus \(v_r^{(\ell)}\) is the total weighted contribution of all
length-\(\ell\) paths
\[
x_0\to x_1\to\cdots\to x_\ell
\]
ending in the shell \(S_r(y)\), where the transfer weight of the path is
\[
\prod_{j=1}^{\ell}
\bigl|\langle x_j\mid U_M(\beta_j)\mid x_{j-1}\rangle\bigr|.
\]
In particular, since \(S_0(y)=\{y\}\),
\begin{equation}
\label{eq:v0-closed-form}
v_0^{(p)}
=
\sum_{\substack{x_0,\dots,x_p\in \mathcal X\\ x_p=y}}
\prod_{j=1}^{p}
\bigl|\langle x_j\mid U_M(\beta_j)\mid x_{j-1}\rangle\bigr|
\end{equation}
is the total weighted contribution of all depth-\(p\) paths ending exactly at
\(y\).
\end{theorem}

\begin{proof}
We argue by induction on \(\ell\).

For \(\ell=0\), the identity reduces to
\(
v_r^{(0)}=|S_r(y)|,
\)
which is eq. \eqref{eq:v0-def-general}.

Assume the claim holds at depth \(\ell-1\). Then
\[
v_r^{(\ell)}
=
\sum_{t=0}^{m}T_{r,t}(\beta_\ell)\,v_t^{(\ell-1)}.
\]
By Lemma~\ref{lem:shell-transfer-well-defined} and \cref{eq:Trt-def-general},
\[
T_{r,t}(\beta_\ell)
=
\sum_{x_\ell\in S_r(y)}
\bigl|\langle x_\ell\mid U_M(\beta_\ell)\mid x_{\ell-1}\rangle\bigr|
\qquad
\text{for any fixed }x_{\ell-1}\in S_t(y).
\]
Summing over all such \(x_{\ell-1}\) with their depth-\((\ell-1)\) weighted
prefix contributions yields the total weighted contribution of all
length-\(\ell\) paths ending in \(S_r(y)\). This is
\eqref{eq:vr-path-meaning}. The final statement follows by taking \(r=0\),
since \(S_0(y)=\{y\}\).
\end{proof}

\noindent
Theorem~\ref{thm:exact-shell-transfer-recursion} gives the radial transport statement needed for the phase-amplitude analysis below. The induction proof establishes the shell recursion at the level of weighted path sums. A measure-theoretic characterization of the same reduction, including its formulation as a radial pushforward of a product path measure and as a lumped Markov recursion, is given in Appendix~\ref{app:measure-shell-reduction}. Further applications and consequences are discussed in Sec. \ref{sec:applications}.

\section{Main Results II: Interference and  Phase Control Analysis}
\label{sec:main2}

\subsection{Lattice normalization and phase control}
\label{subsec:shell-transfer-to-amplitude}

We now separate the second ingredient, namely the passage from  
shell-transfer counts to lower bounds on complex amplitudes. This step requires a phase alignment hypothesis on the full path summands to supplement the recursion in Theorem \ref{thm:exact-shell-transfer-recursion} because the shell-transfer recursion is phase-blind. It accounts for the   product-kernel contribution of all paths ending in a prescribed shell. To turn this   transfer into a lower bound on an actual complex amplitude, one must additionally control the phases of the path summands. The most transparent regime in which this can be done is the lattice-normalized regime. To that end, let \(H_C\) be diagonal in the computational basis, with eigenvalue function \(E:\mathcal X\to\mathbb R\). The diagonal cost is exactly
lattice-normalized if, after an affine rescaling and shift,
\begin{equation}
\label{eq:lattice-normalized-cost}
\widetilde E(x)=\alpha E(x)+\eta_0 \in \{0,1,\dots,T_n\},
\qquad x\in \mathcal X=[n]^m,
\end{equation}
for some integer \(T_n=\poly {n,m}\). Since adding a global constant to the Hamiltonian contributes only an overall phase, we may equivalently work with the normalized cost
\begin{equation}
\label{eq:lattice-normalized-cost-range}
0\le E(x)\le T_n,
\qquad
E(x)\in\{0,1,\dots,T_n\},
\qquad x\in \mathcal X.
\end{equation}

Fix an angle \(0<\theta<\pi\), and choose
\begin{equation}
\label{eq:lattice-gamma-choice}
\gamma=\frac{\theta}{T_n}.
\end{equation}
Then for every \(x\in \mathcal X\),
\begin{equation}
\label{eq:lattice-single-layer-window}
0\le \gamma E(x)\le \theta.
\end{equation}
Hence all cost-phase factors
\[
e^{-i\gamma E(x)}
\]
lie inside a common arc of angular width at most \(\theta\). Thus lattice normalization gives an explicit phase cone for the diagonal layer. For depth \(p\), suppose more generally that for each layer \(\ell\),
\begin{equation}
\label{eq:lattice-layer-window}
0\le \gamma_\ell E(x)\le \theta_\ell
\qquad \text{for all }x\in \mathcal X,
\end{equation}
and define
\begin{equation}
\label{eq:lattice-total-cone-width}
\Theta_p:=\theta_1+\cdots+\theta_p.
\end{equation}
Then the cost-phase contribution to any depth-\(p\) path satisfies
\begin{equation}
\label{eq:lattice-path-phase-window}
0\le
\sum_{\ell=1}^{p}\gamma_\ell E(x_{\ell-1})
\le
\Theta_p.
\end{equation}
Consequently,
\begin{equation}
\label{eq:lattice-path-cost-phase}
\prod_{\ell=1}^{p}e^{-i\gamma_\ell E(x_{\ell-1})}
=
\exp\!\left(
-i\sum_{\ell=1}^{p}\gamma_\ell E(x_{\ell-1})
\right)
\end{equation}
lies in a common arc of angular width at most \(\Theta_p\). If
\(\Theta_p<\pi\), these cost phases lie in a common open cone.

The remaining phases in the path expansion come from the mixer matrix
elements. Thus, in the exact lattice-normalized regime, the diagonal phases are controlled explicitly, while the only additional requirement for a full path-sum cone bound is that the relevant mixer-kernel phases remain controlled within the same transition-signature class. Under that phase-control condition, the shell-transfer count \(v_0^{(p)}\) converts directly into an amplitude lower bound. 

\subsection{ Lower bounds on complex amplitudes}
Let
\[
\ket{\psi_p(\vec\gamma,\vec\beta)}
=
\Bigl(\prod_{\ell=1}^{p}U_M(\beta_\ell)e^{-i\gamma_\ell H_C}\Bigr)\ket{s_0},
\qquad
\ket{s_0}=\frac{1}{n^{m/2}}\sum_{x\in \mathcal X}\ket{x},
\]
where \(H_C\) is diagonal in the computational basis, and \(\mathcal X=[n]^m\).

\begin{lemma}[Exact path-sum formula]
\label{lem:exact-path-sum}
For every \(y\in \mathcal X\),
\begin{equation}
\label{eq:exact-path-sum}
\langle y\mid\psi_p(\vec\gamma,\vec\beta)\rangle
=
\frac{1}{n^{m/2}}
\sum_{x_0,\dots,x_{p-1}\in \mathcal X}
\prod_{\ell=1}^{p}
\Bigl(
e^{-i\gamma_\ell E(x_{\ell-1})}
\langle x_\ell\mid U_M(\beta_\ell)\mid x_{\ell-1}\rangle
\Bigr),
\qquad x_p:=y.
\end{equation}
\end{lemma}

\begin{proof}
Expand \(\ket{s_0}\) in the computational basis, apply each diagonal cost layer
\(e^{-i\gamma_\ell H_C}\), and then insert resolutions of the identity between
the mixer layers.
\end{proof}

\begin{lemma}[Cone lower bound for complex sums]
\label{lem:cone-lower-bound-shell}
Let \(z_1,\dots,z_M\in\mathbb C\) lie in an arc of angular width
\(\Theta\in[0,\pi)\). Then
\[
\left|\sum_{j=1}^{M}z_j\right|
\ge
\cos\!\left(\frac{\Theta}{2}\right)\sum_{j=1}^{M}|z_j|.
\]
\end{lemma}

\begin{proof}
Rotate the arc so that its midpoint lies on the positive real axis, and then
bound the real part of each summand from below by
\(\cos(\Theta/2)\,|z_j|\).
\end{proof}

\begin{theorem}[Phase-aligned path-sum lower bound]
\label{thm:phase-aligned-path-sum-lower-bound}
Fix \(y\in \mathcal X\). Assume that all depth-\(p\) path summands in
\eqref{eq:exact-path-sum} lie in a common arc of angular width
\(\Theta_p<\pi\). Then
\begin{equation}
\label{eq:phase-aligned-pointwise}
\bigl|\langle y\mid\psi_p(\vec\gamma,\vec\beta)\rangle\bigr|
\ge
\frac{1}{n^{m/2}}
\cos\!\left(\frac{\Theta_p}{2}\right)\,v_0^{(p)}.
\end{equation}
Equivalently,
\begin{equation}
\label{eq:phase-aligned-pointwise-prob}
\bigl|\langle y\mid\psi_p(\vec\gamma,\vec\beta)\rangle\bigr|^2
\ge
\frac{1}{n^{m}}
\cos^2\!\left(\frac{\Theta_p}{2}\right)\,\bigl(v_0^{(p)}\bigr)^2.
\end{equation}
\end{theorem}

\begin{proof}
Apply Lemma~\ref{lem:cone-lower-bound-shell} to the complex summands in the
path sum \eqref{eq:exact-path-sum}. Their total sum of moduli is exactly
\[
\sum_{\substack{x_0,\dots,x_p\in \mathcal X\\ x_p=y}}
\prod_{\ell=1}^{p}
\bigl|\langle x_\ell\mid U_M(\beta_\ell)\mid x_{\ell-1}\rangle\bigr|
=
v_0^{(p)}
\]
by Theorem~\ref{thm:exact-shell-transfer-recursion}. This yields
\eqref{eq:phase-aligned-pointwise}. Squaring gives
\eqref{eq:phase-aligned-pointwise-prob}.
\end{proof}


Theorem~\ref{thm:phase-aligned-path-sum-lower-bound} shows that when all
relevant depth-\(p\) path summands lie in an arc of angular width
\(\Theta_p<\pi\), their projections along the arc bisector add
constructively. Thus, the shell-transfer recursion
(Theorem~\ref{thm:exact-shell-transfer-recursion}) and the phase-aligned
path-sum bound
(Theorem~\ref{thm:phase-aligned-path-sum-lower-bound}) separate
transport from phase coherence. The former computes the available
absolute path weight, while the latter quantifies the fraction retained
under coherent summation. The discussion below makes the corresponding
alignment and cancellation geometry explicit. The subsequent
constructive radial-mass analysis then bridges these two ingredients by
identifying the mixer angle that maximizes the one-site transfer factor
contributing to \(v_0^{(p)}\). This yields an explicit logarithmic-depth
route for raising the target radial mass to the scale required by the
phase-aligned amplitude bound, and hence for obtaining a success
probability independent of the problem size and encoded dimension.

\subsection{Phase alignment and cancellation in the path-sum geometry}
\label{subsec:phase-alignment-and-cancellation}

Theorem~\ref{thm:phase-aligned-path-sum-lower-bound} expresses the
target amplitude as the coherent sum of all depth-\(p\) paths ending at
\(y\). To describe the associated interference geometry explicitly,
define the path summand
\begin{equation}
\label{eq:depth-p-path-summand}
A(x_0,\ldots,x_{p-1};y)
:=
\prod_{\ell=1}^{p}
\Bigl(
e^{-i\gamma_\ell E(x_{\ell-1})}
\langle x_\ell\mid U_M(\beta_\ell)\mid x_{\ell-1}\rangle
\Bigr),
\qquad
x_p:=y.
\end{equation}
For every nonzero path summand, write
\begin{equation*}
\label{eq:depth-p-path-polar-form}
A(x_0,\ldots,x_{p-1};y)
=
w(x_0,\ldots,x_{p-1};y)
e^{i\phi(x_0,\ldots,x_{p-1};y)},
\end{equation*}
where
\begin{equation*}
\label{eq:depth-p-path-weight}
w(x_0,\ldots,x_{p-1};y)
:=
\prod_{\ell=1}^{p}
\bigl|
\langle x_\ell\mid U_M(\beta_\ell)\mid x_{\ell-1}\rangle
\bigr|.
\end{equation*}
The accumulated path phase is
\begin{align}
\label{eq:depth-p-accumulated-path-phase}
\phi(x_0,\ldots,x_{p-1};y)
&=
-\sum_{\ell=1}^{p}
\gamma_\ell E(x_{\ell-1})
\notag\\
&\quad
+
\sum_{\ell=1}^{p}
\arg
\bigl(
\langle x_\ell\mid
U_M(\beta_\ell)
\mid x_{\ell-1}\rangle
\bigr)
\pmod{2\pi}.
\end{align}
The first term contains the phase rotations generated by the diagonal
cost layers, while the second contains the phases of the mixer matrix
elements. Each diagonal factor has unit modulus. When either \(E(x)=0\) or \(\gamma_\ell=0\), the corresponding factor is \(e^{-i\gamma_\ell E(x)} =1.\) The layer then contributes zero angular rotation to the path. Nonzero energy differences and nonzero phase angles generate relative rotations between distinct paths.

To see this directly, consider two depth-\(p\) paths
\[
(x_0,\ldots,x_{p-1},y)
\qquad\text{and}\qquad
(x'_0,\ldots,x'_{p-1},y).
\]
Their relative accumulated phase is
\begin{align}
\label{eq:relative-depth-p-path-phase}
\Delta\phi
&:=
\phi(x_0,\ldots,x_{p-1};y)
-
\phi(x'_0,\ldots,x'_{p-1};y)
\notag\\
&=
-\sum_{\ell=1}^{p}
\gamma_\ell
\Bigl(
E(x_{\ell-1})-E(x'_{\ell-1})
\Bigr)
\notag\\
&\quad
+
\sum_{\ell=1}^{p}
\Bigl[
\arg
\bigl(
\langle x_\ell\mid
U_M(\beta_\ell)
\mid x_{\ell-1}\rangle
\bigr)
\notag\\
&\hspace{4.4cm}
-
\arg
\bigl(
\langle x'_\ell\mid
U_M(\beta_\ell)
\mid x'_{\ell-1}\rangle
\bigr)
\Bigr]
\pmod{2\pi},
\end{align}
where \(x_p=x'_p=y\).

Two path summands cancel pairwise when their weights agree and their
relative phase equals \(\pi\) modulo \(2\pi\):
\begin{align}
\label{eq:pairwise-path-cancellation-condition}
w(x_0,\ldots,x_{p-1};y)
&=
w(x'_0,\ldots,x'_{p-1};y),
\\
\Delta\phi
&=
\pi
\pmod{2\pi}.
\end{align}
For paths with the same accumulated mixer phase, this condition becomes
\begin{equation}
\label{eq:energy-induced-pairwise-cancellation}
\sum_{\ell=1}^{p}
\gamma_\ell
\Bigl(
E(x_{\ell-1})-E(x'_{\ell-1})
\Bigr)
=
\pi
\pmod{2\pi}.
\end{equation}
Thus the accumulated energy difference can rotate two equally weighted path contributions into opposite directions in the complex plane. The complete target amplitude vanishes  when all weighted
path phasors close:
\begin{equation}
\label{eq:full-path-phasor-closure}
\sum_{x_0,\ldots,x_{p-1}\in\mathcal X}
w(x_0,\ldots,x_{p-1};y)
e^{i\phi(x_0,\ldots,x_{p-1};y)}
=
0.
\end{equation}
This closure may arise from two path families with equal aggregate
weight and opposite phases. It may also arise from several
symmetry-related phase sectors. For example, \(q\) equally weighted
sectors with phases separated by \(2\pi/q\) satisfy
\begin{equation}
\label{eq:root-of-unity-path-cancellation}
\sum_{j=0}^{q-1}
e^{2\pi i j/q}
=
0.
\end{equation}


\paragraph{Common-cone geometry.}

Suppose the phases of all nonzero depth-\(p\) path summands lie in a
common arc of angular width \(\Theta_p<\pi\). Let \(\alpha_p\) denote the
midpoint of this arc. Then
\begin{equation}
\label{eq:path-phase-common-arc}
\left|
\phi(x_0,\ldots,x_{p-1};y)-\alpha_p
\right|
\le
\frac{\Theta_p}{2}
\end{equation}
for every contributing path. Rotation by \(e^{-i\alpha_p}\) gives
\begin{align}
\label{eq:path-common-projection}
&
\operatorname{Re}
\left[
e^{-i\alpha_p}
A(x_0,\ldots,x_{p-1};y)
\right]
\notag\\
&\hspace{1cm}
=
w(x_0,\ldots,x_{p-1};y)
\cos
\left(
\phi(x_0,\ldots,x_{p-1};y)-\alpha_p
\right)
\notag\\
&\hspace{1cm}
\ge
w(x_0,\ldots,x_{p-1};y)
\cos
\left(
\frac{\Theta_p}{2}
\right).
\end{align}
Every contributing path therefore has a positive projection along the
arc bisector. Summing these projections gives
\begin{align}
\label{eq:path-cone-projection-sum}
&
\left|
\sum_{x_0,\ldots,x_{p-1}\in\mathcal X}
A(x_0,\ldots,x_{p-1};y)
\right|
\ge
\cos
\left(
\frac{\Theta_p}{2}
\right)
v_0^{(p)},
\end{align}
which is the geometric content of
Theorem~\ref{thm:phase-aligned-path-sum-lower-bound}. For \(\Theta_p<\pi\), the positive common projection prevents the weighted phasors from closing and gives a strictly positive target
amplitude whenever \(v_0^{(p)}>0\).  At the semicircle boundary which occurs at \(\Theta_p=\pi\), the smallest common projection reaches zero. 

\paragraph{The angular margin.}

Define the angular margin from the semicircle boundary by
\[
\delta_p
:=
\pi-\Theta_p,
\qquad
\delta_p>0.
\]
Then
\begin{equation*}
\label{eq:cone-factor-angular-margin}
\cos
\left(
\frac{\Theta_p}{2}
\right)
=
\sin
\left(
\frac{\delta_p}{2}
\right).
\end{equation*}
The amplitude bound
\eqref{eq:phase-aligned-pointwise} can therefore be written as
\begin{equation}
\label{eq:angular-margin-amplitude-bound}
\bigl|
\langle y\mid
\psi_p(\vec\gamma,\vec\beta)
\rangle
\bigr|
\ge
\frac{1}{n^{m/2}}
\sin
\left(
\frac{\delta_p}{2}
\right)
v_0^{(p)}.
\end{equation}

For a small angular margin, \(
\sin
\left(
\frac{\delta_p}{2}
\right)
=
\frac{\delta_p}{2}
+
O(\delta_p^3).
\)
Hence an inverse-polynomial margin
\begin{equation}
\label{eq:inverse-polynomial-angular-margin}
\delta_p
\ge
\frac{1}{\poly n}
\end{equation}
retains an inverse-polynomial fraction of the available path weight in
amplitude and an inverse-polynomial-squared fraction in probability. Notably, this implies that approaching the semicircle boundary is quite harmless to the success guarantees derived in this work. 

It is equally important to point out that the condition \(\Theta_p<\pi\) is sufficient rather than necessary for a nonzero target amplitude. At the boundary \(\Theta_p=\pi\), when the uniform projection factor vanishes, the path sum may nevertheless remain nonzero, since exact cancellation requires the stronger weighted phasor-closure condition
\begin{equation}
\label{eq:weighted-path-phasor-closure}
\sum_{\mathsf P:\,x_p=y}
w_{\mathsf P}e^{i\phi_{\mathsf P}}
=
0.
\end{equation}
Accordingly, phase configurations extending across a semicircle may
still produce a substantial coherent amplitude. More general phase
engineering may exploit nonuniform path weights or structured phase
sectors beyond the single-cone regime. Such extensions are outside the present scope. The next subsection combines the available absolute path mass with the common-cone
projection factor to obtain a polynomial-shot success certificate in Theorem \ref{thm:polynomial-shot-success}. Proposition~\ref{prop:constructive-radial-mass-lower-bound} identifies the transport-maximizing angle, while Corollary~\ref{cor:maximal-transport-phase-obstruction} shows that this transport-maximizing choice places the single-cone bound in its inconclusive regime, motivating further phase engineering to exploit the accumulated path mass.

\subsection{Constructive radial-mass amplification and polynomial-shot success certificate}
\begin{proposition}[Constructive radial-mass lower bound]
\label{prop:constructive-radial-mass-lower-bound}
Let the total one-layer transition weight of the product mixer be given as (from \cref{eq:one-layer-weight}) 
\begin{equation}
q_n(\beta):=|a_0(\beta)|+(n-1)|a_1(\beta)|.
\label{eq:one-layer-weight2}
\end{equation}
For the complete-graph one-hot mixer, define
\(
\theta_n(\beta):=\frac{n\beta}{2(n-1)}.
\) For \(n\ge4\), this quantity is maximized by the  angle
\(
\beta_\star=\frac{\pi(n-1)}{n}.
\) If
\(
\beta_1=\cdots=\beta_p=\beta_\star,
\)
then
\begin{equation}
\label{eq:v0-star-exact}
v_0^{(p)}
=
\left(3-\frac{4}{n}\right)^{mp}.
\end{equation}

\end{proposition}

\begin{proof} See Appendix~\ref{app:proof-constructive-radial-mass}. \end{proof}

\begin{theorem}[Polynomial-shot success certificate]
\label{thm:polynomial-shot-success}
Fix $n\geq4$, $m\geq1$, and a target $y\in[n]^m$.
Consider depth-$p$ product-mixer dynamics initialized in
\[
\ket{s_0}=n^{-m/2}\sum_{x\in[n]^m}\ket{x}.
\]
Assume that all nonzero depth-$p$ path summands ending at $y$
lie in a common arc of angular width $\Theta_p<\pi$.
Let $L=L(n,m)\geq1$ and $K=K(n,m)\geq1$ be polynomially
bounded functions satisfying
\begin{equation}
\prod_{\ell=1}^{p}q_n(\beta_\ell)
\geq \sqrt{n}\,L^{-1/m},
\qquad
\cos(\Theta_p/2)\geq K^{-1}.
\label{eq:polynomial-shot-conditions}
\end{equation}
Then the single-shot target probability obeys
\begin{equation}
q_y:=
\bigl|\langle y\mid
\psi_p(\vec{\gamma},\vec{\beta})\rangle\bigr|^2
\geq \frac{1}{L^2K^2}.
\label{eq:polynomial-shot-probability}
\end{equation}
For any $\varepsilon\in(0,1)$, the shot count
\begin{equation}
N=\left\lceil L^2K^2\log(1/\varepsilon)\right\rceil
\label{eq:polynomial-shot-count}
\end{equation}
therefore suffices to observe $y$ at least once with
probability at least $1-\varepsilon$. For a constant mixer angle $\beta$ with $q_n(\beta)>1$,
the transport condition in
Eq.~\eqref{eq:polynomial-shot-conditions} is equivalent to
\begin{equation}
p\log q_n(\beta)
\geq
\frac12\log n-\frac{\log L}{m}.
\label{eq:polynomial-shot-depth}
\end{equation}
\end{theorem}

\begin{proof}
The absolute one-block mixer kernel has constant row sum
$q_n(\beta)$ (Lemma~\ref{lem:single-block-xy-amplitudes}).
Under the uniform initialization, the shell recursion in
Theorem~\ref{thm:exact-shell-transfer-recursion} gives the
total absolute path weight ending at any fixed target as
\begin{equation}
v_0^{(p)}
=
\left[\prod_{\ell=1}^{p}q_n(\beta_\ell)\right]^m.
\end{equation}
The transport condition therefore implies
$v_0^{(p)}\geq n^{m/2}/L$.
Applying Theorem~\ref{thm:phase-aligned-path-sum-lower-bound}
gives
\begin{equation}
q_y
\geq
\frac{\cos^2(\Theta_p/2)}{n^m}
\bigl(v_0^{(p)}\bigr)^2
\geq
\frac{\cos^2(\Theta_p/2)}{L^2}
\geq
\frac{1}{L^2K^2}.
\end{equation}
For $N$ independent shots, the probability of never observing
$y$ satisfies
\begin{equation}
(1-q_y)^N
\leq e^{-Nq_y}
\leq e^{-N/(L^2K^2)}.
\end{equation}
The choice in Eq.~\eqref{eq:polynomial-shot-count} makes
this probability at most $\varepsilon$.
\end{proof}

\begin{corollary}[Transport maximization and single  phase incompatibility]
\label{cor:maximal-transport-phase-obstruction}
Let $n\geq4$, $m\geq2$, and choose
\(
\beta_1=\cdots=\beta_p
=\beta_\star:=\frac{\pi(n-1)}{n}.
\)
Then \(
v_0^{(p)}=\left(3-\frac4n\right)^{mp},
\)
and the transport threshold $v_0^{(p)}\geq n^{m/2}$ is
reached whenever
\begin{equation}
p\geq
\left\lceil
\frac{\log n}{2\log(3-4/n)}
\right\rceil.
\label{eq:maximal-transport-depth}
\end{equation}
Nevertheless, for $p\geq2$, no choice of diagonal cost phases can place all depth-$p$ path summands ending at a fixed target in one common open semicircle. Thus, for $p\geq2$, this transport-maximizing choice places the single-cone bound underlying Theorems ~\ref{thm:phase-aligned-path-sum-lower-bound} and ~\ref{thm:polynomial-shot-success} in its inconclusive regime.
\end{corollary}

\begin{proof} See Appendix~\ref{app:proof-constructive-radial-mass}. \footnote{The proof relies on a counterexample with $p=2$ and four angular directions. This counterexample was suggested by ChatGPT 5.6 Sol; Appendix~\ref{app:proof-constructive-radial-mass} contains the
independently checked proof} \end{proof}

Corollary~\ref{cor:maximal-transport-phase-obstruction} shows that harvesting the maximally accumulated path mass requires phase engineering beyond the single-cone condition discussed above.
If such engineering yields a target probability bounded below by a positive constant, it strengthens polynomial-shot success to
constant-shot success at fixed confidence. This strengthening,
however, is not required for quantum advantage: for the NP-hard
optimization problems considered here, polynomial-shot sampling
of an optimum already provides a route to such an advantage
when parameter selection and circuit execution are also efficient.
The polynomial-shot certificate therefore supplies a meaningful
performance target, while the abundant absolute path mass
motivates further work on phase engineering.

Theorem~\ref{thm:polynomial-shot-success} makes polynomial-shot
success the design objective. It allows arbitrary mixer-angle
schedules and requires sufficient absolute transport together
with a quantitative phase-alignment margin.
Corollary~\ref{cor:maximal-transport-phase-obstruction} illustrates why
these requirements must be considered jointly. Once primitive mixing
supplies connectivity and the accumulated path mass meets the required
threshold, the design priority shifts toward phase alignment,
with submaximal mixer angles providing additional freedom to coordinate
transport with the cost phases and retain enough coherent amplitude for
polynomial-shot success.  Illustrative examples of success beyond the common-cone regime, together with the distinction between constant-shot and polynomial-shot sampling requirements, are discussed in Appendix~\ref{app:sampling-complexity}.

Finally, it is worth noting that exact lattice normalization arises when the relevant cost and penalty data are integer-valued and polynomially bounded. This condition is satisfied by the penalty terms in Definition~\ref{def:kernel-requirement}, whereas the problem-specific objective term need not satisfy it in general. Prior CE--QAOA analyses in Refs.~\cite{onahfund,onahfinite} show, through Riemann--Lebesgue averaging, that guarantees previously derived under lattice normalization extend beyond exact lattice spectra.

\section{Example applications of the framework}
\label{sec:applications}

Connecting the shell reduction to quantum sampling tasks in constrained optimization requires a probabilistic interpretation of the mixer-induced transport. The technical results collected in Appendix~\ref{app:measure-shell-reduction} recast the shell reduction as a finite probability model. This interpretation turns the absolute mixer
envelope into an object that can be studied using classical probability, time-inhomogeneous Markov chains, lumpability, drift, variance, and concentration tools. Recall that Theorem~\ref{thm:exact-shell-transfer-recursion} gives the unnormalized absolute path-weight recursion
\[
v_r^{(\ell)}
=
\sum_{t=0}^{m}T_{r,t}(\beta_\ell)v_t^{(\ell-1)}.
\]
After normalization by the one-layer product mass \(q_n(\beta)^m\), the same
coefficients define
\[
P_{r,t}^{(\beta)}
=
\frac{T_{r,t}(\beta)}{q_n(\beta)^m}.
\]
Corollary~\ref{cor:normalized-shell-kernel-stochastic} proves that \(P^{(\beta)}\) is a stochastic shell kernel.
Proposition~\ref{prop:shell-transfer-product-pushforward} identifies this kernel as the radial pushforward of a coordinatewise product probability kernel, while Theorem~\ref{thm:depth-p-radial-markov} gives the corresponding depth-\(p\) time-inhomogeneous radial Markov recursion. This normalized process uses the quantities
\[
\bigl|\langle z|U_M(\beta)|x\rangle\bigr|
\]
to define the \(L^1\)-normalized probabilistic envelope of the absolute mixer path expansion, thereby turning the shell coefficients into a finite stochastic object. Classical probability and Markov-chain tools can therefore be applied directly to the radial envelope of the mixer path expansion.

This viewpoint yields several consequences for quantum algorithms designed on product spaces. First, the normalized shell process clarifies the onset mechanism for quantum sampling advantage in the phase-aligned regime. Second,
finite-size transport diagnostics show how different \(XY\)-mixer geometries select different reduced statistics. Third, measured bitstrings can be further analyzed through problem-dependent classical statistics to obtain practical
probes of mixing, leakage, topology effects, phase effects, convergence, and application-level constraint structure.

\subsection{The onset of quantum sampling advantage}
\label{subsec:probabilistic-normalized-shell-kernel}
For the complete-graph block-\(XY\) mixer, Corollary~\ref{cor:two-species-poisson-binomial}
gives an explicit one-step law.  Fix \(y\in \mathcal X=[n]^m\), and suppose
\(R_\ell=d(Z^{(\ell)},y)=t\).  Then
\[
R_{\ell+1}
\;\stackrel{d}{=}\;
\mathrm{Bin}\bigl(m-t,p_+(\beta_{\ell+1})\bigr)
+
\mathrm{Bin}\bigl(t,1-p_-(\beta_{\ell+1})\bigr),
\]
where
\[
p_+(\beta)
=
\frac{(n-1)|a_1(\beta)|}{q_n(\beta)},
\qquad
p_-(\beta)
=
\frac{|a_1(\beta)|}{q_n(\beta)}.
\]
Hence
\[
\mathbb E[R_{\ell+1}\mid R_\ell=t]
=
m p_+(\beta_{\ell+1})
+
\bigl(1-p_+(\beta_{\ell+1})-p_-(\beta_{\ell+1})\bigr)t,
\]
and
\[
\mathrm{Var}(R_{\ell+1}\mid R_\ell=t)
=
(m-t)p_+(\beta_{\ell+1})\bigl(1-p_+(\beta_{\ell+1})\bigr)
+
t\,p_-(\beta_{\ell+1})\bigl(1-p_-(\beta_{\ell+1})\bigr).
\]
These formulas give closed drift and variance laws for the normalized radial envelope. For a fixed angle \(\beta\) with \(|a_1(\beta)|>0\), the affine drift has the
fixed point
\[
r_\infty
=
\frac{m p_+(\beta)}{p_+(\beta)+p_-(\beta)}
=
m\frac{n-1}{n}.
\]
The stationary one-coordinate mismatch probability is
\[
\pi_{\mathrm{mis}}
=
\frac{p_+(\beta)}{p_+(\beta)+p_-(\beta)}
=
\frac{n-1}{n},
\]
and the stationary radial law is
\[
\rho_\infty(r)
=
\binom{m}{r}
\left(\frac{n-1}{n}\right)^r
\left(\frac{1}{n}\right)^{m-r}.
\]
However, this is the Hamming-shell law of a uniformly random configuration relative to
\(y\).  The normalized complete-graph envelope relaxes toward the product-space bulk. At the constructive angle
\[
\beta_\star=\frac{\pi(n-1)}{n},
\]
Proposition~\ref{prop:constructive-radial-mass-lower-bound} gives
\[
|a_0(\beta_\star)|=\frac{n-2}{n},
\qquad
|a_1(\beta_\star)|=\frac{2}{n},
\qquad
q_n(\beta_\star)=3-\frac{4}{n}.
\]
Thus
\[
p_+(\beta_\star)
=
\frac{2(n-1)}{3n-4},
\qquad
p_-(\beta_\star)
=
\frac{2}{3n-4}.
\]
Starting from \(R_\ell=0\), the normalized radial envelope satisfies
\[
\mathbb E[R_{\ell+1}\mid R_\ell=0]
=
m\frac{2(n-1)}{3n-4}
\sim
\frac{2m}{3}.
\]
The same angle that maximizes the unnormalized one-layer growth factor
\(q_n(\beta)^m\) drives the normalized radial envelope toward the bulk. For the uniform initialization used here, this stationary law
is already attained at depth zero and remains unchanged under
the normalized envelope dynamics.

This exposes a necessary mechanism for quantum sampling advantage in constrained quantum optimization algorithms designed on product spaces and shows how to engineer for it. The
normalized radial chain describes the product-space shape of the absolute
envelope, whose complete-graph fixed point lies in the ordinary bulk.  The
factor \(q_n(\beta)^m\) controls the total unnormalized \(\ell_1\)-path mass
available to the quantum path expansion.  Advantage arises when the phase
structure converts this available mass into Born probability on useful
configurations.  The lattice-normalized phase-alignment regime in
Theorem~\ref{thm:phase-aligned-path-sum-lower-bound} suggests an engineering criterion for this conversion.  More general instance-dependent phase structures, cost landscapes, and mixer geometries can be analyzed as mechanisms for harvesting the same available path mass or as certificates that phase cancellation prevents such harvesting. We leave further phase engineering for future work.


\subsection{Finite-size transport diagnostics for XY mixer geometries}
\label{subsec:numerical-transport-diagnostics}
\footnote{The numerical data and supporting script used in this work are publicly available on Zenodo~\cite{Michielsen2026GeometryInterferenceData}}

We now examine how the transport layer changes when the one-hot block is preserved and the local XY topology is varied. We compare complete, cycle, path, star, and sparse chordal-cycle block geometries. For a local interaction graph \(G\) on \([n]\), with adjacency matrix \(A_G\), we use the normalized one-block one-excitation-sector mixer
\[
        U_G(\beta)=\exp\!\left(-i\beta A_G/\|A_G\|_2\right).
\]
The first diagnostic is the one-site absolute transfer factor
\[
        q_G(\beta;y_0)
        =
        \sum_{z\in[n]}
        \left|\langle z|U_G(\beta)|y_0\rangle\right|.
\]
For \(G=K_n\), this quantity is the analytical expression in \cref{eq:one-layer-weight}. It is independent of the reference symbol \(y_0\) and equals
\[
        q_n(\beta)=|a_0(\beta)|+(n-1)|a_1(\beta)|.
\]
The analytic theory gives \(q_n(\beta_\star)=3-4/n\) at the constructive
complete-graph angle. Figure~\ref{fig:mixer-q-factor} shows that the same
encoded one-hot alphabet supports different transfer envelopes under different
XY geometries. One readily observes that amplitude contribution growth is a mixer-geometric property. A change in \(G\) changes the amount of absolute mass made available prior to interference.

The second diagnostic tests closure of a proposed reduced statistic under the
absolute transfer kernel. For a partition
\(\mathcal P=\{C_0,C_1,\ldots\}\) of the local alphabet, define
\[
\Delta_{\mathcal P}(G,\beta)
=
\max_{C_i,C_j}
\max_{a,b\in C_i}
\left|
\sum_{z\in C_j}|\langle z|U_G(\beta)|a\rangle|
-
\sum_{z\in C_j}|\langle z|U_G(\beta)|b\rangle|
\right|.
\]
This quantity is the lumpability defect. The condition
\(\Delta_{\mathcal P}=0\) means that all local states inside the same class
\(C_i\) send the same total absolute transfer into every class \(C_j\). The
partition then defines an exact reduced transport model. A positive value of
\(\Delta_{\mathcal P}\) measures the amount of local geometry hidden by the
chosen shell variable.

Figure~\ref{fig:mixer-lumpability-defect} gives the main structural lesson of this diagnostic. The complete graph has zero defect for the Hamming partition. This confirms the exact Hamming-shell recursion in our analysis. Path, cycle, star, and sparse geometries show residual defect under ordinary Hamming shells. Topology-refined partitions reduce the defect by using the geometry of the local graph. The defect therefore acts as a numerical test for a proposed radial variable. A large defect identifies missing transport geometry. A small or zero defect identifies an appropriate reduced statistic.

The third diagnostic connects reduced transport geometry to sampling behavior.
We apply one radial phase-marking layer on each block, with local phase
\(\exp(-i\gamma \mathbf 1_{\{z\ne y_0\}})\), followed by
\(U_G(\beta_\star)\), where \(\beta_\star\) maximizes \(q_G(\beta;y_0)\) on the scanned interval. For \(m\) independent product blocks, this induces a shell law on \(r=d(x,y)\). Figure~\ref{fig:mixer-shell-sampling-profiles} shows that the same radial phase marking produces different Hamming-shell distributions for different mixer topologies because the mixer topology determines which neighborhoods of the reference configuration receive mass efficiently. These numerical results show that mixer geometry determines the available transfer mass, the correct reduced statistic, and the resulting shell sampling behavior on the encoded product manifold.

\begin{figure*}[t]
    \centering
    \begin{subfigure}[t]{0.32\textwidth}
        \centering
        \includegraphics[width=\linewidth]{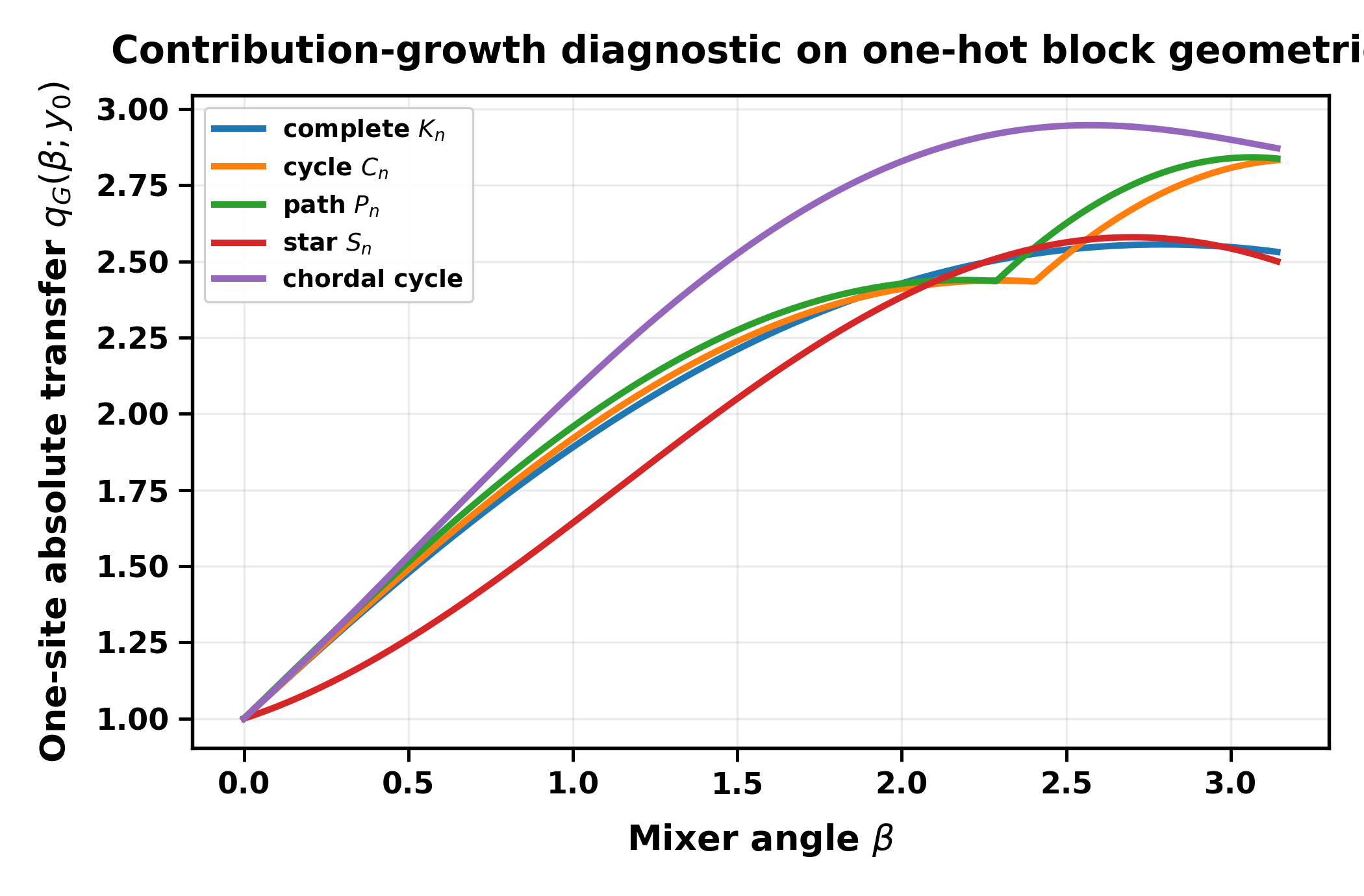}
        \caption{One-site absolute transfer factor \(q_G(\beta;y_0)\).}
        \label{fig:mixer-q-factor}
    \end{subfigure}
    \hfill
    \begin{subfigure}[t]{0.32\textwidth}
        \centering
        \includegraphics[width=\linewidth]{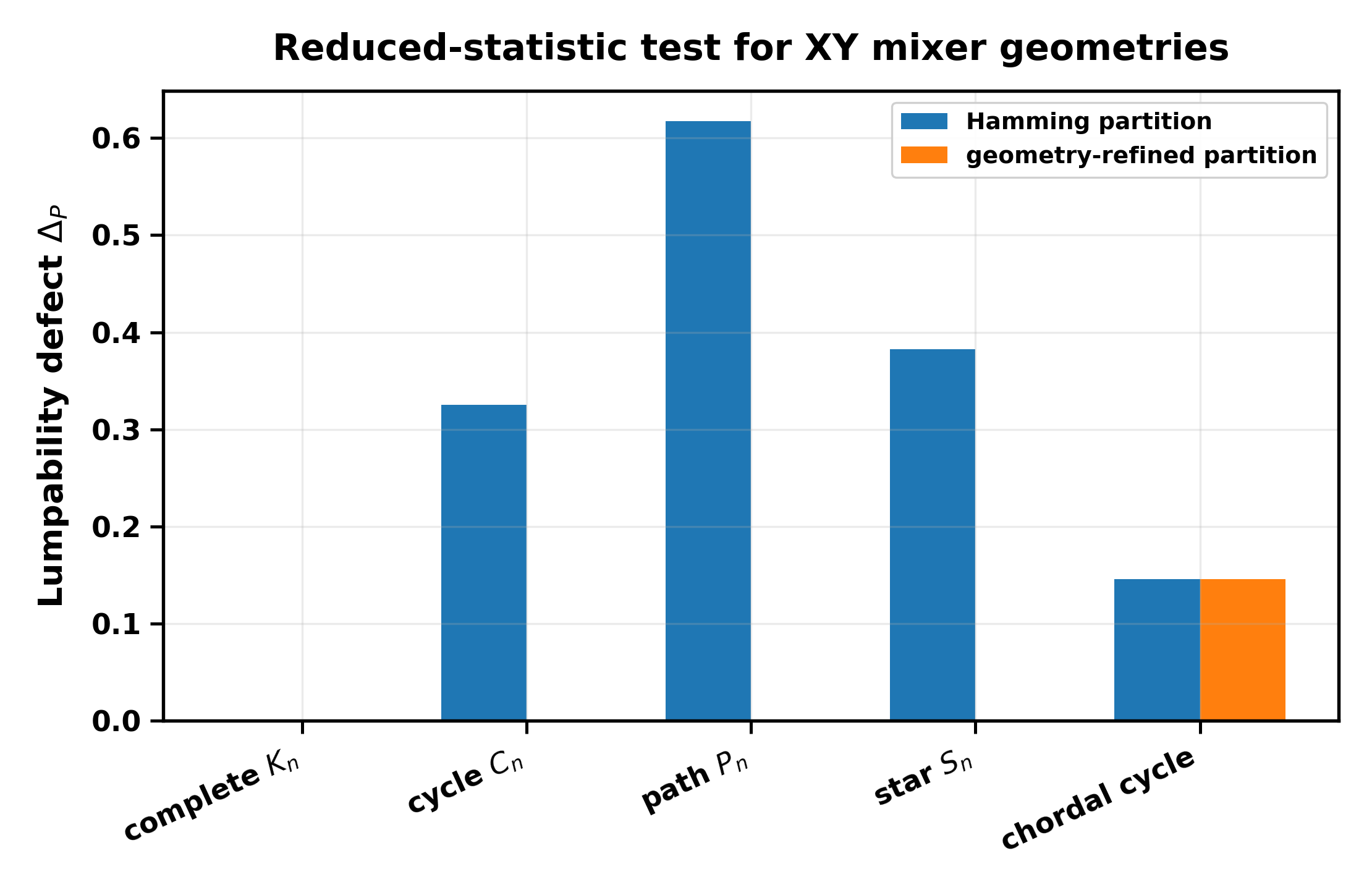}
        \caption{Lumpability defect for Hamming and geometry-refined partitions.}
        \label{fig:mixer-lumpability-defect}
    \end{subfigure}
    \hfill
    \begin{subfigure}[t]{0.32\textwidth}
        \centering
        \includegraphics[width=\linewidth]{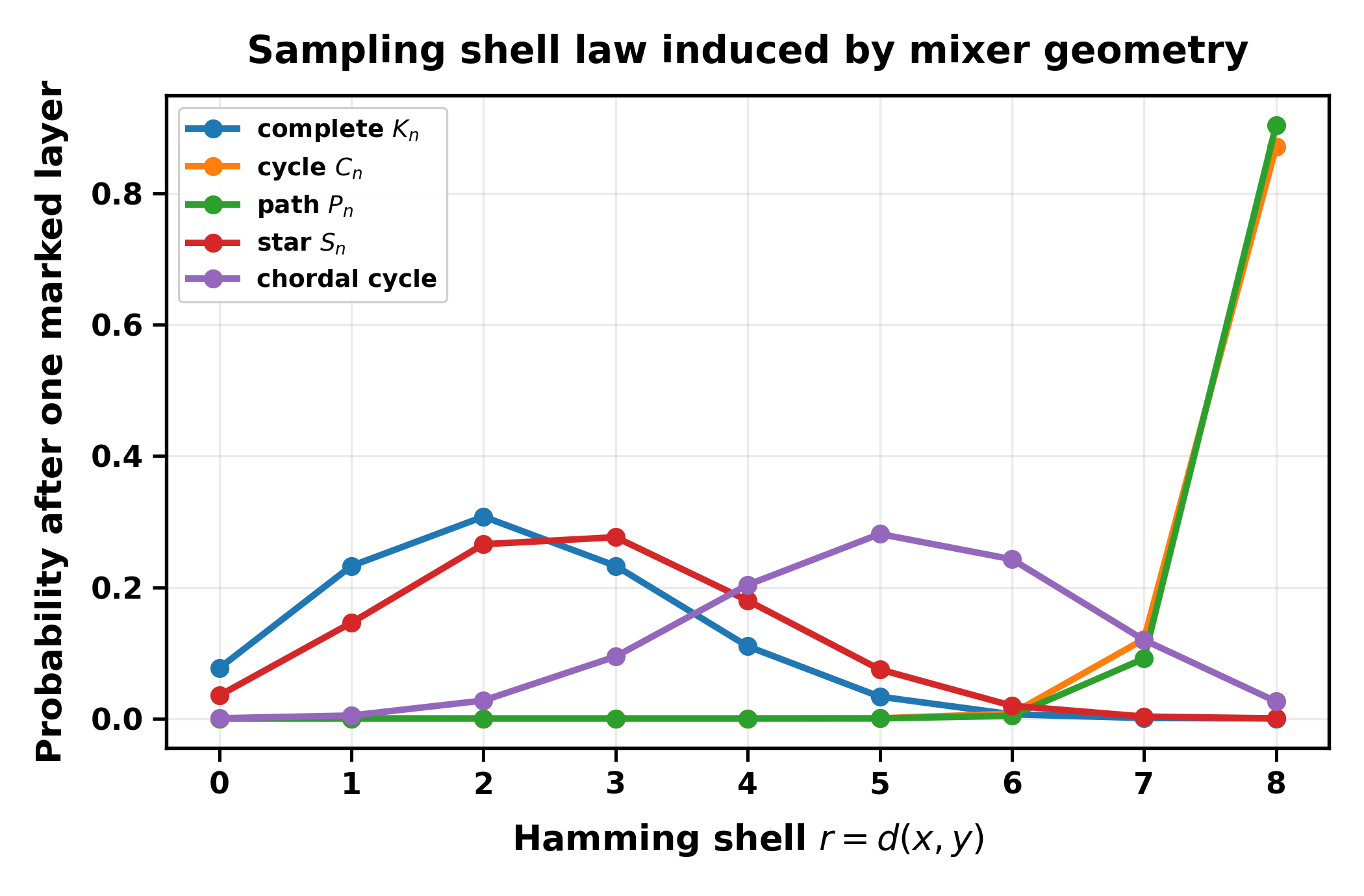}
        \caption{Shell sampling profiles after one radial phase-marking layer.}
        \label{fig:mixer-shell-sampling-profiles}
    \end{subfigure}
    \caption{Finite-size transport diagnostics for XY mixer geometries on the OFM product manifold. The local alphabet size is \(n=9\), the shell-profile experiment uses \(m=8\) product blocks, and each mixer is normalized by the spectral norm of its one-block adjacency matrix. The left panel shows that mixer geometry changes the available one-site absolute transfer mass. The middle panel shows that complete-graph XY closes exactly on Hamming shells, while other geometries select refined reduced statistics such as cyclic distance, path distance, or graph-distance data. The right panel uses a reference-dependent radial marker, but no problem-specific objective or post-sampling map.  The right panel shows that these transport differences already produce distinct shell sampling laws before any problem-specific objective phase or post sampling map is introduced. In addition, one observes that the complete-graph curve is not the largest transfer envelope over the full angle range but it has zero Hamming-shell lumpability defect. The complete-graph curve has zero Hamming-shell lumpability defect and gives the exact shell-transfer law used in the analysis.  Some sparse or boundary-sensitive geometries exhibit larger transfer envelopes over parts of the angle range but suffer from defects that limit the efficiency of amplitude transfer. Some sparse or boundary-sensitive geometries have larger absolute envelopes at selected angles, but their nonzero Hamming defects show that this particular radial reduction is inexact.  }
    \label{fig:xy-mixer-transport-diagnostics}
\end{figure*}

\subsection{Shell-transport preservation as a transpilation diagnostic}
\label{subsec:shell-transport-transpilation}

Here we introduce shell-transport preservation as a problem-structure-aware
transpilation diagnostic motivated by the need to retain or track the algorithmic guarantees derived analytically from circuit implementations. This line of application combines insights from symmetry reductions of quantum dynamics with application-aware compilation methods. This is important because unitary approximations during gate synthesis can interfere with symmetry driven guarantees, such as the ones derived in Sec. \ref{sec:results}. 

Symmetry reductions of quantum dynamics have been studied through quotient graphs, where an invariant subspace supports an effective quantum walk on a smaller graph and can govern properties such as hitting behavior
\cite{KroviBrun2007Quotient}. Related ideas arise in quantum sampling through
lumpable quantum channels, for which transition probabilities depend only on
the equivalence classes defining the reduced dynamics
\cite{Franca2018PerfectSampling}. In parallel, application-aware and
algorithm-oriented compilation methods have emphasized that circuit
optimization should exploit the structure of the algorithm and preserve the
quality of its output alongside generic circuit-level metrics
\cite{Quetschlich2024ApplicationAware,Ji2025AOQMAP,Zhu2024Coqa}.

The key observation is that the finite-size diagnostics introduced in
Section~\ref{subsec:numerical-transport-diagnostics} also provide
algorithm-specific criteria for assessing a compiled mixer. In addition, these transport diagnostics remain tractable as the problem size grows and therefore provide operational quantities for circuit synthesis and execution. For a proposed Trotter depth, matching order, and routing strategy, the local kernel
\(\widehat K_s\), the associated quotient defect, and the resulting
shell-transfer bounds can be computed classically before the circuit is
executed. Residual distortions resulting from the synthesis can be used to guide shot budgets or modify the confidence interval in the underlying guarantees.

To consider a transpiled mixer unitary, let
\[
\mathcal E(K_n)
=
\mathcal M_1\sqcup\cdots\sqcup
\mathcal M_{\chi'(K_n)}
\]
be an edge-coloring of the complete mixer graph into disjoint matchings.
Exchange operations within a matching act on disjoint qubit pairs and can
therefore be implemented in parallel. For \(s\) first-order Trotter steps, the
compiled mixer on block \(b\) is
\begin{equation}
U_{M,s}^{(b)}(\beta_\star)
=
\left[
\prod_{c=1}^{\chi'(K_n)}
\prod_{(u,v)\in\mathcal M_c}
\exp\!\left[
-\frac{i\beta_\star}{2s(n-1)}
\left(
X_u^{(b)}X_v^{(b)}
+
Y_u^{(b)}Y_v^{(b)}
\right)
\right]
\right]^s,
\label{eq:trotterized-block-mixer}
\end{equation}
where the matching layers are applied in a fixed order. A finite Trotter
depth preserves the encoded one-hot subspace while potentially breaking the
permutation symmetry responsible for exact Hamming-shell lumpability. We
assess this effect using the normalized absolute transfer kernel. For local
symbols \(a,c\in[n]\), define
\begin{align}
\widehat Q_s(c)
&=
\sum_{a\in[n]}
\left|
\langle e_a|
U_{M,s}^{(b)}(\beta_\star)
|e_c\rangle
\right|,
\\
\widehat K_s(a\mid c)
&=
\frac{
\left|
\langle e_a|
U_{M,s}^{(b)}(\beta_\star)
|e_c\rangle
\right|
}{
\widehat Q_s(c)
}.
\label{eq:trotter-local-normalized-kernel}
\end{align}
Because the compiled mixer factorizes across blocks, the normalized product
kernel is
\begin{equation}
K_s(z\mid x)
=
\prod_{b=1}^{m}
\widehat K_s(z_b\mid x_b),
\qquad
x,z\in \mathcal \mathcal X.
\label{eq:trotter-product-normalized-kernel}
\end{equation}
For a reference configuration \(y\), the shell profile induced by a source
configuration \(x\) is
\begin{equation}
P_s(r\mid x)
=
\sum_{z\in S_r(y)}
K_s(z\mid x),
\qquad
r=0,\ldots,m.
\label{eq:trotter-shell-profile}
\end{equation}

For the ideal complete-graph mixer, all source configurations in the same
shell \(S_t(y)\) induce the same shell-transfer law. The departure from this
property is measured by the Trotterization-induced lumpability defect ($\Delta_{\mathrm{lump}}(s)$) and shell-kernel error ($\Delta_{\mathrm{shell}} (s)$) where the former is given as 
\begin{equation}
\Delta_{\mathrm{lump}}(s)
=
\max_{0\leq t\leq m}
\max_{x,x'\in S_t(y)}
\frac{1}{2}
\sum_{r=0}^{m}
\left|
P_s(r\mid x)
-
P_s(r\mid x')
\right|.
\label{eq:trotter-lumpability-defect}
\end{equation}
Approximate lumpability alone does not ensure that the compiled mixer
reproduces the intended shell-transition law. We therefore also define the
shell-kernel error
\begin{equation}
\Delta_{\mathrm{shell}}(s)
=
\max_{0\leq t\leq m}
\max_{x\in S_t(y)}
\frac{1}{2}
\sum_{r=0}^{m}
\left|
P_s(r\mid x)
-
P_{r,t}^{(\beta_\star)}
\right|,
\label{eq:trotter-shell-error}
\end{equation}
where \(P_{r,t}^{(\beta_\star)}\) is the normalized ideal shell kernel derived
in Appendix~\ref{app:measure-shell-reduction}. These two quantities
distinguish preservation of the shell partition from agreement with the correct transport law. For the ideal complete-graph mixer, the normalized one-block kernel depends
only on whether the local input and output symbols agree,
\begin{equation}
\widehat K_{\mathrm{ideal}}(a\mid c)
=
\begin{cases}
d_\star, & a=c,\\
o_\star, & a\neq c.
\end{cases}
\label{eq:ideal-local-two-class-kernel}
\end{equation}
This equality-versus-inequality quotient is the local symmetry that produces
exact Hamming-shell lumpability on the product space. Its violation under
Trotterization is quantified by
\begin{align}
\delta_{\mathrm{diag}}(s)
&=
\max_a \widehat K_s(a\mid a)
-
\min_a \widehat K_s(a\mid a),
\\
\delta_{\mathrm{off}}(s)
&=
\max_{a\neq c}\widehat K_s(a\mid c)
-
\min_{a\neq c}\widehat K_s(a\mid c),
\\
\delta_{\mathrm{quot}}(s)
&=
\max\!\left\{
\delta_{\mathrm{diag}}(s),
\delta_{\mathrm{off}}(s)
\right\}.
\label{eq:local-quotient-defect}
\end{align}
The exact local quotient defect vanishes when the compiled normalized kernel retains the two-class structure required by the ideal Hamming-shell reduction. To control agreement with the ideal shell kernel, define
\begin{align}
\epsilon_{\mathrm{diag}}(s)
&=
\max_a
\left|
\widehat K_s(a\mid a)-d_\star
\right|,
\\
\epsilon_{\mathrm{off}}(s)
&=
\max_{a\neq c}
\left|
\widehat K_s(a\mid c)-o_\star
\right|.
\end{align}

\begin{proposition}[Local certification of product-space transport]
\label{prop:local-product-transport-bounds}
For every source shell \(S_t(y)\),
\begin{equation}
\Delta_{\mathrm{lump}}^{(t)}(s)
\leq
\min\!\left\{
1,\,
(m-t)\delta_{\mathrm{diag}}(s)
+
t\,\delta_{\mathrm{off}}(s)
\right\},
\label{eq:lumpability-local-upper-bound}
\end{equation}
and
\begin{equation}
\Delta_{\mathrm{shell}}^{(t)}(s)
\leq
\min\!\left\{
1,\,
(m-t)\epsilon_{\mathrm{diag}}(s)
+
t\,\epsilon_{\mathrm{off}}(s)
\right\}.
\label{eq:shell-kernel-local-upper-bound}
\end{equation}
\end{proposition}

The proof follows by coupling the coordinatewise Bernoulli shell increments
and is given in
Appendix~\ref{app:trotter-transport-certification}.

This result provides a scalable certificate because all four local quantities are obtained from an \(n\times n\) one-block kernel, without constructing the full \(n^m\)-dimensional product kernel. For a prescribed tolerance \(\varepsilon>0\), let

\begin{align}
\overline{\Delta}_{\mathrm{lump}}(s)
&=
\max_{0\leq t\leq m}
\min\!\left\{
1,\,
(m-t)\delta_{\mathrm{diag}}(s)
+
t\,\delta_{\mathrm{off}}(s)
\right\},
\\
\overline{\Delta}_{\mathrm{shell}}(s)
&=
\max_{0\leq t\leq m}
\min\!\left\{
1,\,
(m-t)\epsilon_{\mathrm{diag}}(s)
+
t\,\epsilon_{\mathrm{off}}(s)
\right\}.
\end{align}
The certified Trotter step count is
\begin{equation}
s_{\mathrm{cert}}(\varepsilon)
=
\min\left\{
s\in\mathbb N:
\overline{\Delta}_{\mathrm{lump}}(s)\leq\varepsilon
\ \text{and}\
\overline{\Delta}_{\mathrm{shell}}(s)\leq\varepsilon
\right\},
\label{eq:certified-trotter-step}
\end{equation}
with corresponding logical exchange depth
\begin{equation}
D_{\mathrm{cert}}(\varepsilon)
=
s_{\mathrm{cert}}(\varepsilon)\chi'(K_n),
\qquad
\chi'(K_n)
=
\begin{cases}
n-1, & n\ \text{even},\\
n, & n\ \text{odd}.
\end{cases}
\label{eq:certified-trotter-depth}
\end{equation}
The gate-count expression and the details of the certification
procedure are collected in Appendix~\ref{app:trotter-transport-certification}.


For the numerical investigation, we consider the position-based TSP encoding
with \(m=n\) one-hot blocks. The reference \(y\in[n]^n\) is a valid
permutation, while the source configurations range over the full row-one-hot
product space \(\mathcal X=[n]^n\). We evaluate
\(n=m=9,\ldots,25\) at
\(\beta_\star=\pi(n-1)/n\). For \(s=1,\ldots,10\), we compute the exact local
quotient quantities and evaluate arbitrary and extremal sources from every Hamming shell. The complete numerical protocol is given in Appendix~\ref{app:trotter-transport-certification}. Table~\ref{tab:full-product-trotter-summary} below shows that a Trotterized mixer can approximate the ideal unitary accurately while retaining substantial defects in the reduced transport law. 
At \(s=10\), this fidelity-based Trotter error lies between
\(2.11\times10^{-3}\) and \(2.60\times10^{-3}\) over
\(n=9,\ldots,25\). In contrast, the empirical shell-kernel defects reported
in the table range from \(4.67\times10^{-2}\) to \(7.06\times10^{-2}\),
while the empirical lumpability defects range from \(9.20\times10^{-2}\) to
\(1.40\times10^{-1}\). Consequently, a product-formula realization that
appears accurate according to conventional process fidelity can still
substantially distort the quotient transport structure entering the
analytical sampling guarantee. The shell-transport criterion therefore detects residual algorithmically relevant structure that is hidden by a generic fidelity measure. Extensions to
connectivity-dependent routing, native-gate compilation, and noisy execution are left for future work.

\begin{table}[H]
\centering
\caption{
Full-product-space Trotter transport diagnostics for valid-tour references
with \(n=m=9,\ldots,25\). The empirical defects at \(s=10\) maximize over
sampled arbitrary sources in the full row-one-hot product space and are lower
bounds on the worst-case shell defects. The exact local quotient defect
measures the breaking of the equality-versus-inequality local kernel.
The certified values \(s_{\mathrm{cert}}\) and \(D_{\mathrm{cert}}\) are the
minimum Trotter steps and logical exchange depth for which rigorous product-shell upper bounds on both defects are at most \(10^{-2}\).
}
\label{tab:full-product-trotter-summary}
\small
\begin{tabular}{cccccc}
\toprule
\textbf{\(n=m\)} &
\textbf{\(\widehat\Delta_{\mathrm{lump}}(10)\)} &
\textbf{\(\widehat\Delta_{\mathrm{shell}}(10)\)} &
\textbf{\(\delta_{\mathrm{quot}}(10)\)} &
\textbf{\(s_{\mathrm{cert}}\)} &
\textbf{\(D_{\mathrm{cert}}\)} \\
\midrule
9 & \num{9.196e-02} & \num{4.670e-02} & \num{2.486e-02} & 225 & 2025 \\
10 & \num{9.927e-02} & \num{5.038e-02} & \num{2.273e-02} & 228 & 2052 \\
11 & \num{1.056e-01} & \num{5.353e-02} & \num{2.219e-02} & 245 & 2695 \\
12 & \num{1.106e-01} & \num{5.605e-02} & \num{2.052e-02} & 247 & 2717 \\
13 & \num{1.151e-01} & \num{5.826e-02} & \num{1.983e-02} & 259 & 3367 \\
14 & \num{1.187e-01} & \num{6.008e-02} & \num{1.852e-02} & 260 & 3380 \\
15 & \num{1.220e-01} & \num{6.171e-02} & \num{1.784e-02} & 269 & 4035 \\
20 & \num{1.334e-01} & \num{6.728e-02} & \num{1.408e-02} & 282 & 5358 \\
25 & \num{1.402e-01} & \num{7.061e-02} & \num{1.171e-02} & 294 & 7350 \\
\bottomrule
\end{tabular}
\end{table}

\subsection{Problem-dependent classical maps, hardware probes and feasibility repair}
\label{subsec:multidimensional-jump-statistics}

After measurement, the quantum circuit produces a classical distribution over
the encoded product space \(\mathcal X\). Problem-dependent classical processing of
these samples forms a third layer of the hybrid algorithm, following mixer
transport and phase-sensitive interference. This layer can expose the
constraint structure present in the measured distribution, guide
feasibility-restoring operations, and separate the contribution of the
quantum sampler from that of the subsequent classical processing. 
This is especially relevant for the classical processing of quantum-generated samples from near-term devices, where measured samples may be biased, noisy, or infeasible~\cite{Maciejewski_2025,hadfield2026ndar,Maciejewski2026}.

Without loss of generality, let
\[
x^{(1)},\ldots,x^{(N)}\in \mathcal X
\]
be samples drawn from the measured Born distribution, and let
\[
T:\mathcal X\longrightarrow\mathcal T
\]
be a problem-dependent classical map. Its value may record Hamming distance, penalty level, collision structure, resource usage, route loads, subtours,
machine loads, color conflicts, or missing and repeated assignments. The
corresponding empirical distribution is
\[
\widehat p_T(\tau)
=
\frac{1}{N}
\sum_{j=1}^{N}
\mathbf 1\!\left\{T\!\left(x^{(j)}\right)=\tau\right\},
\qquad
\tau\in\mathcal T.
\]
This distribution reveals which combinatorial features of the optimization
problem are represented in the quantum-generated samples.

The penalty Hamiltonian in
Definition~\ref{def:kernel-requirement} provides a natural source of such maps. If
\begin{equation}
H_{\mathrm{pen}}(x)=F(T(x)),
\label{eq:scheme-alignment-present}
\end{equation}
then \(T\) resolves the combinatorial structure underlying each penalty
value. Configurations with the same total penalty may therefore be
distinguished according to the particular constraint violations that produced
it. The algebraic structure of this relation, including the decomposition of
penalty levels into \(T\)-fibres, multidimensional contingency-table
statistics, and the connection to coding-theoretic methods, is developed in
Appendix~\ref{app:constraint-statistics-coding}.

This additional resolution is especially useful for feasibility repair. Let
\[
\mathcal F\subseteq \mathcal X
\]
be the feasible set and let
\[
\mathcal R:\mathcal X\longrightarrow\mathcal F
\]
be a deterministic repair operation satisfying
\[
\mathcal R(x)=x,
\qquad x\in\mathcal F.
\]
In many applications, the correction applied to a sample depends on its
multidimensional constraint statistic:
\[
\mathcal R(x)=\mathcal R_{T(x)}(x).
\]
The value \(T(x)\) identifies the violation pattern, while
\(\mathcal R_{T(x)}\) specifies the corresponding correction. Routing repair
may depend on visit multiplicities, subtour indicators, and route loads.
Scheduling repair may depend on unassigned jobs, machine overloads, and
time-slot conflicts. Capacity-constrained allocation may require the complete
vector of resource excesses rather than the total penalty alone.

The distribution after repair is
\[
\widehat p_{\mathcal R}(z)
=
\frac{1}{N}
\sum_{j=1}^{N}
\mathbf 1\!\left\{
\mathcal R\!\left(x^{(j)}\right)=z
\right\},
\qquad z\in\mathcal F.
\]
For a target set \(A\subseteq\mathcal F\), the raw and repaired success
probabilities are
\[
P_{\mathrm{raw}}(A)
=
\sum_{x\in A}p(x),
\qquad
P_{\mathrm{repair}}(A)
=
\sum_{x:\,\mathcal R(x)\in A}p(x).
\]
Their difference quantifies the contribution of classical repair to the
observed success probability. The correction distance, reduction in penalty,
and change in objective value can likewise be used to determine whether the
repair performs a local correction of a promising quantum sample or a more
substantial reconstruction.

The layered formalism makes this attribution transparent. The shell-transfer
analysis describes the mass generated by the mixer. The phase analysis
determines how that mass contributes to the measured Born distribution. The
classical map \(T\) exposes the problem structure contained in the samples,
and the repair map \(\mathcal R\) determines how infeasible samples are
converted into feasible solutions. Improvements created by repair can
therefore be distinguished from improvements produced directly by the
quantum circuit.

This distinction is particularly important for near-term quantum devices.
Noise, finite sampling, imperfect phase control, and approximate mixer
implementations can produce infeasible samples that nevertheless retain
useful combinatorial structure. Multidimensional statistics identify which
samples are close to feasibility, which violations are readily correctable,
and which parts of the quantum output can be preserved during repair.
Discarding all infeasible samples would remove this information.

The same empirical distributions can also provide problem-dependent hardware probes. Comparisons across ideal simulation, Trotterized circuits, and noisy execution can reveal changes in collision profiles, load imbalances, subtour
frequencies, constraint violations, and repair costs that are invisible to a single objective value or feasibility rate. They therefore connect circuit behavior directly to the application structure represented in the measured samples. Finally, the multidimensional statistics introduced here provide a route for importing powerful tools from coding theory and algebraic combinatorics into the analysis of constrained quantum optimization. They enlarge the class of product-space structures that can be studied through distance distributions, orthogonal polynomial expansions, weight enumerators, association schemes, and related methods. We leave detailed application-specific studies of the resulting repair strategies and hybrid quantum--classical workflows to future work.

\section{Discussion}
\label{sec:discussion}


The central structural insight of this work is that quantum dynamics on a large encoded product space can admit a finite shell-resolved transport description. For the complete-graph block-\(XY\) mixer, the local transition
amplitude depends only on whether two symbols agree or differ. Together with the coordinatewise factorization of the mixer, this symmetry reduces the absolute transport on \(\mathcal X=[n]^m\) to the \(m+1\) Hamming shells around a reference configuration. Proposition
\ref{prop:shell-transfer-generating} gives the one-layer transfer law, while
Theorem \ref{thm:exact-shell-transfer-recursion} propagates it through arbitrary circuit depth. This gives a broader local-to-global interpretation of product-space quantum
dynamics. The local kernel determines the elementary moves available within each block, while the product structure combines these moves into global transport across the encoded space. In general, the symmetry of the mixer determines
which reduced statistic describes this transport. Complete-graph mixing
selects ordinary Hamming distance, while path, cycle, and nonhomogeneous
mixers generally require path distance, cyclic distance, boundary data, or
other refined statistics.

However, the problem structure determines which transport geometry is most useful. Local
cardinality constraints define the admissible state space of each decision variable, such as one-hot or \(k\)-hot blocks while the remaining routing,
assignment, coloring, scheduling, matching, or capacity constraints couple the blocks and create problem-dependent transport bottlenecks. A mixer can
preserve the correct local sector while still providing transport that is poorly matched to these local and global relations. The finite-size diagnostics in
Section~\ref{subsec:numerical-transport-diagnostics} make this distinction explicit.

A second central insight is the separation of transport from interference.
The shell-transfer recursion determines the absolute path mass supplied by
the mixer, while the actual quantum amplitude remains a coherent sum of complex path contributions. Lemma~\ref{lem:exact-path-sum} restores the phase information, and Theorem
\ref{thm:phase-aligned-path-sum-lower-bound} gives a sufficient condition under which the available transport mass survives as target amplitude. Our constructive result shows that the complete-graph mixer can generate enough absolute path mass to compensate for the large normalization factor that comes from the encoded search space. Under the phase-alignment condition, this produces the success guarantee in Theorem
\ref{thm:polynomial-shot-success}. The guarantee therefore combines
a geometric transport resource with a phase-coherent mechanism for harvesting
that resource.

After normalization, the shell-transfer coefficients define a finite radial Markov process. This makes drift, variance, concentration, and long-time behavior accessible through classical probability theory. In the complete-graph case, the normalized process moves toward the ordinary bulk of
the product space rather than toward the reference configuration. Mixer
transport alone therefore supplies path mass without target
concentration. The phase layers determine whether this mass is converted into
Born probability on a useful configuration or target set.

The mixer and Trotterization studies show that preservation of the encoded
sector does not guarantee preservation of the relevant transport law.
Different logical mixers generate different shell structures, and a finite
product-formula realization can preserve excitation number while distorting
the shell dynamics of the ideal mixer. The numerical results further show
that high process fidelity can coexist with substantial shell-transport
error. The shell diagnostics therefore provide an algorithm-specific measure
of implementation quality.

Problem-specific classical maps form a third layer after transport and interference. They can  expose the constraint structure present in measured samples, guide feasibility-restoring operations, and separate the contribution of the quantum distribution from that of classical post-processing. This distinction is especially important for near-term
devices, where infeasible samples may retain useful and readily repairable structure. Multidimensional statistics can identify the specific violation patterns hidden by a scalar Hamming distance or total penalty and can therefore support more selective repair operations. These statistics also provide a route for applying tools from coding theory
and algebraic combinatorics to constrained quantum optimization. Distance distributions, Krawtchouk expansions, weight enumerators, association schemes, and related methods can be used to study a broader class of product-space structures and constraint statistics. The corresponding
mathematical connection is developed further in
Appendix~\ref{app:constraint-statistics-coding}.

Constrained quantum annealing may provide a second realization of the block-factorized transport geometry considered here when the driver Hamiltonian preserves independent fixed-cardinality sectors and acts through
block-local terms~\cite{HenSpedalieri2016,HenSarandy2016,leipold2026imposing}. However, the connection remains structural. Unlike the quantum alternating operator framework, standard constrained quantum annealing does not generally follow
the repeated alternation of mixer and diagonal phase layers used in the multilayer path expansion. The shell formalism may therefore inform the transport geometry generated by such drivers, while the phase-alignment and finite-depth sampling guarantees developed here apply specifically to the alternating circuit setting with persistent block factorized strucuture.

Natural extensions include consideration of
broader product spaces and heterogeneous mixer geometries, phase control beyond exact lattice normalization, and transport-preserving analysis of higher-order Trotter product formulas, connectivity constraints, and the effects of hardware noise \cite{onahfinite, Quetschlich2024ApplicationAware,Ji2025AOQMAP,Zhu2024Coqa}. Moreover our results may be further applied to quantum algorithms 
for problems encoded directly to emerging qudit-based quantum hardware~\cite{
chi2022programmable, Ringbauer2022UniversalQudit, nguyen2024empowering, kim2025ultracoherent, venturelli2025near,hadfield2026ndar}.

\subsection{Conclusion}
\label{subsec:conclusion}

We have developed a framework for separating mixer transport from
phase-coherent interference in constrained quantum optimization on product
spaces. For the complete-graph block-\(XY\) mixer, the factorized local
transition kernel produces an exact shell-resolved description on
\(\mathcal X=[n]^m\). Proposition~\ref{prop:shell-transfer-generating} gives the
one-layer shell-transfer generating function, while
Theorem~\ref{thm:exact-shell-transfer-recursion} gives the corresponding closed form multilayer recursion. The resulting shell-transfer law quantifies the absolute path mass generated by the mixer. The exact path expansion and the phase-alignment condition determine whether this mass contributes constructively to the target amplitude. The constructive mixer angle of
Proposition~\ref{prop:constructive-radial-mass-lower-bound}, together with Theorem~\ref{thm:phase-aligned-path-sum-lower-bound}, yields the success guarantee in Theorem~\ref{thm:polynomial-shot-success}; identifying a concrete mechanism by which transport across an exponentially large product space can be converted into a sampling probability whose scale is independent of the ambient Hilbert-space dimension, the product-space cardinality, and the feasible-set size.

We show that the normalized shell process provides a finite  probabilistic description of the same mixer dynamics. Our analysis shows how the absolute transfer envelope moves across the product space and why mixer transport alone does not concentrate on the target. The mixer supplies the available path mass, while the phase layers provide the constructive interference required to convert that mass into useful Born probability.

The application results show that the transport description also provides a practical diagnostic. It distinguishes logical mixer geometries, identifies the statistics required to describe their transport, and measures how finite Trotterization changes the shell law of the complete-graph mixer. The comparison with process fidelity shows that a small generic unitary error can coexist with substantial distortion of the transport structure relevant to the sampling guarantee.

The resulting framework connects the structure of the optimization problem to the full hybrid quantum--classical workflow. Local constraints determine the encoded blocks, while global constraints determine how those blocks interact, where useful configurations lie, and which transport bottlenecks the mixer must overcome. The mixer generates transport on the resulting product space, and the phase layers determine whether the available path mass is converted into useful Born probability. After measurement, problem-specific classical maps expose the constraint
structure present in the generated samples and can guide feasibility-restoring repair operations on noisy near-term devices. This makes it possible to separate the contribution of the quantum-generated distribution from the improvement supplied by classical post-processing. Such attribution is especially important for near-term devices, where noisy or approximately implemented circuits may produce infeasible samples that nevertheless retain useful and readily repairable combinatorial structure. The separation of transport, interference, implementation, and post-measurement processing provides a direct language for analyzing and designing constrained quantum optimization algorithms from their product-space dynamics.

Looking ahead, at the transport-maximizing mixer angle, the available
absolute path mass is already large enough that a problem-size-independent
target probability, and hence constant-shot success at fixed confidence,
is not excluded. Corollary~\ref{cor:maximal-transport-phase-obstruction}
shows only that this stronger regime can neither be certified by the
common-cone argument beyond one layer nor ruled out by it. The examples in
Appendix~\ref{app:sampling-complexity} demonstrate that \(O(1)\) target
probability can nevertheless occur outside that certificate, providing a
counterexample to any claim that failure of the common-cone condition rules
out such strong success. In particular, maximal transport leaves open the
possibility of constant-shot success independent of problem size, but
establishing such a guarantee requires phase-control principles beyond the
present cone analysis. The gap between what maximal transport can support
and what the current interference certificate can prove suggests additional algorithmic power in the quantum alternating operator framework still to be exposed and exploited in constrained quantum
optimization. Crossing this chasm, or even approaching it, is left to
future work.

\section*{Data and code availability}

The numerical data and supporting code used in this work are available on Zenodo at
\href{https://doi.org/10.5281/zenodo.21302533}
{https://doi.org/10.5281/zenodo.21302533}.

\section*{Conflict of Interests}
All authors declare no competing interests.

\section*{Acknowledgements}
The authors gratefully acknowledge helpful discussions initiated at the International Workshop on Quantum Optimization held in Oslo, Norway in March 2026. S.H. acknowledges support from U.S. Department of Energy under grant No. DE-SC0026126.

\appendix
\appendix

\section{Appendix}

\subsection{A catalogue of mixer geometries}
\label{app:xy-mixer-catalogue}
Table~\ref{tab:xy-mixer-variants-product spaces} summarizes several mixer constructions acting on fixed-Hamming-weight encoded Hilbert spaces. The mixer unitary acts on \(\mathcal H_{\mathrm{OH}}\), while its matrix elements \[ \langle z|U_M(\beta)|x\rangle, \qquad x,z\in \mathcal X, \] induce a transition kernel indexed by the corresponding classical configurations. The shell-reduction framework exposes the resulting trade-offs by comparing the transition kernel induced by each topology, the amount and range of transfer available per layer, the symmetry that permits a reduced description, and the quantum resources required to implement the mixer. Section ~\ref{subsec:numerical-transport-diagnostics} compares these effects through finite-size transport diagnostics, including the available transfer mass, the appropriate reduced statistic, and the induced shell-sampling behavior. 



\begin{table}[H]
\centering
\footnotesize
\setlength{\tabcolsep}{3pt}
\renewcommand{\arraystretch}{1.18}
\caption{Common \(XY\)-mixer variants viewed as constrained product-space
transport geometries. The table indicates the invariant space generated by each
variant and the corresponding product-space structure on which shell or refined
shell diagnostics can be applied.}
\label{tab:xy-mixer-variants-product spaces}
\begin{adjustbox}{max width=\textwidth}
\begin{tabular}{
@{}
>{\raggedright\arraybackslash}p{0.13\linewidth}
>{\raggedright\arraybackslash}p{0.22\linewidth}
>{\raggedright\arraybackslash}p{0.49\linewidth}
>{\raggedright\arraybackslash}p{0.10\linewidth}
@{}}
\toprule
\textbf{Mixer} &
\textbf{Typical form} &
\textbf{product-space consequence} &
\textbf{Examples} \\
\midrule

Ring or cycle \(XY\) &
\(\sum_i (X_iX_{i+1}+Y_iY_{i+1})\), cyclic indices &
Preserves each fixed-Hamming-weight sector. Applied blockwise to one-hot
registers, it gives a product of cyclic local transport spaces. &
\cite{Hadfield2019AOA,Wang2020XYMixers,Cook2019MaxKVertexCover} \\

Path or chain \(XY\) &
\(\sum_{i=1}^{n-1}(X_iX_{i+1}+Y_iY_{i+1})\) &
Preserves excitation number with open-boundary local transport. In block form,
the encoded space remains a product space, but local amplitudes depend on path
geometry. &
\cite{He2023A,KordonowyLeipold2026XYLie} \\

Complete-graph \(XY\) &
\(\sum_{i<j}(X_iX_j+Y_iY_j)\) &
Preserves fixed Hamming weight and treats all local symbols symmetrically. The transition amplitude depends only on equality versus inequality, giving the product kernel used in the shell reduction. &
\cite{Cook2019MaxKVertexCover,Hadfield2019AOA,KordonowyLeipold2026XYLie} \\

One-hot block \(XY\) &
\(H_M=\sum_{b=1}^m H_{XY}^{(b)}\) &
Each block remains in \(\mathcal H_1\). The full encoded manifold is the
product alphabet space \(\mathcal X=[n]^m\), as in assignment, routing, coloring, and
scheduling encodings. &
\cite{Hadfield2019AOA,onahce} \\

\(k\)-hot / Dicke \(XY\) &
\(H_{XY}(G)\) restricted to \(\mathcal H_k=\{x:\|x\|_1=k\}\) &
Preserves cardinality. With \(m\) blocks, the encoded manifold is a product of
local \(k\)-subset spaces, \(X_k\). &
\cite{Cook2019MaxKVertexCover,B_rtschi_2019} \\

Partitioned / trotterized \(XY\) &
Product over edge matchings or commuting edge sets &
Preserves the same constrained sector as the underlying \(XY\) graph, while the implemented unitary may differ. This gives a
hardware-aware product-space dynamics. &
\cite{Wang2020XYMixers,awasthi2026constraintpreservingxymixerstrotterized,He2023A} \\

Nearby-value \(XY\) &
Exchange between selected nearby local symbols &
Creates a product space with additional local geometry. The reduced dynamics
may depend on cyclic distance, path distance, or another local relation rather
than only equality versus inequality. &
\cite{Hadfield2019AOA,Wang2020XYMixers} \\

Arbitrary-topology \(XY\) &
\(H_{XY}(G)\) for a chosen interaction graph \(G\) &
Preserves Hamming weight, while the transport geometry is controlled by \(G\).
Blockwise use again produces constrained product spaces, with topology-specific
radial reductions. &
\cite{KordonowyLeipold2026XYLie} \\

Multi-angle \(XY\) &
\(\sum_{(i,j)\in E}\beta_{ij}(X_iX_j+Y_iY_j)\) &
Preserves the same excitation-number constraint but changes controllability,
trainability, and the effective geometry of transport on the manifold. &
\cite{KordonowyLeipold2026XYLie} \\

Extended \(XY\) &
\(XY\) plus diagonal, \(ZZ\), warm-start, or counterdiabatic terms &
Can preserve the same cardinality when the added generators
respect excitation number. &
\cite{KordonowyLeipold2026XYLie,Ruan2025XYCD} \\

\bottomrule
\end{tabular}
\end{adjustbox}
\end{table}

\subsection{Technical Proofs}
\label{app:proofs}

\begin{proof}[Proof of Lemma \ref{lem:single-block-xy-amplitudes}]
We work inside the subspace \(\mathcal H_1\).
Using the raising and lowering operators
\[
\sigma_j^+ = \frac{1}{2}(X_j+iY_j),
\qquad
\sigma_j^- = \frac{1}{2}(X_j-iY_j),
\]
one has
\[
X_uX_v+Y_uY_v
=
2(\sigma_u^+\sigma_v^-+\sigma_u^-\sigma_v^+).
\]
Hence the exchange term moves a single excitation between the two sites
\(u\) and \(v\) with 
\[
(X_uX_v+Y_uY_v)\ket u=2\ket v,
\qquad
(X_uX_v+Y_uY_v)\ket v=2\ket u,
\]
and
\[
(X_uX_v+Y_uY_v)\ket r=0
\qquad
\text{if } r\notin\{u,v\}.
\]
Therefore, restricted to \(\mathcal H_1\),
\begin{equation}
\label{eq:xy-restricts-to-complete-graph}
\left.
\sum_{1\le u<v\le n}
(X_uX_v+Y_uY_v)
\right|_{\mathcal H_1}
=
2\sum_{1\le u<v\le n}
\bigl(\ket u\bra v+\ket v\bra u\bigr)
=
2A_{K_n},
\end{equation}
where \(A_{K_n}\) is the adjacency matrix of the complete graph on
\([n]\). Thus
\[
\left.
\widetilde H_{XY}^{(b)}
\right|_{\mathcal H_1}
=
\frac{2}{n-1}A_{K_n}.
\]
Since the mixer contains the factor \(1/2\) in the exponent, we obtain
\begin{equation}
\label{eq:block-mixer-adjacency-exponential}
\left.
U_M^{(b)}(\beta)
\right|_{\mathcal H_1}
=
\exp\!\left(
-i\beta\frac{A_{K_n}}{n-1}
\right).
\end{equation}

It remains to exponentiate the complete-graph adjacency matrix. Let
\[
\ket{s_{\mathrm{blk}}}
=
\frac{1}{\sqrt n}\sum_{r=1}^{n}\ket r,
\qquad
P_s=\ket{s_{\mathrm{blk}}}\bra{s_{\mathrm{blk}}},
\qquad
P_s^\perp=I-P_s.
\]
The complete-graph adjacency matrix satisfies
\[
A_{K_n}\ket{s_{\mathrm{blk}}}
=
(n-1)\ket{s_{\mathrm{blk}}},
\]
because every vertex of \(K_n\) has degree \(n-1\). On the orthogonal
complement of \(\ket{s_{\mathrm{blk}}}\), the eigenvalue is \(-1\). Indeed,
if \(\ket\phi=\sum_{r=1}^n \phi_r\ket r\) satisfies
\(\sum_{r=1}^n \phi_r=0\), then
\[
A_{K_n}\ket\phi
=
\sum_{r=1}^n
\left(
\sum_{\substack{q=1\\q\neq r}}^n \phi_q
\right)
\ket r
=
\sum_{r=1}^n
(-\phi_r)\ket r
=
-\ket\phi.
\]
Equivalently,
\begin{equation}
\label{eq:complete-graph-projector-decomposition}
A_{K_n}
=
(n-1)P_s-P_s^\perp.
\end{equation}
Substituting \eqref{eq:complete-graph-projector-decomposition} into
\eqref{eq:block-mixer-adjacency-exponential} and using the orthogonality of
\(P_s\) and \(P_s^\perp\), we get
\begin{equation}
\label{eq:block-mixer-projector-form}
U_M^{(b)}(\beta)
=
e^{-i\beta}P_s
+
e^{i\beta/(n-1)}P_s^\perp .
\end{equation}

Now take matrix elements in the one-hot basis. Since
\[
P_s
=
\frac{1}{n}
\sum_{r,q=1}^{n}\ket r\bra q,
\]
we have
\[
\bra r P_s\ket q=\frac{1}{n},
\qquad
\bra r P_s^\perp\ket q
=
\bra r(I-P_s)\ket q
=
\delta_{rq}-\frac{1}{n}.
\]
Therefore
\begin{align}
\label{eq:block-mixer-general-entry}
\bra r U_M^{(b)}(\beta)\ket q
&=
e^{-i\beta}\bra r P_s\ket q
+
e^{i\beta/(n-1)}
\bra r P_s^\perp\ket q
\\
&=
\frac{1}{n}e^{-i\beta}
+
\left(\delta_{rq}-\frac{1}{n}\right)e^{i\beta/(n-1)}.
\end{align}
If \(r=q\), this gives
\[
\bra r U_M^{(b)}(\beta)\ket r
=
\frac{1}{n}e^{-i\beta}
+
\left(1-\frac{1}{n}\right)e^{i\beta/(n-1)}
=
a_0(\beta).
\]
If \(r\neq q\), this gives
\[
\bra r U_M^{(b)}(\beta)\ket q
=
\frac{1}{n}
\left(
e^{-i\beta}
-
e^{i\beta/(n-1)}
\right)
=
a_1(\beta).
\]
Thus the single-block complete-graph XY mixer has  one diagonal
transition amplitude and one off-diagonal transition amplitude.
\end{proof}

\begin{proof}[Proof of Proposition~\ref{prop:constructive-radial-mass-lower-bound}]
\label{app:proof-constructive-radial-mass}
Define
\[
\theta_n(\beta):=\frac{n\beta}{2(n-1)}.
\]
From
\[
a_1(\beta)
=
\frac1n\left(e^{-i\beta}-e^{i\beta/(n-1)}\right),
\]
the phase difference is \(n\beta/(n-1)\), and therefore
\[
|a_1(\beta)|
=
\frac{2}{n}
\left|
\sin\left(\frac{n\beta}{2(n-1)}\right)
\right|
=
\frac{2}{n}|\sin\theta_n(\beta)|.
\]

Similarly,
\[
a_0(\beta)
=
\frac1n e^{-i\beta}
+
\left(1-\frac1n\right)e^{i\beta/(n-1)}.
\]
Hence
\begin{align}
|a_0(\beta)|^2
&=
\frac1{n^2}
+
\left(1-\frac1n\right)^2
+
\frac{2(n-1)}{n^2}
\cos\left(\frac{n\beta}{n-1}\right)
\nonumber\\
&=
1-\frac{4(n-1)}{n^2}\sin^2\theta_n(\beta).
\end{align}

Consequently,
\begin{equation}
\label{eq:qn-closed-form}
q_n(\beta)
=
\left(
1-\frac{4(n-1)}{n^2}\sin^2\theta_n(\beta)
\right)^{1/2}
+
\frac{2(n-1)}{n}
|\sin\theta_n(\beta)|.
\end{equation}

For \(n\ge4\), the right-hand side of
\eqref{eq:qn-closed-form} is maximized at
\[
|\sin\theta_n(\beta)|=1.
\]
Choosing
\[
\beta_\star=\frac{\pi(n-1)}{n}
\]
gives
\[
\theta_n(\beta_\star)=\frac{\pi}{2},
\]
and therefore
\[
|a_1(\beta_\star)|=\frac2n,
\qquad
|a_0(\beta_\star)|=\frac{n-2}{n}.
\]
Thus
\[
q_n(\beta_\star)
=
\frac{n-2}{n}
+
(n-1)\frac2n
=
3-\frac4n.
\]

Finally, using
\[
v_0^{(p)}
=
\prod_{\ell=1}^{p}q_n(\beta_\ell)^m
\]
and setting
\[
\beta_1=\cdots=\beta_p=\beta_\star
\]
gives
\[
v_0^{(p)}
=
\left(3-\frac4n\right)^{mp},
\]
which is Eq.~\eqref{eq:v0-star-exact}.
\end{proof}

\begin{proof}[Proof of Corollary~\ref{cor:maximal-transport-phase-obstruction}]
The absolute path-mass formula follows from Proposition~\ref{prop:constructive-radial-mass-lower-bound},
and Eq.~\eqref{eq:maximal-transport-depth} follows by taking logarithms. For the phase claim, the one-block mixer at $\beta_\star$ is
\begin{equation}
U_\star
=e^{i\pi/n}\left(I-\frac2nJ\right),
\end{equation}
where $J$ is the all-ones matrix. After removal of its
global phase, diagonal entries are positive and
off-diagonal entries are negative. All entries are nonzero.

Fix $x_2,\ldots,x_p=y$ and all but one coordinate of
$x_0$ and $x_1$. Write the resulting configurations as
$X(a)$ and $I(i)$, respectively, with $a,i\in[n]$.
For this family, each complete path contribution has the form
\begin{equation}
A_{ia}=r_{ia}c_i z_a s_{ia},
\qquad
r_{ia}>0,
\qquad
s_{ia}=
\begin{cases}
+1,&i=a,\\
-1,&i\neq a.
\end{cases}
\label{eq:obstruction-path-family}
\end{equation}
Here
\[
z_a=e^{-i\gamma_1E(X(a))}
\]
contains the full first cost phase, while $c_i$ contains
the full second cost phase, the phase of the second mixer
transition, and all remaining fixed phases.
No separability of the cost Hamiltonian is assumed.

Suppose that all path contributions lie in one open
semicircle. Rotating its midpoint to the positive real
axis and absorbing that rotation into $c_i$ yields
\begin{equation}
\operatorname{Re}(c_i z_i)>0,
\qquad
\operatorname{Re}(c_i z_a)<0
\quad(a\neq i).
\label{eq:obstruction-separation}
\end{equation}
Thus each direction $z_i$ must be strictly separable from
every other direction by a line through the origin.

Choose four indices. Coincident directions immediately
contradict Eq.~\eqref{eq:obstruction-separation}.
Otherwise, order their four distinct directions around
the circle and denote the successive positive gaps by
$g_1,\ldots,g_4$, with
\[
g_1+g_2+g_3+g_4=2\pi.
\]
Strict separation of each direction from the other three
requires its two adjacent gaps to sum to more than $\pi$.
Summing these four requirements gives
\[
2(g_1+g_2+g_3+g_4)>4\pi,
\]
a contradiction. Hence the selected path family, and
therefore the complete path sum, cannot satisfy the
common-open-semicircle condition.
\end{proof}

\section{Measure-theoretic characterization of the shell reduction}
\label{app:measure-shell-reduction}


\paragraph{Shell-independent normalization.}

The shell-transfer generating function Proposition~\ref{prop:shell-transfer-generating} already contains the normalization
needed for a probabilistic interpretation of the radial dynamics. In
particular, evaluating it at \(\zeta=1\) sums the absolute transfer weight
over all outgoing shells. The resulting quantity is independent of the
incoming shell \(t\). Thus every column of the shell-transfer matrix carries
the same total absolute mass, and a single normalization converts the
shell-transfer coefficients into a stochastic radial kernel. Moreover, the
normalized generating function retains the coordinatewise factorization of
the underlying product mixer, which will lead directly to the
Poisson-binomial description below in Cor. \ref{cor:two-species-poisson-binomial}.

\begin{corollary}[Shell-independent normalization]
\label{cor:shell-independent-normalization}
Let \(T_{r,t}(\beta)\) be the shell-transfer coefficients of
Proposition~\ref{prop:shell-transfer-generating}, and recall
\[
q_n(\beta)
=
|a_0(\beta)|+(n-1)|a_1(\beta)|
\]
from \cref{eq:one-layer-weight}. Then, for every
\(t\in\{0,\dots,m\}\),
\begin{equation}
\label{eq:shell-transfer-total-mass}
\sum_{r=0}^{m}T_{r,t}(\beta)
=
q_n(\beta)^m.
\end{equation}
Hence the total absolute transfer mass is independent of the incoming
shell \(t\). Consequently,
\[
P_{r,t}^{(\beta)}
:=
\frac{T_{r,t}(\beta)}{q_n(\beta)^m}
\]
is normalized over \(r\), and its generating function is
\begin{align}
\label{eq:normalized-shell-transfer-gf}
\sum_{r=0}^{m}P_{r,t}^{(\beta)}\zeta^r
&=
\left[
\frac{
|a_1(\beta)|
+
\bigl(|a_0(\beta)|+(n-2)|a_1(\beta)|\bigr)\zeta
}{
q_n(\beta)
}
\right]^t
\left[
\frac{
|a_0(\beta)|
+
(n-1)|a_1(\beta)|\zeta
}{
q_n(\beta)
}
\right]^{m-t}.
\end{align}
\end{corollary}

\begin{proof}
Setting \(\zeta=1\) in
\cref{eq:shell-transfer-gf-general} gives
\begin{align}
\sum_{r=0}^{m}T_{r,t}(\beta)
&=
\bigl(
|a_1(\beta)|
+
|a_0(\beta)|
+
(n-2)|a_1(\beta)|
\bigr)^t
\nonumber\\
&\quad\times
\bigl(
|a_0(\beta)|
+
(n-1)|a_1(\beta)|
\bigr)^{m-t}
\nonumber\\
&=
q_n(\beta)^t q_n(\beta)^{m-t}
=
q_n(\beta)^m.
\end{align}
Dividing \cref{eq:shell-transfer-gf-general} by
\(q_n(\beta)^m\) gives
\cref{eq:normalized-shell-transfer-gf}.
\end{proof}

\paragraph{Normalized product kernel.}

Corollary~\ref{cor:shell-independent-normalization} shows that the same
normalization applies to every incoming shell. We now realize the resulting
shell kernel as the radial pushforward of a coordinatewise product
probability measure.

For \(a,b\in[n]\), let
\[
k_\beta(a;b)
:=
\begin{cases}
|a_0(\beta)|, & a=b,\\[2mm]
|a_1(\beta)|, & a\neq b,
\end{cases}
\]
and normalize it to the one-site probability kernel
\[
p_\beta(a\mid b):=\frac{k_\beta(a;b)}{q_n(\beta)}.
\]
For \(x=(x_1,\dots,x_m)\in \mathcal X=[n]^m\), define the associated product measure
\[
\mu_x^\beta
:=
\bigotimes_{v=1}^{m}\mu_{x,v}^\beta,
\qquad
\mu_{x,v}^\beta(\{a\})=p_\beta(a\mid x_v).
\]
For a fixed reference configuration \(y\in \mathcal X\), define the radial map
\[
R_y:\mathcal X\to\{0,1,\dots,m\},
\qquad
R_y(z):=d(z,y).
\]

\begin{lemma}[Finite Fubini principle for radial shell reduction]
\label{lem:fubini-radial-shell}
Let \(y\in \mathcal X=[n]^m\), and let
\[
R_y:\mathcal X\to\{0,\dots,m\},
\qquad
R_y(z)=d(z,y),
\]
be the radial map. For each depth \(\ell\), define the nonnegative path weight
\[
W_\ell(x_0,\dots,x_\ell)
:=
\prod_{j=1}^{\ell}
\bigl|\langle x_j|U_M(\beta_j)|x_{j-1}\rangle\bigr|.
\]
Then for every function \(f:\{0,\dots,m\}\to\mathbb R\),
\begin{equation}
\label{eq:fubini-radial-shell}
\sum_{x_0,\dots,x_\ell\in \mathcal X}
f\!\left(R_y(x_\ell)\right)
W_\ell(x_0,\dots,x_\ell)
=
\sum_{r=0}^{m}
f(r)
\sum_{\substack{x_0,\dots,x_\ell\in \mathcal X\\ x_\ell\in S_r(y)}}
W_\ell(x_0,\dots,x_\ell).
\end{equation}
In particular, choosing \(f=\mathbf 1_{\{r\}}\) gives
\[
\sum_{\substack{x_0,\dots,x_\ell\in \mathcal X\\ x_\ell\in S_r(y)}}
W_\ell(x_0,\dots,x_\ell)
=
v_r^{(\ell)}.
\]
Thus the shell vector \(v^{(\ell)}\) is the pushforward of the full
nonnegative path-weight measure under the terminal radial map
\[
(x_0,\dots,x_\ell)\longmapsto R_y(x_\ell).
\]
\end{lemma}

\begin{proof}
Since all sums are finite, we may partition the terminal configurations
\(x_\ell\in \mathcal X\) according to the disjoint shell decomposition
\[
\mathcal X=\bigsqcup_{r=0}^{m}S_r(y).
\]
Therefore
\[
\sum_{x_0,\dots,x_\ell\in \mathcal X}
f(R_y(x_\ell))W_\ell(x_0,\dots,x_\ell)
=
\sum_{r=0}^{m}
\sum_{\substack{x_0,\dots,x_\ell\in \mathcal X\\ x_\ell\in S_r(y)}}
f(r)W_\ell(x_0,\dots,x_\ell),
\]
which is \eqref{eq:fubini-radial-shell}. The final identity follows from
Theorem~\ref{thm:exact-shell-transfer-recursion}.
\end{proof}

The finite Fubini principle makes explicit the measure-theoretic content of the
shell reduction. Before phases are taken into account, the product mixer defines
a nonnegative path-weight measure on the full path space \(\mathcal X^{\ell+1}\). The
radial map \(R_y(x_\ell)=d(x_\ell,y)\) then pushes this full path measure forward
to a measure on the shell space \(\{0,\dots,m\}\). The shell recursion admits a particularly useful pointwise corollary.

\begin{corollary}[Normalized shell-count recursion]
\label{cor:normalized-shell-count-recursion}
Let
\[
v_t^{(0)}:=|S_t(y)|=\binom{m}{t}(n-1)^t,
\qquad
v_r^{(\ell)}
=
\sum_{t=0}^{m} T_{r,t}(\beta_\ell)\,v_t^{(\ell-1)},
\qquad \ell\ge 1,
\]
be the unnormalized shell-transfer recursion from
\cref{thm:exact-shell-transfer-recursion}. Define
\[
\widehat v_r^{(\ell)}
:=
\left(\prod_{j=1}^{\ell} q_n(\beta_j)^{-m}\right) v_r^{(\ell)},
\qquad
\widehat v_r^{(0)}:=v_r^{(0)}.
\]
Then
\begin{equation}
\label{eq:normalized-v-recursion-general}
\widehat v_r^{(\ell)}
=
\sum_{t=0}^{m}
P_{r,t}^{(\beta_\ell)}\,\widehat v_t^{(\ell-1)}.
\end{equation}
Moreover, if
\[
\rho_t^{(0)}:=\frac{|S_t(y)|}{n^m}=\frac{v_t^{(0)}}{n^m},
\]
is the shell law of a uniformly random initial configuration in \(\mathcal X\), then
\[
\widehat v_r^{(\ell)}=n^m\,\rho_r^{(\ell)}.
\]
Here \(\rho^{(\ell)}\) is defined by the normalized recursion
\[
\rho_r^{(\ell)}
=
\sum_{t=0}^{m}P_{r,t}^{(\beta_\ell)}\rho_t^{(\ell-1)}.
\]
\end{corollary}

\begin{proof}
Since
\[
T_{r,t}(\beta_\ell)=q_n(\beta_\ell)^m P_{r,t}^{(\beta_\ell)},
\]
the recursion immediately normalizes to
\eqref{eq:normalized-v-recursion-general}. Since
\[
\sum_{t=0}^{m}v_t^{(0)}=|\mathcal X|=n^m,
\]
the second claim follows by comparing the normalized shell to the
radial Markov recursion.
\end{proof}

\subsection{Shell transfer as a radial pushforward of a product kernel}

\begin{proposition}[Shell transfer as a radial pushforward of a product kernel]
\label{prop:shell-transfer-product-pushforward}
Fix \(x,y\in \mathcal X\), and let
\[
t:=d(x,y).
\]
Then the pushforward law \((R_y)_\#\mu_x^\beta\) depends only on \(t\), not on
the specific pair \((x,y)\). More precisely,
\begin{equation}
\label{eq:radial-pushforward-T-general}
\Pr_{\mu_x^\beta}(R_y(Z)=r)
=
\frac{T_{r,t}(\beta)}{q_n(\beta)^m},
\qquad r=0,\dots,m.
\end{equation}
Equivalently, the probability generating function of the radial pushforward law is
\begin{equation}
\label{eq:radial-pgf-general}
\sum_{r=0}^{m}
\Pr_{\mu_x^\beta}(R_y(Z)=r)\,\zeta^r
=
\left(
\frac{|a_1(\beta)|+\bigl(|a_0(\beta)|+(n-2)|a_1(\beta)|\bigr)\zeta}
     {q_n(\beta)}
\right)^t
\left(
\frac{|a_0(\beta)|+(n-1)|a_1(\beta)|\zeta}
     {q_n(\beta)}
\right)^{m-t}.
\end{equation}
Multiplying by \(q_n(\beta)^m\) recovers the unnormalized shell-transfer
generating function
\begin{equation}
\label{eq:radial-pgf-unnormalized-general}
\sum_{r=0}^{m} T_{r,t}(\beta)\,\zeta^r
=
\Bigl(|a_1(\beta)| + \bigl(|a_0(\beta)|+(n-2)|a_1(\beta)|\bigr)\zeta\Bigr)^t
\Bigl(|a_0(\beta)| + (n-1)|a_1(\beta)|\zeta\Bigr)^{m-t}.
\end{equation}
\end{proposition}

\begin{proof}
Because the mixer factorizes across coordinates, one has
\[
\bigl|\langle z\mid U_M(\beta)\mid x\rangle\bigr|
=
\prod_{v=1}^{m} k_\beta(z_v;x_v),
\qquad z\in \mathcal X.
\]
Hence
\[
T_{r,t}(\beta)
=
\sum_{z:\,R_y(z)=r}\prod_{v=1}^{m} k_\beta(z_v;x_v).
\]
Summing over \(r\) with weight \(\zeta^r\) gives
\[
\sum_{r=0}^{m} T_{r,t}(\beta)\,\zeta^r
=
\sum_{z\in \mathcal X}
\prod_{v=1}^{m}
\Bigl(
k_\beta(z_v;x_v)\,\zeta^{\mathbf 1_{\{z_v\neq y_v\}}}
\Bigr).
\]
By repeated summation over the finite product space \(\mathcal X=\prod_{v=1}^{m}[n]\),
this factorizes as
\[
\prod_{v=1}^{m}
\left(
\sum_{a\in[n]} k_\beta(a;x_v)\,\zeta^{\mathbf 1_{\{a\neq y_v\}}}
\right).
\]

If \(x_v=y_v\), the local sum equals
\[
|a_0(\beta)|+(n-1)|a_1(\beta)|\,\zeta.
\]
If \(x_v\neq y_v\), the local sum equals
\[
|a_1(\beta)|
+
\bigl(|a_0(\beta)|+(n-2)|a_1(\beta)|\bigr)\zeta.
\]
Since exactly \(t=d(x,y)\) coordinates satisfy \(x_v\neq y_v\), while the
remaining \(m-t\) coordinates satisfy \(x_v=y_v\), one obtains
\eqref{eq:radial-pgf-unnormalized-general}. Dividing by \(q_n(\beta)^m\) yields
\eqref{eq:radial-pgf-general}, since
\[
\mu_x^\beta(\{z\})
=
\prod_{v=1}^{m} p_\beta(z_v\mid x_v)
=
\frac{1}{q_n(\beta)^m}
\prod_{v=1}^{m} k_\beta(z_v;x_v).
\]
Taking the coefficient of \(\zeta^r\) gives
\eqref{eq:radial-pushforward-T-general}.
\end{proof}

\begin{corollary}[Normalized shell kernel is stochastic]
\label{cor:normalized-shell-kernel-stochastic}
Define
\[
P_{r,t}^{(\beta)}
:=
\frac{T_{r,t}(\beta)}{q_n(\beta)^m},
\qquad r,t\in\{0,\dots,m\}.
\]
Then for each fixed \(t\),
\[
P_{r,t}^{(\beta)}\ge 0,
\qquad
\sum_{r=0}^{m} P_{r,t}^{(\beta)}=1.
\]

Thus \(P^{(\beta)}=(P_{r,t}^{(\beta)})_{r,t=0}^{m}\) is column-stochastic
in the convention used here.
\end{corollary}

\begin{proof}
Nonnegativity is immediate. Setting \(\zeta=1\) in
\eqref{eq:radial-pgf-general} gives
\[
\sum_{r=0}^{m} P_{r,t}^{(\beta)}
=
\left(\frac{q_n(\beta)}{q_n(\beta)}\right)^t
\left(\frac{q_n(\beta)}{q_n(\beta)}\right)^{m-t}
=
1.
\]
\end{proof}

\begin{corollary}[Two-species Poisson-binomial form of the shell law]
\label{cor:two-species-poisson-binomial}
Let \(Z\sim\mu_x^\beta\), and let \(t=d(x,y)\). Then
\[
R_y(Z)=d(Z,y)
\]
is a sum of independent Bernoulli variables of two species. More precisely,
there exist independent random variables
\[
\xi_1,\dots,\xi_{m-t},
\qquad
\eta_1,\dots,\eta_t,
\]
such that
\[
R_y(Z)
\;\stackrel{d}{=}\;
\sum_{j=1}^{m-t}\xi_j+\sum_{k=1}^{t}\eta_k,
\]
with
\[
\Pr(\xi_j=1)
=
\frac{(n-1)|a_1(\beta)|}{q_n(\beta)},
\qquad
\Pr(\xi_j=0)
=
\frac{|a_0(\beta)|}{q_n(\beta)},
\]
and
\[
\Pr(\eta_k=1)
=
\frac{|a_0(\beta)|+(n-2)|a_1(\beta)|}{q_n(\beta)},
\qquad
\Pr(\eta_k=0)
=
\frac{|a_1(\beta)|}{q_n(\beta)}.
\]
\end{corollary}

\begin{proof}
Under the product measure \(\mu_x^\beta\), the coordinates \(Z_v\) are
independent. For each coordinate \(v\), the indicator
\[
\mathbf 1_{\{Z_v\neq y_v\}}
\]
depends only on whether \(x_v=y_v\) or \(x_v\neq y_v\). If \(x_v=y_v\), then
\[
\Pr(Z_v\neq y_v)
=
\frac{(n-1)|a_1(\beta)|}{q_n(\beta)}.
\]
If \(x_v\neq y_v\), then
\[
\Pr(Z_v\neq y_v)
=
\frac{|a_0(\beta)|+(n-2)|a_1(\beta)|}{q_n(\beta)}.
\]
Since
\[
R_y(Z)=\sum_{v=1}^{m}\mathbf 1_{\{Z_v\neq y_v\}},
\]
the claim follows.
\end{proof}

\begin{theorem}[Depth-\(p\) radial Markov recursion from product path measures]
\label{thm:depth-p-radial-markov}
Fix \(y\in \mathcal X\). Let \(Z^{(0)},Z^{(1)},\dots,Z^{(p)}\) be a time-inhomogeneous
Markov chain on \(\mathcal X\) defined by the one-step transition kernels
\[
\Pr\!\bigl(Z^{(\ell)}=z \,\big|\, Z^{(\ell-1)}=x\bigr)
=
p_{\beta_\ell}(z\mid x),
\qquad \ell=1,\dots,p.
\]

Equivalently, we write
\[
p_\beta(z\mid x)
:=
\prod_{v=1}^{m}p_\beta(z_v\mid x_v),
\qquad x,z\in \mathcal X,
\]
for the associated product transition kernel on \(\mathcal X\).

Define the radial process
\[
R_\ell:=d(Z^{(\ell)},y)\in\{0,\dots,m\}.
\]
Then \((R_\ell)_{\ell=0}^{p}\) is itself a time-inhomogeneous Markov chain on
\(\{0,\dots,m\}\), with one-step kernels
\begin{equation}
\label{eq:radial-markov-kernel-general}
\Pr(R_\ell=r\mid R_{\ell-1}=t)
=
P_{r,t}^{(\beta_\ell)}
=
\frac{T_{r,t}(\beta_\ell)}{q_n(\beta_\ell)^m}.
\end{equation}
Consequently, if
\[
\rho_r^{(\ell)}:=\Pr(R_\ell=r),
\]
then
\begin{equation}
\label{eq:radial-law-recursion-general}
\rho_r^{(\ell)}
=
\sum_{t=0}^{m}
P_{r,t}^{(\beta_\ell)}\,\rho_t^{(\ell-1)},
\qquad \ell=1,\dots,p.
\end{equation}
\end{theorem}

\begin{proof}
Condition on \(Z^{(\ell-1)}=x\) with \(d(x,y)=t\). By
Proposition~\ref{prop:shell-transfer-product-pushforward}, the pushforward law
of \(R_y(Z^{(\ell)})\) under the transition kernel
\(p_{\beta_\ell}(\cdot\mid x)\) depends only on \(t\), and is exactly
\(P_{r,t}^{(\beta_\ell)}\). This proves
\eqref{eq:radial-markov-kernel-general}. The recursion
\eqref{eq:radial-law-recursion-general} is then the law of total probability.
\end{proof}

\subsection{Proof of the local-to-product transport bounds} 
\label{app:trotter-transport-certification}

\begin{proof}[Proof of Proposition~\ref{prop:local-product-transport-bounds}] 
Fix \(x\in S_t(y)\). Under the product kernel 
\[ K_s(z\mid x) = \prod_{b=1}^{m} \widehat K_s(z_b\mid x_b), 
\] 

the shell index \(d(z,y)\) is the sum of independent Bernoulli variables, 

\[ d(z,y) = \sum_{b=1}^{m}B_b, \qquad B_b=\mathbf 1_{\{z_b\neq y_b\}}.
\] 

If \(x_b=y_b\), then \[ \Pr(B_b=1) = 1-\widehat K_s(y_b\mid y_b). \] If \(x_b\neq y_b\), then 
\[ \Pr(B_b=1) = 1-\widehat K_s(y_b\mid x_b).
\] 

Thus the shell profile is a Poisson-binomial law containing \(m-t\) parameters from the diagonal class and \(t\) parameters from the off-diagonal class. Consider \(x,x'\in S_t(y)\). Since both profiles contain the same numbers of diagonal- and off-diagonal-class factors, their Bernoulli variables can be paired within the two classes. Under a maximal coupling, the probability that a paired Bernoulli variable disagrees is the absolute difference of its parameters. Hence the diagonal pairs disagree with probability at most \(\delta_{\mathrm{diag}}(s)\), and the off-diagonal pairs disagree with probability at most \(\delta_{\mathrm{off}}(s)\). A union bound gives 
\[ \left\| P_s(\,\cdot\mid x) - P_s(\,\cdot\mid x') \right\|_{\mathrm{TV}} \leq (m-t)\delta_{\mathrm{diag}}(s) + t\,\delta_{\mathrm{off}}(s). \] Since total variation distance is at most one, this proves Eq.~\eqref{eq:lumpability-local-upper-bound}. For the ideal kernel, the Bernoulli parameter of a coordinate satisfying \(x_b=y_b\) is \(1-d_\star\), while that of a coordinate satisfying \(x_b\neq y_b\) is \(1-o_\star\). Coupling each compiled Bernoulli variable to its ideal counterpart gives disagreement probabilities bounded by \(\epsilon_{\mathrm{diag}}(s)\) and \(\epsilon_{\mathrm{off}}(s)\), respectively. Therefore \[ \left\| P_s(\,\cdot\mid x) - P^{(\beta_\star)}_{\cdot,t} \right\|_{\mathrm{TV}} \leq (m-t)\epsilon_{\mathrm{diag}}(s) + t\,\epsilon_{\mathrm{off}}(s), \] which proves Eq.~\eqref{eq:shell-kernel-local-upper-bound}. \end{proof} \subsection{Certified depth and exchange-gate count} The worst-case product-space bounds are \begin{align} \overline{\Delta}_{\mathrm{lump}}(s) &= \max_{0\leq t\leq m} \min\!\left\{ 1,\, (m-t)\delta_{\mathrm{diag}}(s) + t\,\delta_{\mathrm{off}}(s) \right\}, \\ \overline{\Delta}_{\mathrm{shell}}(s) &= \max_{0\leq t\leq m} \min\!\left\{ 1,\, (m-t)\epsilon_{\mathrm{diag}}(s) + t\,\epsilon_{\mathrm{off}}(s) \right\}. \end{align} For tolerance \(\varepsilon\), the certified Trotter count is \[ s_{\mathrm{cert}}(\varepsilon) = \min\left\{ s\in\mathbb N: \overline{\Delta}_{\mathrm{lump}}(s)\leq\varepsilon,\, \overline{\Delta}_{\mathrm{shell}}(s)\leq\varepsilon \right\}. \] An edge coloring of \(K_n\) requires \[ \chi'(K_n) = \begin{cases} n-1, & n\ \mathrm{even},\\ n, & n\ \mathrm{odd}. \end{cases} \] The corresponding logical exchange depth and two-qubit exchange-gate count are \begin{equation} D_{\mathrm{cert}}(\varepsilon) = s_{\mathrm{cert}}(\varepsilon)\chi'(K_n), \end{equation} and \begin{equation} G_{2\mathrm q}(\varepsilon) = m\,s_{\mathrm{cert}}(\varepsilon)\binom{n}{2}. \end{equation} The depth counts parallel matching layers before connectivity-dependent routing or decomposition into native gates. 

In our numerical protocol, for every \(n=m\in\{9,\ldots,25\}\), the ideal one-block mixer is evaluated exactly in the one-excitation sector. The Trotterized unitary is constructed from a fixed ordering of the matchings in an edge coloring of \(K_n\), at \[ \beta_\star=\frac{\pi(n-1)}{n}. \] The scan uses \(s=1,\ldots,10\) Trotter steps. The reference configuration \(y\in[n]^n\) is a valid permutation. Source configurations range over the full row-one-hot product space \(\mathcal X=[n]^n\). Arbitrary sources are drawn from every Hamming shell, and deterministic extremal sources are added to expose the largest observed violations of the local quotient symmetry. The maxima over the tested sources define \(\widehat{\Delta}_{\mathrm{lump}}\) and \(\widehat{\Delta}_{\mathrm{shell}}\). The conventional local process-infidelity error is \begin{equation} \varepsilon_{\mathrm{F}}^{\mathrm{loc}}(s) = 1- \frac{ \left| \operatorname{Tr}\!\left[ U_M^{(b)}(\beta_\star)^\dagger U_{M,s}^{(b)}(\beta_\star) \right] \right|^2 }{n^2}. \end{equation} This quantity is computed from the exact \(n\times n\) one-block unitaries. The transport defects are computed from the normalized absolute kernel \(\widehat K_s\). The two calculations therefore compare generic unitary approximation with preservation of the equality-versus-inequality quotient that generates the shell reduction.

\section{Choose Your Sampling Complexity Wisely}
\label{app:sampling-complexity}

The distinction between constant target probability and inverse-polynomial
target probability is most naturally understood through the corresponding
sampling complexity. Constant target probability implies constant-shot
success at fixed confidence, whereas inverse-polynomial target probability
already suffices for polynomial-shot success. The following example also
shows explicitly that failure of the common-cone certificate does not imply
failure of coherent amplification.

Let \(n=4\), take \(\beta_\ell=\beta_\star\) at every layer, and use an
identity cost phase for layers \(1,\ldots,p-1\).\footnote{This
counterexample was suggested by ChatGPT (GPT-5.6 Sol) and subsequently
verified by the authors.}
At the final layer set
\[
D_y\ket{x}
=
(-1)^{d(x,y)}\ket{x},
\]
equivalently,
\[
E_y(x)=d(x,y),
\qquad
\gamma_p=\pi.
\]

Writing
\[
V:=v_0^{(p)},
\qquad
N:=n^m,
\]
define the surviving coherent fraction
\[
\kappa_y
:=
\frac{|\langle y|\psi_p\rangle|}{V/\sqrt{N}}
=
\frac{\sqrt{N P_y}}{V}.
\]
Up to a global phase,
\[
U_\star^{\otimes m}
=
\left(I-\frac{J}{2}\right)^{\otimes m},
\]
and the uniform initial state is an eigenstate of this mixer. The row
\(\bra{y}U_\star^{\otimes m}\) has sign
\[
(-1)^{d(x,y)}
\]
and magnitude \(2^{-m}\) at every \(x\). Thus the final phase exactly
matches its signs and
\begin{equation}
\label{eq:known-target-exact-success}
P_y=1,
\qquad
V=2^{mp},
\qquad
\kappa_y=2^{-m(p-1)}.
\end{equation}

At \(p=1\), all full paths align and the common-arc certificate is attained
with \(L=K=1\). For every \(p\ge 2\),
Corollary~\ref{cor:maximal-transport-phase-obstruction} rules out that
certificate, yet \(P_y \approx  1\).

\begin{table}[t]
\centering
\caption{Finite-circuit benchmarks at the transport-maximizing angle
\(\beta_\star\). Here \(N=n^m\),
\(V/\sqrt{N}\) is the absolute-path amplification,
\(\kappa_y=\sqrt{N P_y}/V\) is the surviving coherent fraction,
and \(P_y\) is the target Born probability. The one-hot realization uses
\(nm\) qubits. A shortest arc of \(\pi\) indicates that the nonzero full
path summands cannot be contained in an open semicircle.}
\label{tab:known-target-benchmark}
\begin{tabular}{cccccc}
\hline
\((n,m,p)\)
&
qubits
&
\(V/\sqrt{N}\)
&
\(\kappa_y\)
&
\(P_y\)
&
shortest arc
\\
\hline
\((4,3,1)\) & \(12\) & \(1\)       & \(1\)        & \(1\)        & \(0\)   \\
\((4,3,2)\) & \(12\) & \(8\)       & \(0.125\)    & \(1\)        & \(\pi\) \\
\((4,3,3)\) & \(12\) & \(64\)      & \(0.015625\) & \(1\)        & \(\pi\) \\[1mm]

\((4,4,1)\) & \(16\) & \(1\)       & \(1\)        & \(1\)        & \(0\)   \\
\((4,4,2)\) & \(16\) & \(16\)      & \(0.0625\)   & \(1\)        & \(\pi\) \\
\((4,4,3)\) & \(16\) & \(256\)     & \(0.003906\) & \(1\)        & \(\pi\) \\[1mm]

\((5,3,1)\) & \(15\) & \(0.952386\)  & \(1\)        & \(0.907039\) & \(0\)   \\
\((5,3,2)\) & \(15\) & \(10.141007\) & \(0.093914\) & \(0.907039\) & \(\pi\) \\
\((5,3,3)\) & \(15\) & \(107.981442\)& \(0.008820\) & \(0.907039\) & \(\pi\) \\[1mm]

\((5,4,1)\) & \(20\) & \(0.937024\)  & \(1\)        & \(0.878014\) & \(0\)   \\
\((5,4,2)\) & \(20\) & \(21.950349\) & \(0.042688\) & \(0.878014\) & \(\pi\) \\
\((5,4,3)\) & \(20\) & \(514.200105\)& \(0.001822\) & \(0.878014\) & \(\pi\) \\[1mm]

\((5,5,1)\) & \(25\) & \(0.921910\)   & \(1\)        & \(0.849918\) & \(0\)   \\
\((5,5,2)\) & \(25\) & \(47.511834\)  & \(0.019404\) & \(0.849918\) & \(\pi\) \\
\((5,5,3)\) & \(25\) & \(2448.585095\)& \(0.000377\) & \(0.849918\) & \(\pi\) \\
\hline
\end{tabular}
\end{table}

Table~\ref{tab:known-target-benchmark} compares path enumeration
with dense-unitary evolution for the indicated small cases. The shortest
arc is computed from the nonzero full path summands; \(\pi\) means that no
open semicircle contains them. The two calculations agree to an absolute
tolerance of \(10^{-11}\).

The examples illustrate the point that  failure of the single-cone condition does not imply failure of coherent amplification. Although Corollary~\ref{cor:maximal-transport-phase-obstruction} places the
transport-maximizing schedule outside the reach of that certificate for
\(p\ge 2\),  \cref{eq:known-target-exact-success} shows that the Born
probability can nevertheless remain finite.

\paragraph{Sampling-complexity interpretation.}
If the single-shot target probability is \(q_y\), then after
\(N_{\mathrm{shot}}\) independent shots,
\[
\Pr[\text{never observe }y]
=
(1-q_y)^{N_{\mathrm{shot}}}
\le
e^{-N_{\mathrm{shot}}q_y}.
\]
Hence observing \(y\) with probability at least \(1-\varepsilon\) requires
\[
N_{\mathrm{shot}}
=
O\!\left(
\frac{\log(1/\varepsilon)}{q_y}
\right).
\]

A constant target probability,
\(
q_y=\Omega(1),
\)
therefore gives constant-shot success at fixed confidence,
\(
N_{\mathrm{shot}}=O(\log(1/\varepsilon)).
\)
By contrast, an inverse-polynomial guarantee,
\(
q_y
\ge
\frac{1}{\mathrm{poly}(n,m)},
\)
gives
\(
N_{\mathrm{shot}}
=
\mathrm{poly}(n,m)\log(1/\varepsilon),
\) which is the polynomial-shot regime of
Theorem~\ref{thm:polynomial-shot-success}. Requiring \(q_y=\Omega(1)\) amounts to demanding constant-shot success as
the problem size grows, whereas efficient sampling requires only an
inverse-polynomial target probability.

\section{Problem structure, constraint statistics and coding-theoretic structure}
\label{app:constraint-statistics-coding}

Let
\(
T:\mathcal X\longrightarrow\mathcal T
\)
be a problem-dependent classical map that records the combinatorial
information relevant to the constraints. If the penalty depends only on
\(T\), so that
\begin{equation}
\label{eq:scheme-alignment-present2}
H_{\mathrm{pen}}(x)=F(T(x)),
\end{equation}
then every penalty level
\[
L_s=\{x\in \mathcal X:H_{\mathrm{pen}}(x)=s\}
\]
is a union of \(T\)-fibres:
\[
L_s
=
\bigcup_{\tau\in\mathcal T:\,F(\tau)=s}
L_\tau^T,
\qquad
L_\tau^T
=
\{x\in \mathcal X:T(x)=\tau\}.
\]
The statistic \(T\) therefore resolves the combinatorial classes on which the
penalty is constant. This is the natural setting in which hard constraints
induce an algebraic structure on the encoded product space that is generally
finer than the scalar penalty levels themselves.

For many constrained optimization problems, the relevant statistic is
genuinely multidimensional. To see how such statistics arise in the present
setting, fix a reference configuration \(y\in \mathcal X=[n]^m\) and define
\[
N^{(y)}_{u,v}(x)
=
\#\left\{
b\in\{1,\ldots,m\}:
y_b=u,\ x_b=v
\right\},
\qquad u,v\in[n].
\]
The row sums are fixed by the symbol populations of the reference
configuration,
\[
\sum_{v=1}^{n}N^{(y)}_{u,v}(x)
=
\#\{b:y_b=u\},
\]
while the column sums recover the symbol histogram of \(x\):
\[
h_v(x)
=
\sum_{u=1}^{n}N^{(y)}_{u,v}(x).
\]
The Hamming distance is obtained from the diagonal,
\[
d(x,y)
=
m-\operatorname{tr}N^{(y)}(x).
\]
Thus the scalar shell index records only the total number of mismatched
coordinates, whereas \(N^{(y)}(x)\) retains the full pattern of symbol
transitions between \(y\) and \(x\). Constraints involving symbol balance,
collisions, capacities, visit multiplicities, or permutation conditions can
therefore be expressed through the histogram, marginal, or contingency-table
data contained in \(N^{(y)}(x)\).

This type of multidimensional constraint statistic appears across several
classical optimization settings. Fixed-margin contingency-table fibres are
central in algebraic statistics and Markov-basis methods
\cite{DiaconisSturmfels1998,DeLoeraOnn2004,
SlavkovicZhuPetrovic2015}, while universality results for integer programming arise through slim three-way transportation programs
\cite{DeLoeraOnn2006}. Overlap-matrix statistics play an analogous role in
random coloring and constraint-satisfaction problems
\cite{AchlioptasCojaOghlan2008,CojaOghlanVilenchik2016}.

Related structures also occur in routing and scheduling. In the standard TSP,
the Miller--Tucker--Zemlin formulation places the binary arc-incidence
variables \((x_{ij})\) in an assignment structure with one outgoing and one
incoming arc at each city before additional ordering constraints eliminate
subtours \cite{MillerTuckerZemlin1960}. Many-visits TSP fixes city-visit
multiplicities \cite{CosmadakisPapadimitriou1984}. Vehicle-routing
formulations combine customer-cover, route-incidence, capacity, and resource
constraints. Configuration and \(N\)-fold formulations for scheduling
similarly fix aggregate job-type, machine-load, and time-slot statistics
\cite{MnichWiese2015,KnopKoutecky2018}.

These examples show that measured product-space samples can be analyzed
through shell indices, penalty levels, contingency tables, visit
multiplicities, collision profiles, route loads, scheduling loads, and other
problem-dependent classical maps. The corresponding empirical distributions
expose which parts of the constraint structure are represented in the
quantum-generated samples. They can also provide the information required by
post-measurement repair operations that restore feasibility while preserving
useful structure already present in an infeasible sample.

This separation is particularly relevant for near-term quantum devices.
Approximate mixer implementations, noise, finite sampling, and imperfect
phase control may generate infeasible samples that nevertheless retain
low-penalty or readily repairable combinatorial structure. A multidimensional
statistic can distinguish different violation patterns that share the same
Hamming distance or total penalty and can therefore support more selective
feasibility-restoring operations. The formal separation between quantum
transport, phase-sensitive sampling, the problem-dependent classical map, and
the subsequent repair operation also makes the contribution of each stage to
the final performance transparent.

\paragraph{Coding-theoretic and association-scheme viewpoint.}

The product space \(\mathcal X=[n]^m\), equipped with Hamming distance, is the
\(n\)-ary Hamming scheme \(H(m,n)\) in the language of coding theory and
association schemes
\cite{MacWilliamsSloane1977,DelsarteLevenshtein1998}. For a fixed reference
configuration \(y\in \mathcal X\), the shells
\[
S_r(y)
=
\{x\in \mathcal X:d(x,y)=r\},
\qquad r=0,\ldots,m,
\]
are its distance classes.

From this viewpoint, the shell-transfer coefficients derived in
Proposition~\ref{prop:shell-transfer-generating} are radial matrix elements of
the product-space transfer operator with respect to the Hamming-distance
partition. The generating functions appearing in the shell reduction are
therefore natural Hamming-scheme objects and are closely related to the
\(n\)-ary Krawtchouk structure familiar from coding theory
\cite{MacWilliamsSloane1977}.

The multidimensional statistics introduced above provide a route for
importing powerful tools from coding theory and algebraic  combinatorics into
the analysis of constrained quantum optimization. They enlarge the class of
product-space structures that can be studied through distance distributions,
orthogonal polynomial expansions, weight enumerators, association schemes,
and related methods.

When a statistic \(T\) exactly resolves the relevant constraint structure,
its fibres define natural relation classes on the product space. In symmetric
settings where these relations are closed under composition, their adjacency
matrices span a Bose--Mesner algebra
\cite{BoseMesner1959}. The scalar Hamming-shell reduction studied in the main
text is then the radial description associated with the ordinary
Hamming-distance classes, while richer statistics can reveal additional
algebraic structure induced by the constraints. We leave the detailed study of hybrid quantum--classical workflows leveraging these algebraic insights, including application-specific feasibility-repair procedures, to future work.

\bibliographystyle{apsrev4-2}
\bibliography{references-prx-quantum}

\begin{thebibliography}{69}%
\makeatletter
\providecommand \@ifxundefined [1]{%
 \@ifx{#1\undefined}
}%
\providecommand \@ifnum [1]{%
 \ifnum #1\expandafter \@firstoftwo
 \else \expandafter \@secondoftwo
 \fi
}%
\providecommand \@ifx [1]{%
 \ifx #1\expandafter \@firstoftwo
 \else \expandafter \@secondoftwo
 \fi
}%
\providecommand \natexlab [1]{#1}%
\providecommand \enquote  [1]{``#1''}%
\providecommand \bibnamefont  [1]{#1}%
\providecommand \bibfnamefont [1]{#1}%
\providecommand \citenamefont [1]{#1}%
\providecommand \href@noop [0]{\@secondoftwo}%
\providecommand \href [0]{\begingroup \@sanitize@url \@href}%
\providecommand \@href[1]{\@@startlink{#1}\@@href}%
\providecommand \@@href[1]{\endgroup#1\@@endlink}%
\providecommand \@sanitize@url [0]{\catcode `\\12\catcode `\$12\catcode `\&12\catcode `\#12\catcode `\^12\catcode `\_12\catcode `\%12\relax}%
\providecommand \@@startlink[1]{}%
\providecommand \@@endlink[0]{}%
\providecommand \url  [0]{\begingroup\@sanitize@url \@url }%
\providecommand \@url [1]{\endgroup\@href {#1}{\urlprefix }}%
\providecommand \urlprefix  [0]{URL }%
\providecommand \Eprint [0]{\href }%
\providecommand \doibase [0]{https://doi.org/}%
\providecommand \selectlanguage [0]{\@gobble}%
\providecommand \bibinfo  [0]{\@secondoftwo}%
\providecommand \bibfield  [0]{\@secondoftwo}%
\providecommand \translation [1]{[#1]}%
\providecommand \BibitemOpen [0]{}%
\providecommand \bibitemStop [0]{}%
\providecommand \bibitemNoStop [0]{.\EOS\space}%
\providecommand \EOS [0]{\spacefactor3000\relax}%
\providecommand \BibitemShut  [1]{\csname bibitem#1\endcsname}%
\let\auto@bib@innerbib\@empty
\bibitem [{\citenamefont {Korte}\ and\ \citenamefont {Vygen}(2018)}]{KorteVygen}%
  \BibitemOpen
  \bibfield  {author} {\bibinfo {author} {\bibfnamefont {B.}~\bibnamefont {Korte}}\ and\ \bibinfo {author} {\bibfnamefont {J.}~\bibnamefont {Vygen}},\ }\href {https://doi.org/10.1007/978-3-662-56039-6} {\emph {\bibinfo {title} {Combinatorial optimization: Theory and algorithms}}},\ \bibinfo {edition} {6th}\ ed.\ (\bibinfo  {publisher} {Springer},\ \bibinfo {address} {Berlin, Heidelberg},\ \bibinfo {year} {2018})\BibitemShut {NoStop}%
\bibitem [{\citenamefont {Schrijver}(2003)}]{Schrijver2003}%
  \BibitemOpen
  \bibfield  {author} {\bibinfo {author} {\bibfnamefont {A.}~\bibnamefont {Schrijver}},\ }\href@noop {} {\emph {\bibinfo {title} {Combinatorial optimization: Polyhedra and efficiency}}},\ \bibinfo {series} {Algorithms and Combinatorics}, Vol.~\bibinfo {volume} {24}\ (\bibinfo  {publisher} {Springer},\ \bibinfo {address} {Berlin, Heidelberg},\ \bibinfo {year} {2003})\BibitemShut {NoStop}%
\bibitem [{\citenamefont {Abbas}\ \emph {et~al.}(2024)\citenamefont {Abbas}, \citenamefont {Ambainis}, \citenamefont {Augustino}, \citenamefont {B{\"a}rtschi}, \citenamefont {Buhrman}, \citenamefont {Coffrin}, \citenamefont {Cortiana}, \citenamefont {Dunjko}, \citenamefont {Egger}, \citenamefont {Elmegreen} \emph {et~al.}}]{abbas2024challenges}%
  \BibitemOpen
  \bibfield  {author} {\bibinfo {author} {\bibfnamefont {A.}~\bibnamefont {Abbas}}, \bibinfo {author} {\bibfnamefont {A.}~\bibnamefont {Ambainis}}, \bibinfo {author} {\bibfnamefont {B.}~\bibnamefont {Augustino}}, \bibinfo {author} {\bibfnamefont {A.}~\bibnamefont {B{\"a}rtschi}}, \bibinfo {author} {\bibfnamefont {H.}~\bibnamefont {Buhrman}}, \bibinfo {author} {\bibfnamefont {C.}~\bibnamefont {Coffrin}}, \bibinfo {author} {\bibfnamefont {G.}~\bibnamefont {Cortiana}}, \bibinfo {author} {\bibfnamefont {V.}~\bibnamefont {Dunjko}}, \bibinfo {author} {\bibfnamefont {D.~J.}\ \bibnamefont {Egger}}, \bibinfo {author} {\bibfnamefont {B.~G.}\ \bibnamefont {Elmegreen}}, \emph {et~al.},\ }\href {https://doi.org/10.1038/s42254-024-00770-9} {\bibfield  {journal} {\bibinfo  {journal} {Nat. Rev. Phys.}\ }\textbf {\bibinfo {volume} {6}},\ \bibinfo {pages} {718} (\bibinfo {year} {2024})}\BibitemShut {NoStop}%
\bibitem [{\citenamefont {Lucas}(2014)}]{Lucas2014}%
  \BibitemOpen
  \bibfield  {author} {\bibinfo {author} {\bibfnamefont {A.}~\bibnamefont {Lucas}},\ }\href {https://doi.org/10.3389/fphy.2014.00005} {\bibfield  {journal} {\bibinfo  {journal} {Front. Phys.}\ }\textbf {\bibinfo {volume} {2}},\ \bibinfo {pages} {5} (\bibinfo {year} {2014})}\BibitemShut {NoStop}%
\bibitem [{\citenamefont {Farhi}\ \emph {et~al.}(2014)\citenamefont {Farhi}, \citenamefont {Goldstone},\ and\ \citenamefont {Gutmann}}]{Farhi2014QAOA}%
  \BibitemOpen
  \bibfield  {author} {\bibinfo {author} {\bibfnamefont {E.}~\bibnamefont {Farhi}}, \bibinfo {author} {\bibfnamefont {J.}~\bibnamefont {Goldstone}},\ and\ \bibinfo {author} {\bibfnamefont {S.}~\bibnamefont {Gutmann}},\ }\href@noop {} {\bibinfo {title} {A quantum approximate optimization algorithm}} (\bibinfo {year} {2014}),\ \Eprint {https://arxiv.org/abs/1411.4028} {arXiv:1411.4028} \BibitemShut {NoStop}%
\bibitem [{\citenamefont {Hadfield}\ \emph {et~al.}(2019)\citenamefont {Hadfield}, \citenamefont {Wang}, \citenamefont {O'Gorman}, \citenamefont {Rieffel}, \citenamefont {Venturelli},\ and\ \citenamefont {Biswas}}]{Hadfield2019AOA}%
  \BibitemOpen
  \bibfield  {author} {\bibinfo {author} {\bibfnamefont {S.}~\bibnamefont {Hadfield}}, \bibinfo {author} {\bibfnamefont {Z.}~\bibnamefont {Wang}}, \bibinfo {author} {\bibfnamefont {B.}~\bibnamefont {O'Gorman}}, \bibinfo {author} {\bibfnamefont {E.~G.}\ \bibnamefont {Rieffel}}, \bibinfo {author} {\bibfnamefont {D.}~\bibnamefont {Venturelli}},\ and\ \bibinfo {author} {\bibfnamefont {R.}~\bibnamefont {Biswas}},\ }\href {https://doi.org/10.3390/a12020034} {\bibfield  {journal} {\bibinfo  {journal} {Algorithms}\ }\textbf {\bibinfo {volume} {12}},\ \bibinfo {pages} {34} (\bibinfo {year} {2019})}\BibitemShut {NoStop}%
\bibitem [{\citenamefont {Onah}\ \emph {et~al.}(2025{\natexlab{a}})\citenamefont {Onah}, \citenamefont {Firt},\ and\ \citenamefont {Michielsen}}]{onahce}%
  \BibitemOpen
  \bibfield  {author} {\bibinfo {author} {\bibfnamefont {C.}~\bibnamefont {Onah}}, \bibinfo {author} {\bibfnamefont {R.}~\bibnamefont {Firt}},\ and\ \bibinfo {author} {\bibfnamefont {K.}~\bibnamefont {Michielsen}},\ }\href@noop {} {\bibinfo {title} {Empirical quantum advantage in constrained optimization from encoded unitary designs}} (\bibinfo {year} {2025}{\natexlab{a}}),\ \Eprint {https://arxiv.org/abs/2511.14296} {arXiv:2511.14296 [cs.ET]} \BibitemShut {NoStop}%
\bibitem [{\citenamefont {Wang}\ \emph {et~al.}(2020)\citenamefont {Wang}, \citenamefont {Rubin}, \citenamefont {Dominy},\ and\ \citenamefont {Rieffel}}]{Wang2020XYMixers}%
  \BibitemOpen
  \bibfield  {author} {\bibinfo {author} {\bibfnamefont {Z.}~\bibnamefont {Wang}}, \bibinfo {author} {\bibfnamefont {N.~C.}\ \bibnamefont {Rubin}}, \bibinfo {author} {\bibfnamefont {J.~M.}\ \bibnamefont {Dominy}},\ and\ \bibinfo {author} {\bibfnamefont {E.~G.}\ \bibnamefont {Rieffel}},\ }\href {https://doi.org/10.1103/PhysRevA.101.012320} {\bibfield  {journal} {\bibinfo  {journal} {Phys. Rev. A}\ }\textbf {\bibinfo {volume} {101}},\ \bibinfo {pages} {012320} (\bibinfo {year} {2020})}\BibitemShut {NoStop}%
\bibitem [{\citenamefont {Hadfield}\ \emph {et~al.}(2023)\citenamefont {Hadfield}, \citenamefont {Hogg},\ and\ \citenamefont {Rieffel}}]{Hadfield_2022}%
  \BibitemOpen
  \bibfield  {author} {\bibinfo {author} {\bibfnamefont {S.}~\bibnamefont {Hadfield}}, \bibinfo {author} {\bibfnamefont {T.}~\bibnamefont {Hogg}},\ and\ \bibinfo {author} {\bibfnamefont {E.~G.}\ \bibnamefont {Rieffel}},\ }\href {https://doi.org/10.1088/2058-9565/aca3ce} {\bibfield  {journal} {\bibinfo  {journal} {Quantum Sci. Technol.}\ }\textbf {\bibinfo {volume} {8}},\ \bibinfo {pages} {015017} (\bibinfo {year} {2023})}\BibitemShut {NoStop}%
\bibitem [{\citenamefont {Fuchs}\ and\ \citenamefont {Bassa}(2022)}]{Fuchs2022ConstrainedMixers}%
  \BibitemOpen
  \bibfield  {author} {\bibinfo {author} {\bibfnamefont {F.~G.}\ \bibnamefont {Fuchs}}\ and\ \bibinfo {author} {\bibfnamefont {R.~P.}\ \bibnamefont {Bassa}},\ }\href {https://doi.org/10.3390/a15060202} {\bibfield  {journal} {\bibinfo  {journal} {Algorithms}\ }\textbf {\bibinfo {volume} {15}},\ \bibinfo {pages} {202} (\bibinfo {year} {2022})}\BibitemShut {NoStop}%
\bibitem [{\citenamefont {Sawaya}\ \emph {et~al.}(2023)\citenamefont {Sawaya}, \citenamefont {Schmitz},\ and\ \citenamefont {Hadfield}}]{sawaya2023encoding}%
  \BibitemOpen
  \bibfield  {author} {\bibinfo {author} {\bibfnamefont {N.~P.}\ \bibnamefont {Sawaya}}, \bibinfo {author} {\bibfnamefont {A.~T.}\ \bibnamefont {Schmitz}},\ and\ \bibinfo {author} {\bibfnamefont {S.}~\bibnamefont {Hadfield}},\ }\href {https://doi.org/10.22331/q-2023-09-14-1111} {\bibfield  {journal} {\bibinfo  {journal} {Quantum}\ }\textbf {\bibinfo {volume} {7}},\ \bibinfo {pages} {1111} (\bibinfo {year} {2023})}\BibitemShut {NoStop}%
\bibitem [{\citenamefont {Onah}\ and\ \citenamefont {Michielsen}(2025)}]{onahfund}%
  \BibitemOpen
  \bibfield  {author} {\bibinfo {author} {\bibfnamefont {C.}~\bibnamefont {Onah}}\ and\ \bibinfo {author} {\bibfnamefont {K.}~\bibnamefont {Michielsen}},\ }\href@noop {} {\bibinfo {title} {Fundamental limitations of {QAOA} on constrained problems and a route to exponential enhancement}} (\bibinfo {year} {2025}),\ \Eprint {https://arxiv.org/abs/2511.17259} {arXiv:2511.17259 [quant-ph]} \BibitemShut {NoStop}%
\bibitem [{\citenamefont {Miller}\ \emph {et~al.}(1960)\citenamefont {Miller}, \citenamefont {Tucker},\ and\ \citenamefont {Zemlin}}]{MillerTuckerZemlin1960}%
  \BibitemOpen
  \bibfield  {author} {\bibinfo {author} {\bibfnamefont {C.~E.}\ \bibnamefont {Miller}}, \bibinfo {author} {\bibfnamefont {A.~W.}\ \bibnamefont {Tucker}},\ and\ \bibinfo {author} {\bibfnamefont {R.~A.}\ \bibnamefont {Zemlin}},\ }\href {https://doi.org/10.1145/321043.321046} {\bibfield  {journal} {\bibinfo  {journal} {J. ACM}\ }\textbf {\bibinfo {volume} {7}},\ \bibinfo {pages} {326} (\bibinfo {year} {1960})}\BibitemShut {NoStop}%
\bibitem [{\citenamefont {Diaconis}\ and\ \citenamefont {Sturmfels}(1998)}]{DiaconisSturmfels1998}%
  \BibitemOpen
  \bibfield  {author} {\bibinfo {author} {\bibfnamefont {P.}~\bibnamefont {Diaconis}}\ and\ \bibinfo {author} {\bibfnamefont {B.}~\bibnamefont {Sturmfels}},\ }\href {https://doi.org/10.1214/aos/1030563990} {\bibfield  {journal} {\bibinfo  {journal} {Ann. Stat.}\ }\textbf {\bibinfo {volume} {26}},\ \bibinfo {pages} {363} (\bibinfo {year} {1998})}\BibitemShut {NoStop}%
\bibitem [{\citenamefont {De~Loera}\ and\ \citenamefont {Onn}(2004)}]{DeLoeraOnn2004}%
  \BibitemOpen
  \bibfield  {author} {\bibinfo {author} {\bibfnamefont {J.~A.}\ \bibnamefont {De~Loera}}\ and\ \bibinfo {author} {\bibfnamefont {S.}~\bibnamefont {Onn}},\ }\href {https://doi.org/10.1137/S0097539702403803} {\bibfield  {journal} {\bibinfo  {journal} {SIAM J. Comput.}\ }\textbf {\bibinfo {volume} {33}},\ \bibinfo {pages} {819} (\bibinfo {year} {2004})}\BibitemShut {NoStop}%
\bibitem [{\citenamefont {De~Loera}\ and\ \citenamefont {Onn}(2006)}]{DeLoeraOnn2006}%
  \BibitemOpen
  \bibfield  {author} {\bibinfo {author} {\bibfnamefont {J.~A.}\ \bibnamefont {De~Loera}}\ and\ \bibinfo {author} {\bibfnamefont {S.}~\bibnamefont {Onn}},\ }\href {https://doi.org/10.1137/040610623} {\bibfield  {journal} {\bibinfo  {journal} {SIAM J. Optim.}\ }\textbf {\bibinfo {volume} {17}},\ \bibinfo {pages} {806} (\bibinfo {year} {2006})}\BibitemShut {NoStop}%
\bibitem [{\citenamefont {Slavkovi{\'c}}\ \emph {et~al.}(2015)\citenamefont {Slavkovi{\'c}}, \citenamefont {Zhu},\ and\ \citenamefont {Petrovi{\'c}}}]{SlavkovicZhuPetrovic2015}%
  \BibitemOpen
  \bibfield  {author} {\bibinfo {author} {\bibfnamefont {A.~B.}\ \bibnamefont {Slavkovi{\'c}}}, \bibinfo {author} {\bibfnamefont {X.}~\bibnamefont {Zhu}},\ and\ \bibinfo {author} {\bibfnamefont {S.}~\bibnamefont {Petrovi{\'c}}},\ }\href {https://doi.org/10.1007/s10463-014-0471-z} {\bibfield  {journal} {\bibinfo  {journal} {Ann. Inst. Stat. Math.}\ }\textbf {\bibinfo {volume} {67}},\ \bibinfo {pages} {621} (\bibinfo {year} {2015})}\BibitemShut {NoStop}%
\bibitem [{\citenamefont {Hen}\ and\ \citenamefont {Spedalieri}(2016)}]{HenSpedalieri2016}%
  \BibitemOpen
  \bibfield  {author} {\bibinfo {author} {\bibfnamefont {I.}~\bibnamefont {Hen}}\ and\ \bibinfo {author} {\bibfnamefont {F.~M.}\ \bibnamefont {Spedalieri}},\ }\href {https://doi.org/10.1103/PhysRevApplied.5.034007} {\bibfield  {journal} {\bibinfo  {journal} {Phys. Rev. Applied}\ }\textbf {\bibinfo {volume} {5}},\ \bibinfo {pages} {034007} (\bibinfo {year} {2016})}\BibitemShut {NoStop}%
\bibitem [{\citenamefont {Hen}\ and\ \citenamefont {Sarandy}(2016)}]{HenSarandy2016}%
  \BibitemOpen
  \bibfield  {author} {\bibinfo {author} {\bibfnamefont {I.}~\bibnamefont {Hen}}\ and\ \bibinfo {author} {\bibfnamefont {M.~S.}\ \bibnamefont {Sarandy}},\ }\href {https://doi.org/10.1103/PhysRevA.93.062312} {\bibfield  {journal} {\bibinfo  {journal} {Phys. Rev. A}\ }\textbf {\bibinfo {volume} {93}},\ \bibinfo {pages} {062312} (\bibinfo {year} {2016})}\BibitemShut {NoStop}%
\bibitem [{\citenamefont {MacWilliams}\ and\ \citenamefont {Sloane}(1977)}]{MacWilliamsSloane1977}%
  \BibitemOpen
  \bibfield  {author} {\bibinfo {author} {\bibfnamefont {F.~J.}\ \bibnamefont {MacWilliams}}\ and\ \bibinfo {author} {\bibfnamefont {N.~J.~A.}\ \bibnamefont {Sloane}},\ }\href@noop {} {\emph {\bibinfo {title} {The theory of error-correcting codes}}}\ (\bibinfo  {publisher} {North-Holland},\ \bibinfo {address} {Amsterdam},\ \bibinfo {year} {1977})\BibitemShut {NoStop}%
\bibitem [{\citenamefont {Delsarte}\ and\ \citenamefont {Levenshtein}(1998)}]{DelsarteLevenshtein1998}%
  \BibitemOpen
  \bibfield  {author} {\bibinfo {author} {\bibfnamefont {P.}~\bibnamefont {Delsarte}}\ and\ \bibinfo {author} {\bibfnamefont {V.~I.}\ \bibnamefont {Levenshtein}},\ }\href {https://doi.org/10.1109/18.720545} {\bibfield  {journal} {\bibinfo  {journal} {IEEE Trans. Inf. Theory}\ }\textbf {\bibinfo {volume} {44}},\ \bibinfo {pages} {2477} (\bibinfo {year} {1998})}\BibitemShut {NoStop}%
\bibitem [{\citenamefont {Moore}\ and\ \citenamefont {Russell}(2002)}]{MooreRussell2002}%
  \BibitemOpen
  \bibfield  {author} {\bibinfo {author} {\bibfnamefont {C.}~\bibnamefont {Moore}}\ and\ \bibinfo {author} {\bibfnamefont {A.}~\bibnamefont {Russell}},\ }in\ \href {https://doi.org/10.1007/3-540-45726-7_14} {\emph {\bibinfo {booktitle} {Randomization and Approximation Techniques in Computer Science: 6th International Workshop, RANDOM 2002, Cambridge, MA, USA, September 13--15, 2002, Proceedings}}},\ \bibinfo {series} {Lecture Notes in Computer Science}, Vol.\ \bibinfo {volume} {2483},\ \bibinfo {editor} {edited by\ \bibinfo {editor} {\bibfnamefont {J.~D.~P.}\ \bibnamefont {Rolim}}\ and\ \bibinfo {editor} {\bibfnamefont {S.}~\bibnamefont {Vadhan}}}\ (\bibinfo  {publisher} {Springer},\ \bibinfo {address} {Berlin, Heidelberg},\ \bibinfo {year} {2002})\ pp.\ \bibinfo {pages} {164--178}\BibitemShut {NoStop}%
\bibitem [{\citenamefont {Childs}\ and\ \citenamefont {Goldstone}(2004)}]{Childs2004SpatialSearch}%
  \BibitemOpen
  \bibfield  {author} {\bibinfo {author} {\bibfnamefont {A.~M.}\ \bibnamefont {Childs}}\ and\ \bibinfo {author} {\bibfnamefont {J.}~\bibnamefont {Goldstone}},\ }\href {https://doi.org/10.1103/PhysRevA.70.022314} {\bibfield  {journal} {\bibinfo  {journal} {Phys. Rev. A}\ }\textbf {\bibinfo {volume} {70}},\ \bibinfo {pages} {022314} (\bibinfo {year} {2004})}\BibitemShut {NoStop}%
\bibitem [{\citenamefont {Childs}(2010)}]{Childs2008WalkCorrespondence}%
  \BibitemOpen
  \bibfield  {author} {\bibinfo {author} {\bibfnamefont {A.~M.}\ \bibnamefont {Childs}},\ }\href {https://doi.org/10.1007/s00220-009-0930-1} {\bibfield  {journal} {\bibinfo  {journal} {Commun. Math. Phys.}\ }\textbf {\bibinfo {volume} {294}},\ \bibinfo {pages} {581} (\bibinfo {year} {2010})}\BibitemShut {NoStop}%
\bibitem [{\citenamefont {Krovi}\ and\ \citenamefont {Brun}(2007)}]{KroviBrun2007Quotient}%
  \BibitemOpen
  \bibfield  {author} {\bibinfo {author} {\bibfnamefont {H.}~\bibnamefont {Krovi}}\ and\ \bibinfo {author} {\bibfnamefont {T.~A.}\ \bibnamefont {Brun}},\ }\href {https://doi.org/10.1103/PhysRevA.75.062332} {\bibfield  {journal} {\bibinfo  {journal} {Phys. Rev. A}\ }\textbf {\bibinfo {volume} {75}},\ \bibinfo {pages} {062332} (\bibinfo {year} {2007})}\BibitemShut {NoStop}%
\bibitem [{\citenamefont {Farhi}\ \emph {et~al.}(2020{\natexlab{a}})\citenamefont {Farhi}, \citenamefont {Gamarnik},\ and\ \citenamefont {Gutmann}}]{FarhiGamarnikGutmannTypical}%
  \BibitemOpen
  \bibfield  {author} {\bibinfo {author} {\bibfnamefont {E.}~\bibnamefont {Farhi}}, \bibinfo {author} {\bibfnamefont {D.}~\bibnamefont {Gamarnik}},\ and\ \bibinfo {author} {\bibfnamefont {S.}~\bibnamefont {Gutmann}},\ }\href@noop {} {\bibinfo {title} {The quantum approximate optimization algorithm needs to see the whole graph: A typical case}} (\bibinfo {year} {2020}{\natexlab{a}}),\ \Eprint {https://arxiv.org/abs/2004.09002} {arXiv:2004.09002} \BibitemShut {NoStop}%
\bibitem [{\citenamefont {Farhi}\ \emph {et~al.}(2020{\natexlab{b}})\citenamefont {Farhi}, \citenamefont {Gamarnik},\ and\ \citenamefont {Gutmann}}]{FarhiGamarnikGutmannWorstCase}%
  \BibitemOpen
  \bibfield  {author} {\bibinfo {author} {\bibfnamefont {E.}~\bibnamefont {Farhi}}, \bibinfo {author} {\bibfnamefont {D.}~\bibnamefont {Gamarnik}},\ and\ \bibinfo {author} {\bibfnamefont {S.}~\bibnamefont {Gutmann}},\ }\href@noop {} {\bibinfo {title} {The quantum approximate optimization algorithm needs to see the whole graph: Worst case examples}} (\bibinfo {year} {2020}{\natexlab{b}}),\ \Eprint {https://arxiv.org/abs/2005.08747} {arXiv:2005.08747} \BibitemShut {NoStop}%
\bibitem [{\citenamefont {Basso}\ \emph {et~al.}(2022)\citenamefont {Basso}, \citenamefont {Farhi}, \citenamefont {Marwaha}, \citenamefont {Villalonga},\ and\ \citenamefont {Zhou}}]{BassoFarhiMarwahaVillalongaZhou2022}%
  \BibitemOpen
  \bibfield  {author} {\bibinfo {author} {\bibfnamefont {J.}~\bibnamefont {Basso}}, \bibinfo {author} {\bibfnamefont {E.}~\bibnamefont {Farhi}}, \bibinfo {author} {\bibfnamefont {K.}~\bibnamefont {Marwaha}}, \bibinfo {author} {\bibfnamefont {B.}~\bibnamefont {Villalonga}},\ and\ \bibinfo {author} {\bibfnamefont {L.}~\bibnamefont {Zhou}},\ }in\ \href {https://doi.org/10.4230/LIPIcs.TQC.2022.7} {\emph {\bibinfo {booktitle} {17th Conference on the Theory of Quantum Computation, Communication and Cryptography (TQC 2022)}}},\ \bibinfo {series} {Leibniz International Proceedings in Informatics}, Vol.\ \bibinfo {volume} {232}\ (\bibinfo {year} {2022})\ pp.\ \bibinfo {pages} {7:1--7:21}\BibitemShut {NoStop}%
\bibitem [{\citenamefont {Anshu}\ and\ \citenamefont {Metger}(2023)}]{AnshuMetger2023}%
  \BibitemOpen
  \bibfield  {author} {\bibinfo {author} {\bibfnamefont {A.}~\bibnamefont {Anshu}}\ and\ \bibinfo {author} {\bibfnamefont {T.}~\bibnamefont {Metger}},\ }\href {https://doi.org/10.22331/q-2023-05-11-999} {\bibfield  {journal} {\bibinfo  {journal} {Quantum}\ }\textbf {\bibinfo {volume} {7}},\ \bibinfo {pages} {999} (\bibinfo {year} {2023})}\BibitemShut {NoStop}%
\bibitem [{\citenamefont {Akshay}\ \emph {et~al.}(2021)\citenamefont {Akshay}, \citenamefont {Rabinovich}, \citenamefont {Campos},\ and\ \citenamefont {Biamonte}}]{AkshayRabinovichCamposBiamonte2021}%
  \BibitemOpen
  \bibfield  {author} {\bibinfo {author} {\bibfnamefont {V.}~\bibnamefont {Akshay}}, \bibinfo {author} {\bibfnamefont {D.}~\bibnamefont {Rabinovich}}, \bibinfo {author} {\bibfnamefont {E.}~\bibnamefont {Campos}},\ and\ \bibinfo {author} {\bibfnamefont {J.}~\bibnamefont {Biamonte}},\ }\href {https://doi.org/10.1103/PhysRevA.104.L010401} {\bibfield  {journal} {\bibinfo  {journal} {Phys. Rev. A}\ }\textbf {\bibinfo {volume} {104}},\ \bibinfo {pages} {L010401} (\bibinfo {year} {2021})}\BibitemShut {NoStop}%
\bibitem [{\citenamefont {Wurtz}\ and\ \citenamefont {Love}(2021)}]{WurtzLove2021}%
  \BibitemOpen
  \bibfield  {author} {\bibinfo {author} {\bibfnamefont {J.}~\bibnamefont {Wurtz}}\ and\ \bibinfo {author} {\bibfnamefont {P.~J.}\ \bibnamefont {Love}},\ }\href {https://doi.org/10.1103/PhysRevA.103.042612} {\bibfield  {journal} {\bibinfo  {journal} {Phys. Rev. A}\ }\textbf {\bibinfo {volume} {103}},\ \bibinfo {pages} {042612} (\bibinfo {year} {2021})}\BibitemShut {NoStop}%
\bibitem [{\citenamefont {P{\'e}rez-Salinas}\ \emph {et~al.}(2024)\citenamefont {P{\'e}rez-Salinas}, \citenamefont {Wang},\ and\ \citenamefont {Bonet-Monroig}}]{PerezSalinasWangBonetMonroig2024}%
  \BibitemOpen
  \bibfield  {author} {\bibinfo {author} {\bibfnamefont {A.}~\bibnamefont {P{\'e}rez-Salinas}}, \bibinfo {author} {\bibfnamefont {H.}~\bibnamefont {Wang}},\ and\ \bibinfo {author} {\bibfnamefont {X.}~\bibnamefont {Bonet-Monroig}},\ }\href {https://doi.org/10.1038/s41534-024-00819-8} {\bibfield  {journal} {\bibinfo  {journal} {npj Quantum Inf.}\ }\textbf {\bibinfo {volume} {10}},\ \bibinfo {pages} {27} (\bibinfo {year} {2024})}\BibitemShut {NoStop}%
\bibitem [{\citenamefont {Onah}\ and\ \citenamefont {Michielsen}(2026{\natexlab{a}})}]{onahfinite}%
  \BibitemOpen
  \bibfield  {author} {\bibinfo {author} {\bibfnamefont {C.}~\bibnamefont {Onah}}\ and\ \bibinfo {author} {\bibfnamefont {K.}~\bibnamefont {Michielsen}},\ }\href@noop {} {\bibinfo {title} {Finite-depth, finite-shot guarantees for constrained quantum optimization via {Fej\'er} filtering}} (\bibinfo {year} {2026}{\natexlab{a}}),\ \Eprint {https://arxiv.org/abs/2603.01809} {arXiv:2603.01809 [quant-ph]} \BibitemShut {NoStop}%
\bibitem [{\citenamefont {Li}\ \emph {et~al.}(2021)\citenamefont {Li}, \citenamefont {Ding},\ and\ \citenamefont {Xie}}]{Li2021CoDesign}%
  \BibitemOpen
  \bibfield  {author} {\bibinfo {author} {\bibfnamefont {G.}~\bibnamefont {Li}}, \bibinfo {author} {\bibfnamefont {Y.}~\bibnamefont {Ding}},\ and\ \bibinfo {author} {\bibfnamefont {Y.}~\bibnamefont {Xie}},\ }in\ \href {https://doi.org/10.1145/3477206.3477464} {\emph {\bibinfo {booktitle} {ICCAD '21: IEEE/ACM International Conference on Computer-Aided Design}}}\ (\bibinfo {year} {2021})\BibitemShut {NoStop}%
\bibitem [{\citenamefont {Tsvelikhovskiy}\ \emph {et~al.}(2023)\citenamefont {Tsvelikhovskiy}, \citenamefont {Safro},\ and\ \citenamefont {Alexeev}}]{tsvelikhovskiy2024symmetries}%
  \BibitemOpen
  \bibfield  {author} {\bibinfo {author} {\bibfnamefont {B.}~\bibnamefont {Tsvelikhovskiy}}, \bibinfo {author} {\bibfnamefont {I.}~\bibnamefont {Safro}},\ and\ \bibinfo {author} {\bibfnamefont {Y.}~\bibnamefont {Alexeev}},\ }\href@noop {} {\bibinfo {title} {Symmetries and dimension reduction in quantum approximate optimization algorithm}} (\bibinfo {year} {2023}),\ \Eprint {https://arxiv.org/abs/2309.13787} {arXiv:2309.13787 [quant-ph]} \BibitemShut {NoStop}%
\bibitem [{\citenamefont {Onah}\ and\ \citenamefont {Michielsen}(2026{\natexlab{b}})}]{onah2026cvrp}%
  \BibitemOpen
  \bibfield  {author} {\bibinfo {author} {\bibfnamefont {C.}~\bibnamefont {Onah}}\ and\ \bibinfo {author} {\bibfnamefont {K.}~\bibnamefont {Michielsen}},\ }\href@noop {} {\bibinfo {title} {Optimal, qubit-efficient quantum vehicle routing via colored-permutations}} (\bibinfo {year} {2026}{\natexlab{b}}),\ \Eprint {https://arxiv.org/abs/2604.04570} {arXiv:2604.04570 [quant-ph]} \BibitemShut {NoStop}%
\bibitem [{\citenamefont {Toth}\ and\ \citenamefont {Vigo}(2014)}]{TothVigo2014VRP}%
  \BibitemOpen
  \bibinfo {editor} {\bibfnamefont {P.}~\bibnamefont {Toth}}\ and\ \bibinfo {editor} {\bibfnamefont {D.}~\bibnamefont {Vigo}},\ eds.,\ \href {https://doi.org/10.1137/1.9781611973587} {\emph {\bibinfo {title} {Vehicle routing: Problems, methods, and applications}}},\ \bibinfo {edition} {2nd}\ ed.\ (\bibinfo  {publisher} {SIAM},\ \bibinfo {address} {Philadelphia},\ \bibinfo {year} {2014})\BibitemShut {NoStop}%
\bibitem [{\citenamefont {Cook}\ \emph {et~al.}(2020)\citenamefont {Cook}, \citenamefont {Eidenbenz},\ and\ \citenamefont {B{\"a}rtschi}}]{Cook2019MaxKVertexCover}%
  \BibitemOpen
  \bibfield  {author} {\bibinfo {author} {\bibfnamefont {J.}~\bibnamefont {Cook}}, \bibinfo {author} {\bibfnamefont {S.}~\bibnamefont {Eidenbenz}},\ and\ \bibinfo {author} {\bibfnamefont {A.}~\bibnamefont {B{\"a}rtschi}},\ }in\ \href {https://doi.org/10.1109/QCE49297.2020.00021} {\emph {\bibinfo {booktitle} {2020 IEEE International Conference on Quantum Computing and Engineering (QCE)}}}\ (\bibinfo  {publisher} {IEEE},\ \bibinfo {year} {2020})\ pp.\ \bibinfo {pages} {83--92}\BibitemShut {NoStop}%
\bibitem [{\citenamefont {B{\"a}rtschi}\ and\ \citenamefont {Eidenbenz}(2019)}]{B_rtschi_2019}%
  \BibitemOpen
  \bibfield  {author} {\bibinfo {author} {\bibfnamefont {A.}~\bibnamefont {B{\"a}rtschi}}\ and\ \bibinfo {author} {\bibfnamefont {S.}~\bibnamefont {Eidenbenz}},\ }in\ \href {https://doi.org/10.1007/978-3-030-25027-0_9} {\emph {\bibinfo {booktitle} {Fundamentals of Computation Theory}}}\ (\bibinfo  {publisher} {Springer International Publishing},\ \bibinfo {address} {Cham},\ \bibinfo {year} {2019})\ pp.\ \bibinfo {pages} {126--139}\BibitemShut {NoStop}%
\bibitem [{\citenamefont {Kordonowy}\ and\ \citenamefont {Leipold}(2026)}]{KordonowyLeipold2026XYLie}%
  \BibitemOpen
  \bibfield  {author} {\bibinfo {author} {\bibfnamefont {S.}~\bibnamefont {Kordonowy}}\ and\ \bibinfo {author} {\bibfnamefont {H.}~\bibnamefont {Leipold}},\ }\href {https://doi.org/10.1038/s41534-026-01192-4} {\bibfield  {journal} {\bibinfo  {journal} {npj Quantum Inf.}\ }\textbf {\bibinfo {volume} {12}},\ \bibinfo {pages} {61} (\bibinfo {year} {2026})}\BibitemShut {NoStop}%
\bibitem [{\citenamefont {Leipold}\ \emph {et~al.}(2026)\citenamefont {Leipold}, \citenamefont {Spedalieri}, \citenamefont {Hadfield},\ and\ \citenamefont {Rieffel}}]{leipold2026imposing}%
  \BibitemOpen
  \bibfield  {author} {\bibinfo {author} {\bibfnamefont {H.}~\bibnamefont {Leipold}}, \bibinfo {author} {\bibfnamefont {F.}~\bibnamefont {Spedalieri}}, \bibinfo {author} {\bibfnamefont {S.}~\bibnamefont {Hadfield}},\ and\ \bibinfo {author} {\bibfnamefont {E.~G.}\ \bibnamefont {Rieffel}},\ }\bibfield  {journal} {\bibinfo  {journal} {ACM Trans. Quantum Comput.}\ }\href {https://doi.org/10.1145/3821414} {10.1145/3821414} (\bibinfo {year} {2026})\BibitemShut {NoStop}%
\bibitem [{\citenamefont {He}\ \emph {et~al.}(2023)\citenamefont {He}, \citenamefont {Shaydulin}, \citenamefont {Chakrabarti}, \citenamefont {Herman}, \citenamefont {Li}, \citenamefont {Sun},\ and\ \citenamefont {Pistoia}}]{He2023A}%
  \BibitemOpen
  \bibfield  {author} {\bibinfo {author} {\bibfnamefont {Z.}~\bibnamefont {He}}, \bibinfo {author} {\bibfnamefont {R.}~\bibnamefont {Shaydulin}}, \bibinfo {author} {\bibfnamefont {S.}~\bibnamefont {Chakrabarti}}, \bibinfo {author} {\bibfnamefont {D.}~\bibnamefont {Herman}}, \bibinfo {author} {\bibfnamefont {C.}~\bibnamefont {Li}}, \bibinfo {author} {\bibfnamefont {Y.}~\bibnamefont {Sun}},\ and\ \bibinfo {author} {\bibfnamefont {M.}~\bibnamefont {Pistoia}},\ }\bibfield  {journal} {\bibinfo  {journal} {npj Quantum Inf.}\ }\textbf {\bibinfo {volume} {9}},\ \href {https://doi.org/10.1038/s41534-023-00787-5} {10.1038/s41534-023-00787-5} (\bibinfo {year} {2023})\BibitemShut {NoStop}%
\bibitem [{\citenamefont {Awasthi}\ \emph {et~al.}(2026)\citenamefont {Awasthi}, \citenamefont {Hess}, \citenamefont {Lomadze}, \citenamefont {B{\"a}r},\ and\ \citenamefont {Biefel}}]{awasthi2026constraintpreservingxymixerstrotterized}%
  \BibitemOpen
  \bibfield  {author} {\bibinfo {author} {\bibfnamefont {A.}~\bibnamefont {Awasthi}}, \bibinfo {author} {\bibfnamefont {M.}~\bibnamefont {Hess}}, \bibinfo {author} {\bibfnamefont {S.}~\bibnamefont {Lomadze}}, \bibinfo {author} {\bibfnamefont {F.}~\bibnamefont {B{\"a}r}},\ and\ \bibinfo {author} {\bibfnamefont {C.}~\bibnamefont {Biefel}},\ }\href@noop {} {\bibinfo {title} {Constraint preserving {$XY$}-mixers under {Trotterized} adiabatic evolution}} (\bibinfo {year} {2026}),\ \Eprint {https://arxiv.org/abs/2605.02465} {arXiv:2605.02465 [quant-ph]} \BibitemShut {NoStop}%
\bibitem [{\citenamefont {Ruan}\ \emph {et~al.}(2025)\citenamefont {Ruan}, \citenamefont {Chen}, \citenamefont {Li}, \citenamefont {Yang}, \citenamefont {Yuan}, \citenamefont {Xue}, \citenamefont {Li},\ and\ \citenamefont {Liu}}]{Ruan2025XYCD}%
  \BibitemOpen
  \bibfield  {author} {\bibinfo {author} {\bibfnamefont {Y.}~\bibnamefont {Ruan}}, \bibinfo {author} {\bibfnamefont {P.}~\bibnamefont {Chen}}, \bibinfo {author} {\bibfnamefont {Q.}~\bibnamefont {Li}}, \bibinfo {author} {\bibfnamefont {L.}~\bibnamefont {Yang}}, \bibinfo {author} {\bibfnamefont {Z.}~\bibnamefont {Yuan}}, \bibinfo {author} {\bibfnamefont {X.}~\bibnamefont {Xue}}, \bibinfo {author} {\bibfnamefont {X.}~\bibnamefont {Li}},\ and\ \bibinfo {author} {\bibfnamefont {Z.}~\bibnamefont {Liu}},\ }\href {https://doi.org/10.1103/PhysRevResearch.7.013243} {\bibfield  {journal} {\bibinfo  {journal} {Phys. Rev. Res.}\ }\textbf {\bibinfo {volume} {7}},\ \bibinfo {pages} {013243} (\bibinfo {year} {2025})}\BibitemShut {NoStop}%
\bibitem [{\citenamefont {Yu}\ \emph {et~al.}(2026)\citenamefont {Yu}, \citenamefont {Muleady}, \citenamefont {Wang}, \citenamefont {Schine}, \citenamefont {Gorshkov},\ and\ \citenamefont {Childs}}]{yu2026efficient}%
  \BibitemOpen
  \bibfield  {author} {\bibinfo {author} {\bibfnamefont {J.}~\bibnamefont {Yu}}, \bibinfo {author} {\bibfnamefont {S.~R.}\ \bibnamefont {Muleady}}, \bibinfo {author} {\bibfnamefont {Y.-X.}\ \bibnamefont {Wang}}, \bibinfo {author} {\bibfnamefont {N.}~\bibnamefont {Schine}}, \bibinfo {author} {\bibfnamefont {A.~V.}\ \bibnamefont {Gorshkov}},\ and\ \bibinfo {author} {\bibfnamefont {A.~M.}\ \bibnamefont {Childs}},\ }\href {https://doi.org/10.1103/9gjk-rgql} {\bibfield  {journal} {\bibinfo  {journal} {Phys. Rev. Lett.}\ }\textbf {\bibinfo {volume} {136}},\ \bibinfo {pages} {030601} (\bibinfo {year} {2026})}\BibitemShut {NoStop}%
\bibitem [{\citenamefont {Garey}\ and\ \citenamefont {Johnson}(1979)}]{GareyJohnson1979}%
  \BibitemOpen
  \bibfield  {author} {\bibinfo {author} {\bibfnamefont {M.~R.}\ \bibnamefont {Garey}}\ and\ \bibinfo {author} {\bibfnamefont {D.~S.}\ \bibnamefont {Johnson}},\ }\href@noop {} {\emph {\bibinfo {title} {Computers and intractability: A guide to the theory of {NP}-completeness}}}\ (\bibinfo  {publisher} {W.~H.~Freeman and Company},\ \bibinfo {address} {New York},\ \bibinfo {year} {1979})\BibitemShut {NoStop}%
\bibitem [{\citenamefont {Loiola}\ \emph {et~al.}(2007)\citenamefont {Loiola}, \citenamefont {de~Abreu}, \citenamefont {Boaventura-Netto}, \citenamefont {Hahn},\ and\ \citenamefont {Querido}}]{Loiola2007QAPSurvey}%
  \BibitemOpen
  \bibfield  {author} {\bibinfo {author} {\bibfnamefont {E.~M.}\ \bibnamefont {Loiola}}, \bibinfo {author} {\bibfnamefont {N.~M.~M.}\ \bibnamefont {de~Abreu}}, \bibinfo {author} {\bibfnamefont {P.~O.}\ \bibnamefont {Boaventura-Netto}}, \bibinfo {author} {\bibfnamefont {P.}~\bibnamefont {Hahn}},\ and\ \bibinfo {author} {\bibfnamefont {T.}~\bibnamefont {Querido}},\ }\href {https://doi.org/10.1016/j.ejor.2005.09.032} {\bibfield  {journal} {\bibinfo  {journal} {Eur. J. Oper. Res.}\ }\textbf {\bibinfo {volume} {176}},\ \bibinfo {pages} {657} (\bibinfo {year} {2007})}\BibitemShut {NoStop}%
\bibitem [{\citenamefont {Onah}\ \emph {et~al.}(2025{\natexlab{b}})\citenamefont {Onah}, \citenamefont {Misciasci}, \citenamefont {Othmer},\ and\ \citenamefont {Michielsen}}]{onah2025waas}%
  \BibitemOpen
  \bibfield  {author} {\bibinfo {author} {\bibfnamefont {C.}~\bibnamefont {Onah}}, \bibinfo {author} {\bibfnamefont {N.}~\bibnamefont {Misciasci}}, \bibinfo {author} {\bibfnamefont {C.}~\bibnamefont {Othmer}},\ and\ \bibinfo {author} {\bibfnamefont {K.}~\bibnamefont {Michielsen}},\ }in\ \href {https://doi.org/10.1109/QCE65121.2025.00235} {\emph {\bibinfo {booktitle} {2025 IEEE International Conference on Quantum Computing and Engineering (QCE)}}},\ Vol.~\bibinfo {volume} {01}\ (\bibinfo {year} {2025})\ pp.\ \bibinfo {pages} {2149--2160}\BibitemShut {NoStop}%
\bibitem [{\citenamefont {Sahni}\ and\ \citenamefont {Gonzalez}(1976)}]{SahniGonzalez1976GAP}%
  \BibitemOpen
  \bibfield  {author} {\bibinfo {author} {\bibfnamefont {S.}~\bibnamefont {Sahni}}\ and\ \bibinfo {author} {\bibfnamefont {T.}~\bibnamefont {Gonzalez}},\ }\href {https://doi.org/10.1145/321958.321975} {\bibfield  {journal} {\bibinfo  {journal} {J. ACM}\ }\textbf {\bibinfo {volume} {23}},\ \bibinfo {pages} {555} (\bibinfo {year} {1976})}\BibitemShut {NoStop}%
\bibitem [{\citenamefont {Karp}(1972)}]{Karp1972}%
  \BibitemOpen
  \bibfield  {author} {\bibinfo {author} {\bibfnamefont {R.~M.}\ \bibnamefont {Karp}},\ }in\ \href@noop {} {\emph {\bibinfo {booktitle} {Complexity of Computer Computations}}},\ \bibinfo {editor} {edited by\ \bibinfo {editor} {\bibfnamefont {R.~E.}\ \bibnamefont {Miller}}\ and\ \bibinfo {editor} {\bibfnamefont {J.~W.}\ \bibnamefont {Thatcher}}}\ (\bibinfo  {publisher} {Plenum},\ \bibinfo {address} {New York},\ \bibinfo {year} {1972})\ pp.\ \bibinfo {pages} {85--103}\BibitemShut {NoStop}%
\bibitem [{\citenamefont {Michielsen}\ \emph {et~al.}(2026)\citenamefont {Michielsen}, \citenamefont {Hadfield},\ and\ \citenamefont {Onah}}]{Michielsen2026GeometryInterferenceData}%
  \BibitemOpen
  \bibfield  {author} {\bibinfo {author} {\bibfnamefont {K.}~\bibnamefont {Michielsen}}, \bibinfo {author} {\bibfnamefont {S.}~\bibnamefont {Hadfield}},\ and\ \bibinfo {author} {\bibfnamefont {C.}~\bibnamefont {Onah}},\ }\href {https://doi.org/10.5281/zenodo.21302533} {\bibinfo {title} {Data for ``{Separating Geometry From Interference in Constrained Quantum Optimization}''}},\ \bibinfo {howpublished} {Zenodo dataset} (\bibinfo {year} {2026}),\ \bibinfo {note} {doi: 10.5281/zenodo.21302533}\BibitemShut {NoStop}%
\bibitem [{\citenamefont {Fran{\c{c}}a}(2018)}]{Franca2018PerfectSampling}%
  \BibitemOpen
  \bibfield  {author} {\bibinfo {author} {\bibfnamefont {D.~S.}\ \bibnamefont {Fran{\c{c}}a}},\ }\href {https://doi.org/10.26421/QIC18.5-6-1} {\bibfield  {journal} {\bibinfo  {journal} {Quantum Inf. Comput.}\ }\textbf {\bibinfo {volume} {18}},\ \bibinfo {pages} {361} (\bibinfo {year} {2018})}\BibitemShut {NoStop}%
\bibitem [{\citenamefont {Quetschlich}\ \emph {et~al.}(2024)\citenamefont {Quetschlich}, \citenamefont {Kiwit}, \citenamefont {Wolf}, \citenamefont {Riofrio}, \citenamefont {Burgholzer}, \citenamefont {Luckow},\ and\ \citenamefont {Wille}}]{Quetschlich2024ApplicationAware}%
  \BibitemOpen
  \bibfield  {author} {\bibinfo {author} {\bibfnamefont {N.}~\bibnamefont {Quetschlich}}, \bibinfo {author} {\bibfnamefont {F.~J.}\ \bibnamefont {Kiwit}}, \bibinfo {author} {\bibfnamefont {M.~A.}\ \bibnamefont {Wolf}}, \bibinfo {author} {\bibfnamefont {C.~A.}\ \bibnamefont {Riofrio}}, \bibinfo {author} {\bibfnamefont {L.}~\bibnamefont {Burgholzer}}, \bibinfo {author} {\bibfnamefont {A.}~\bibnamefont {Luckow}},\ and\ \bibinfo {author} {\bibfnamefont {R.}~\bibnamefont {Wille}},\ }in\ \href {https://doi.org/10.1109/QSW62656.2024.00028} {\emph {\bibinfo {booktitle} {2024 IEEE International Conference on Quantum Software}}}\ (\bibinfo {year} {2024})\ pp.\ \bibinfo {pages} {135--142}\BibitemShut {NoStop}%
\bibitem [{\citenamefont {Ji}\ \emph {et~al.}(2025)\citenamefont {Ji}, \citenamefont {Chen}, \citenamefont {Polian},\ and\ \citenamefont {Ban}}]{Ji2025AOQMAP}%
  \BibitemOpen
  \bibfield  {author} {\bibinfo {author} {\bibfnamefont {Y.}~\bibnamefont {Ji}}, \bibinfo {author} {\bibfnamefont {X.}~\bibnamefont {Chen}}, \bibinfo {author} {\bibfnamefont {I.}~\bibnamefont {Polian}},\ and\ \bibinfo {author} {\bibfnamefont {Y.}~\bibnamefont {Ban}},\ }\href {https://doi.org/10.1103/PhysRevApplied.23.034022} {\bibfield  {journal} {\bibinfo  {journal} {Phys. Rev. Applied}\ }\textbf {\bibinfo {volume} {23}},\ \bibinfo {pages} {034022} (\bibinfo {year} {2025})}\BibitemShut {NoStop}%
\bibitem [{\citenamefont {Zhu}\ \emph {et~al.}(2024)\citenamefont {Zhu}, \citenamefont {Zhou}, \citenamefont {Cheng}, \citenamefont {Jin}, \citenamefont {Li}, \citenamefont {Niu},\ and\ \citenamefont {Liang}}]{Zhu2024Coqa}%
  \BibitemOpen
  \bibfield  {author} {\bibinfo {author} {\bibfnamefont {Y.}~\bibnamefont {Zhu}}, \bibinfo {author} {\bibfnamefont {Y.}~\bibnamefont {Zhou}}, \bibinfo {author} {\bibfnamefont {J.}~\bibnamefont {Cheng}}, \bibinfo {author} {\bibfnamefont {Y.}~\bibnamefont {Jin}}, \bibinfo {author} {\bibfnamefont {B.}~\bibnamefont {Li}}, \bibinfo {author} {\bibfnamefont {S.}~\bibnamefont {Niu}},\ and\ \bibinfo {author} {\bibfnamefont {Z.}~\bibnamefont {Liang}},\ }\href@noop {} {\bibinfo {title} {{CoQA}: Blazing fast compiler optimizations for {QAOA}}} (\bibinfo {year} {2024}),\ \Eprint {https://arxiv.org/abs/2408.08365} {arXiv:2408.08365 [quant-ph]} \BibitemShut {NoStop}%
\bibitem [{\citenamefont {Maciejewski}\ \emph {et~al.}(2025)\citenamefont {Maciejewski}, \citenamefont {Biamonte}, \citenamefont {Hadfield},\ and\ \citenamefont {Venturelli}}]{Maciejewski_2025}%
  \BibitemOpen
  \bibfield  {author} {\bibinfo {author} {\bibfnamefont {F.~B.}\ \bibnamefont {Maciejewski}}, \bibinfo {author} {\bibfnamefont {J.}~\bibnamefont {Biamonte}}, \bibinfo {author} {\bibfnamefont {S.}~\bibnamefont {Hadfield}},\ and\ \bibinfo {author} {\bibfnamefont {D.}~\bibnamefont {Venturelli}},\ }\href {https://doi.org/10.22331/q-2025-11-06-1906} {\bibfield  {journal} {\bibinfo  {journal} {Quantum}\ }\textbf {\bibinfo {volume} {9}},\ \bibinfo {pages} {1906} (\bibinfo {year} {2025})}\BibitemShut {NoStop}%
\bibitem [{\citenamefont {Hadfield}\ \emph {et~al.}(2026)\citenamefont {Hadfield}, \citenamefont {Maciejewski},\ and\ \citenamefont {Venturelli}}]{hadfield2026ndar}%
  \BibitemOpen
  \bibfield  {author} {\bibinfo {author} {\bibfnamefont {S.}~\bibnamefont {Hadfield}}, \bibinfo {author} {\bibfnamefont {F.~B.}\ \bibnamefont {Maciejewski}},\ and\ \bibinfo {author} {\bibfnamefont {D.}~\bibnamefont {Venturelli}},\ }\href@noop {} {\bibinfo {title} {Noise-directed adaptive remapping for integer optimization: From qubits to (encoded) qudits}} (\bibinfo {year} {2026}),\ \Eprint {https://arxiv.org/abs/2606.28234} {arXiv:2606.28234 [quant-ph]} \BibitemShut {NoStop}%
\bibitem [{\citenamefont {Maciejewski}\ \emph {et~al.}(2026)\citenamefont {Maciejewski}, \citenamefont {Hadfield}, \citenamefont {Wallis}, \citenamefont {Pennington}, \citenamefont {Brandhofer}, \citenamefont {Woerner}, \citenamefont {Egger},\ and\ \citenamefont {Venturelli}}]{Maciejewski2026}%
  \BibitemOpen
  \bibfield  {author} {\bibinfo {author} {\bibfnamefont {F.}~\bibnamefont {Maciejewski}}, \bibinfo {author} {\bibfnamefont {S.}~\bibnamefont {Hadfield}}, \bibinfo {author} {\bibfnamefont {O.}~\bibnamefont {Wallis}}, \bibinfo {author} {\bibfnamefont {G.}~\bibnamefont {Pennington}}, \bibinfo {author} {\bibfnamefont {S.}~\bibnamefont {Brandhofer}}, \bibinfo {author} {\bibfnamefont {S.}~\bibnamefont {Woerner}}, \bibinfo {author} {\bibfnamefont {D.~J.}\ \bibnamefont {Egger}},\ and\ \bibinfo {author} {\bibfnamefont {D.}~\bibnamefont {Venturelli}},\ }\href@noop {} {\bibinfo {title} {Quantum approximate optimization via noise-directed adaptive warm-starting}} (\bibinfo {year} {2026}),\ \Eprint {https://arxiv.org/abs/2607.09368} {arXiv:2607.09368} \BibitemShut {NoStop}%
\bibitem [{\citenamefont {Chi}\ \emph {et~al.}(2022)\citenamefont {Chi}, \citenamefont {Huang}, \citenamefont {Zhang}, \citenamefont {Mao}, \citenamefont {Zhou}, \citenamefont {Chen}, \citenamefont {Zhai}, \citenamefont {Bao}, \citenamefont {Dai}, \citenamefont {Yuan} \emph {et~al.}}]{chi2022programmable}%
  \BibitemOpen
  \bibfield  {author} {\bibinfo {author} {\bibfnamefont {Y.}~\bibnamefont {Chi}}, \bibinfo {author} {\bibfnamefont {J.}~\bibnamefont {Huang}}, \bibinfo {author} {\bibfnamefont {Z.}~\bibnamefont {Zhang}}, \bibinfo {author} {\bibfnamefont {J.}~\bibnamefont {Mao}}, \bibinfo {author} {\bibfnamefont {Z.}~\bibnamefont {Zhou}}, \bibinfo {author} {\bibfnamefont {X.}~\bibnamefont {Chen}}, \bibinfo {author} {\bibfnamefont {C.}~\bibnamefont {Zhai}}, \bibinfo {author} {\bibfnamefont {J.}~\bibnamefont {Bao}}, \bibinfo {author} {\bibfnamefont {T.}~\bibnamefont {Dai}}, \bibinfo {author} {\bibfnamefont {H.}~\bibnamefont {Yuan}}, \emph {et~al.},\ }\href {https://doi.org/10.1038/s41467-022-28767-x} {\bibfield  {journal} {\bibinfo  {journal} {Nat. Commun.}\ }\textbf {\bibinfo {volume} {13}},\ \bibinfo {pages} {1166} (\bibinfo {year} {2022})}\BibitemShut {NoStop}%
\bibitem [{\citenamefont {Ringbauer}\ \emph {et~al.}(2022)\citenamefont {Ringbauer}, \citenamefont {Meth}, \citenamefont {Postler}, \citenamefont {Stricker}, \citenamefont {Pogorelov}, \citenamefont {Lanyon}, \citenamefont {Blatt},\ and\ \citenamefont {Monz}}]{Ringbauer2022UniversalQudit}%
  \BibitemOpen
  \bibfield  {author} {\bibinfo {author} {\bibfnamefont {M.}~\bibnamefont {Ringbauer}}, \bibinfo {author} {\bibfnamefont {M.}~\bibnamefont {Meth}}, \bibinfo {author} {\bibfnamefont {L.}~\bibnamefont {Postler}}, \bibinfo {author} {\bibfnamefont {R.}~\bibnamefont {Stricker}}, \bibinfo {author} {\bibfnamefont {I.}~\bibnamefont {Pogorelov}}, \bibinfo {author} {\bibfnamefont {B.~P.}\ \bibnamefont {Lanyon}}, \bibinfo {author} {\bibfnamefont {R.}~\bibnamefont {Blatt}},\ and\ \bibinfo {author} {\bibfnamefont {T.}~\bibnamefont {Monz}},\ }\href {https://doi.org/10.1038/s41567-022-01640-6} {\bibfield  {journal} {\bibinfo  {journal} {Nat. Phys.}\ }\textbf {\bibinfo {volume} {18}},\ \bibinfo {pages} {1053} (\bibinfo {year} {2022})}\BibitemShut {NoStop}%
\bibitem [{\citenamefont {Nguyen}\ \emph {et~al.}(2024)\citenamefont {Nguyen}, \citenamefont {Goss}, \citenamefont {Siva}, \citenamefont {Kim}, \citenamefont {Younis}, \citenamefont {Qing}, \citenamefont {Hashim}, \citenamefont {Santiago},\ and\ \citenamefont {Siddiqi}}]{nguyen2024empowering}%
  \BibitemOpen
  \bibfield  {author} {\bibinfo {author} {\bibfnamefont {L.~B.}\ \bibnamefont {Nguyen}}, \bibinfo {author} {\bibfnamefont {N.}~\bibnamefont {Goss}}, \bibinfo {author} {\bibfnamefont {K.}~\bibnamefont {Siva}}, \bibinfo {author} {\bibfnamefont {Y.}~\bibnamefont {Kim}}, \bibinfo {author} {\bibfnamefont {E.}~\bibnamefont {Younis}}, \bibinfo {author} {\bibfnamefont {B.}~\bibnamefont {Qing}}, \bibinfo {author} {\bibfnamefont {A.}~\bibnamefont {Hashim}}, \bibinfo {author} {\bibfnamefont {D.~I.}\ \bibnamefont {Santiago}},\ and\ \bibinfo {author} {\bibfnamefont {I.}~\bibnamefont {Siddiqi}},\ }\href {https://doi.org/10.1038/s41467-024-51434-2} {\bibfield  {journal} {\bibinfo  {journal} {Nat. Commun.}\ }\textbf {\bibinfo {volume} {15}},\ \bibinfo {pages} {7117} (\bibinfo {year} {2024})}\BibitemShut {NoStop}%
\bibitem [{\citenamefont {Kim}\ \emph {et~al.}(2025)\citenamefont {Kim}, \citenamefont {Roy}, \citenamefont {You}, \citenamefont {Li}, \citenamefont {Lamm}, \citenamefont {Pronitchev}, \citenamefont {Bal}, \citenamefont {Garattoni}, \citenamefont {Crisa}, \citenamefont {Bafia} \emph {et~al.}}]{kim2025ultracoherent}%
  \BibitemOpen
  \bibfield  {author} {\bibinfo {author} {\bibfnamefont {T.}~\bibnamefont {Kim}}, \bibinfo {author} {\bibfnamefont {T.}~\bibnamefont {Roy}}, \bibinfo {author} {\bibfnamefont {X.}~\bibnamefont {You}}, \bibinfo {author} {\bibfnamefont {A.~C.}\ \bibnamefont {Li}}, \bibinfo {author} {\bibfnamefont {H.}~\bibnamefont {Lamm}}, \bibinfo {author} {\bibfnamefont {O.}~\bibnamefont {Pronitchev}}, \bibinfo {author} {\bibfnamefont {M.}~\bibnamefont {Bal}}, \bibinfo {author} {\bibfnamefont {S.}~\bibnamefont {Garattoni}}, \bibinfo {author} {\bibfnamefont {F.}~\bibnamefont {Crisa}}, \bibinfo {author} {\bibfnamefont {D.}~\bibnamefont {Bafia}}, \emph {et~al.},\ }\href@noop {} {\bibinfo {title} {Ultracoherent superconducting cavity-based multiqudit platform with error-resilient control}} (\bibinfo {year} {2025}),\ \Eprint {https://arxiv.org/abs/2506.03286} {arXiv:2506.03286} \BibitemShut {NoStop}%
\bibitem [{\citenamefont {Venturelli}\ \emph {et~al.}(2025)\citenamefont {Venturelli}, \citenamefont {Gustafson}, \citenamefont {Kurkcuoglu},\ and\ \citenamefont {Zorzetti}}]{venturelli2025near}%
  \BibitemOpen
  \bibfield  {author} {\bibinfo {author} {\bibfnamefont {D.}~\bibnamefont {Venturelli}}, \bibinfo {author} {\bibfnamefont {E.}~\bibnamefont {Gustafson}}, \bibinfo {author} {\bibfnamefont {D.}~\bibnamefont {Kurkcuoglu}},\ and\ \bibinfo {author} {\bibfnamefont {S.}~\bibnamefont {Zorzetti}},\ }\href@noop {} {\bibinfo {title} {Near-term application engineering challenges in emerging superconducting qudit processors}} (\bibinfo {year} {2025}),\ \Eprint {https://arxiv.org/abs/2506.05608} {arXiv:2506.05608} \BibitemShut {NoStop}%
\bibitem [{\citenamefont {Achlioptas}\ and\ \citenamefont {Coja-Oghlan}(2008)}]{AchlioptasCojaOghlan2008}%
  \BibitemOpen
  \bibfield  {author} {\bibinfo {author} {\bibfnamefont {D.}~\bibnamefont {Achlioptas}}\ and\ \bibinfo {author} {\bibfnamefont {A.}~\bibnamefont {Coja-Oghlan}},\ }in\ \href {https://doi.org/10.1109/FOCS.2008.11} {\emph {\bibinfo {booktitle} {Proceedings of the 49th Annual IEEE Symposium on Foundations of Computer Science}}}\ (\bibinfo  {publisher} {IEEE Computer Society},\ \bibinfo {year} {2008})\ pp.\ \bibinfo {pages} {793--802}\BibitemShut {NoStop}%
\bibitem [{\citenamefont {Coja-Oghlan}\ and\ \citenamefont {Vilenchik}(2016)}]{CojaOghlanVilenchik2016}%
  \BibitemOpen
  \bibfield  {author} {\bibinfo {author} {\bibfnamefont {A.}~\bibnamefont {Coja-Oghlan}}\ and\ \bibinfo {author} {\bibfnamefont {D.}~\bibnamefont {Vilenchik}},\ }\href {https://doi.org/10.1093/imrn/rnv333} {\bibfield  {journal} {\bibinfo  {journal} {Int. Math. Res. Not.}\ }\textbf {\bibinfo {volume} {2016}},\ \bibinfo {pages} {5801} (\bibinfo {year} {2016})}\BibitemShut {NoStop}%
\bibitem [{\citenamefont {Cosmadakis}\ and\ \citenamefont {Papadimitriou}(1984)}]{CosmadakisPapadimitriou1984}%
  \BibitemOpen
  \bibfield  {author} {\bibinfo {author} {\bibfnamefont {S.~S.}\ \bibnamefont {Cosmadakis}}\ and\ \bibinfo {author} {\bibfnamefont {C.~H.}\ \bibnamefont {Papadimitriou}},\ }\href {https://doi.org/10.1137/0213007} {\bibfield  {journal} {\bibinfo  {journal} {SIAM J. Comput.}\ }\textbf {\bibinfo {volume} {13}},\ \bibinfo {pages} {99} (\bibinfo {year} {1984})}\BibitemShut {NoStop}%
\bibitem [{\citenamefont {Mnich}\ and\ \citenamefont {Wiese}(2015)}]{MnichWiese2015}%
  \BibitemOpen
  \bibfield  {author} {\bibinfo {author} {\bibfnamefont {M.}~\bibnamefont {Mnich}}\ and\ \bibinfo {author} {\bibfnamefont {A.}~\bibnamefont {Wiese}},\ }\href {https://doi.org/10.1007/s10107-014-0830-9} {\bibfield  {journal} {\bibinfo  {journal} {Math. Program.}\ }\textbf {\bibinfo {volume} {154}},\ \bibinfo {pages} {533} (\bibinfo {year} {2015})}\BibitemShut {NoStop}%
\bibitem [{\citenamefont {Knop}\ and\ \citenamefont {Kouteck{\'y}}(2018)}]{KnopKoutecky2018}%
  \BibitemOpen
  \bibfield  {author} {\bibinfo {author} {\bibfnamefont {D.}~\bibnamefont {Knop}}\ and\ \bibinfo {author} {\bibfnamefont {M.}~\bibnamefont {Kouteck{\'y}}},\ }\href {https://doi.org/10.1007/s10951-017-0550-0} {\bibfield  {journal} {\bibinfo  {journal} {J. Sched.}\ }\textbf {\bibinfo {volume} {21}},\ \bibinfo {pages} {493} (\bibinfo {year} {2018})}\BibitemShut {NoStop}%
\bibitem [{\citenamefont {Bose}\ and\ \citenamefont {Mesner}(1959)}]{BoseMesner1959}%
  \BibitemOpen
  \bibfield  {author} {\bibinfo {author} {\bibfnamefont {R.~C.}\ \bibnamefont {Bose}}\ and\ \bibinfo {author} {\bibfnamefont {D.~M.}\ \bibnamefont {Mesner}},\ }\href {https://doi.org/10.1214/aoms/1177706356} {\bibfield  {journal} {\bibinfo  {journal} {Ann. Math. Stat.}\ }\textbf {\bibinfo {volume} {30}},\ \bibinfo {pages} {21} (\bibinfo {year} {1959})}\BibitemShut {NoStop}%
\end{thebibliography}%
\end{document}